\documentclass[11pt]{article}
\usepackage{graphicx,latexsym,amsmath,amsthm,amssymb,color,verbatim,mathrsfs,bbm,amscd,placeins,hyperref}
\usepackage[font=footnotesize]{caption}
\usepackage{subcaption}
\usepackage[utf8] {inputenc}
\usepackage{geometry}
\usepackage{tikz}
\usetikzlibrary{shapes,arrows,fit,calc,positioning,automata,decorations.markings}
\usepackage[vcentermath]{youngtab}
\usepackage{ytableau}

\swapnumbers
\numberwithin{equation}{section}
\newtheorem{theorem}[equation]{Theorem}
\newtheorem*{theorem*}{Theorem (informal)}
\newtheorem*{theorem*thm:evalabp}{\ref{thm:evalabp}~Theorem (informal)}
\newtheorem*{theorem*thm:oursaxenaduality}{\ref{thm:oursaxenaduality}~Theorem (informal)}
\newtheorem*{theorem*thm:evaltreewidth}{\ref{thm:evaltreewidth}~Theorem (informal)}
\newtheorem*{theorem*thm:evaluationhardness}{\ref{thm:evaluationhardness}~Theorem (informal)}
\newtheorem*{theorem*thm:evaluationsharpphardness}{\ref{thm:evaluationsharpphardness}~Theorem (informal)}
\newtheorem*{theorem*thm:semistdhardness}{\ref{thm:semistdhardness}~Theorem (informal)}
\newtheorem{corollary}[equation]{Corollary}
\newtheorem{lemma}[equation]{Lemma}

\newtheorem{proposition}[equation]{Proposition}

\newtheorem{question}[equation]{Question}

\theoremstyle{definition}
\newtheorem{remark}[equation]{Remark}
\newtheorem{definition}[equation]{Definition}

\title{On the complexity of evaluating highest weight vectors}

\author{Markus Bläser\footnote{mblaeser@cs.uni-saarland.de (Saarland University, Germany)}, Julian Dörfler\footnote{jdoerfler@cs.uni-saarland.de (Saarbrücken Graduate School of Computer Science, Germany)}, Christian Ikenmeyer\footnote{christian.ikenmeyer$@$liverpool.ac.uk (University of Liverpool, United Kingdom. Part of this research was done when CI was at the Max Planck Institute for Software Systems, Germany, and the Simons Institute for the Theory of Computing, United States. CI was supported by DFG grant IK 116/2-1.)}}

\usepackage{tikz}
\usetikzlibrary{shapes,shapes.geometric,arrows,fit,calc,positioning,automata,decorations.markings}
\usepackage[vcentermath]{youngtab}
\usepackage{ytableau}
\usepackage{esvect}

\newcommand{\Sym}{\textup{Sym}}
\newcommand{\GL}{\mathsf{GL}}
\newcommand{\HWV}{\mathsf{HWV}}
\newcommand{\mult}{\mathsf{mult}}
\newcommand{\IC}{\mathbb{C}}
\newcommand{\IN}{\mathbb{N}}
\newcommand{\IR}{\mathbb{R}}
\newcommand{\CT}{\mathcal{T}}
\newcommand{\CF}{\mathcal{F}}
\newcommand{\CW}{\mathcal{W}}
\newcommand{\aS}{\mathfrak{S}}
\newcommand{\ncw}{\textup{\textsf{ncw}}}
\newcommand{\la}{\lambda}
\newcommand{\per}{\text{per}}

\newcommand{\PP}{\ensuremath{\mathsf{P}}}
\newcommand{\NP}{\ensuremath{\mathsf{NP}}}

\newcommand{\VBP}{\ensuremath{\mathsf{VBP}}}
\newcommand{\VWaring}{\ensuremath{\mathsf{VWaring}}}
\newcommand{\ETH}{\ensuremath{\textup{ETH}}}
\newcommand{\Ch}{\textup{\textsf{Ch}}}
\newcommand{\SAT}{\ensuremath{\mathsf{3SAT}}}

\newcommand{\Thor}{{\ensuremath{\hat{T}_{\leftrightarrow}}}}
\newcommand{\Tvert}{{\ensuremath{\hat{T}_{\updownarrow}}}}
\newcommand{\Ehor}{{\ensuremath{E_{\leftrightarrow}}}}
\newcommand{\Evert}{{\ensuremath{E_{\updownarrow}}}}

\newcommand{\WR}{{\mathsf{WR}}}
\newcommand{\wW}{{\mathsf{w}}}
\newcommand{\WRbar}{\underline{\mathsf{WR}}}

\DeclareMathOperator{\rank}{rank}
\DeclareMathOperator{\diag}{diag}
\DeclareMathOperator{\poly}{poly}
\DeclareMathOperator{\spn}{span}

\begin{document}
\sloppy

\maketitle
\thispagestyle{empty}\setcounter{page}{0}
\begin{abstract}
Geometric complexity theory (GCT) is an approach towards separating algebraic complexity classes through algebraic geometry and representation theory.
Originally Mulmuley and Sohoni proposed (SIAM J Comput 2001, 2008) to use occurrence obstructions to prove Valiant's determinant vs permament conjecture, but recently B\"urgisser, Ikenmeyer, and Panova (Journal of the AMS 2019) proved this impossible.
However, fundamental theorems of algebraic geometry and representation theory grant that every lower bound in GCT can be proved by the use of so-called highest weight vectors (HWVs).
In the setting of interest in GCT (namely in the setting of polynomials) we prove the NP-hardness of the evaluation of HWVs in general, and we give efficient algorithms if the treewidth of the corresponding Young-tableau is small, where the point of evaluation is concisely encoded as a noncommutative algebraic branching program! In particular, this gives a large new class of separating functions that can be efficiently evaluated at points with low (border) Waring rank.
As a structural side result we prove that border Waring rank is bounded from above by the ABP width complexity.
\end{abstract}

\medskip

\noindent\footnotesize
\begin{minipage}{2cm}
\end{minipage}
\begin{minipage}{13cm}
Algebraic complexity theory, geometric complexity theory, algebraic branching program, Waring rank, border Waring rank, representation theory, highest weight vector, treewidth
\end{minipage}

\newpage
\pagestyle{plain}

\section{Introduction}
\label{sec:intro}
Geometric complexity theory (GCT) is an approach towards the separation of algebraic complexity classes using algebraic geometry and representation theory \cite{gct1, gct2, BLMW:11}.
Let $\per_i := \sum_{\pi\in\aS_i} \prod_{j=1}^i x_{j,\pi(j)}$ be the permanent polynomial. Valiant 
asked for the smallest size of a matrix $A$ whose entries are affine linear polynomials such that $\det(A)=\per_i$ and his famous $\text{VBP}\neq\text{VNP}$ conjecture (also known as the ``determinant vs permanent conjecture'') states that this size is not polynomially bounded.
Mulmuley and Sohoni strengthened the conjecture by allowing $\per_i$ to be approximated arbitrarily closely, i.e., $\text{VNP}\not\subseteq\overline{\text{VBP}}$.
This question can be attacked with GCT.

In the GCT approach, we set $m:=d^2$ and let the group $\GL_{m}:=\GL(\IC^m)$ act on a the space of homogeneous degree $d$ polynomials in $m$ variables by linear transformation of the variables. The Mulmuley--Sohoni conjecture can
be rephrased as ``eventually $x_{11}^{d-i} \per_i \notin \overline{\GL_m \det_d}$'' if $d$ grows polynomially in~$i$. Now we try to attack this problem by representation theoretic methods, so-called \emph{obstructions}.
A first crucial insight is that $x_{11}^{d-i} \per_i \in \overline{\GL_m \det_d}$
iff
$\overline{\GL_m (x_{11}^{d-i} \per_i)} \subseteq \overline{\GL_m \det_d}$. Thus, we compare two varieties and we want to disprove that
the orbit closure of the padded permanent is contained in the orbit 
closure of the determinant for polynomially large $d$.
To to so, an important object to study are so-called \emph{highest weight vectors} (HWVs) of weight $\la \in \IN^m$.
They are homogeneous degree $n$ polynomials in the coefficients of homogeneous degree $d$ polynomials in $m$ variables, satisfying two properties (see Sec.~\ref{sec:prel}). Their dimension is called the plethysm coefficient. The dimension of their restriction to a $\GL_m$-variety $X$ is called the \emph{multiplicity $\mult_\la \IC[X]$ of $\la$ in the coordinate ring $\IC[X]$}. They are important, because if $\mult_\la \IC[X] > \mult_\la \IC[Y]$, then Schur's lemma implies that $X \not\subseteq Y$. In this case, $\la$ is called a \emph{multiplicity obstruction}.
If additionally $\mult_\la \IC[X] > 0 = \mult_\la \IC[Y]$, then $\la$ is called an \emph{occurrence obstruction}.
Even more fundamentally, the properties of the representation theory of $\GL_m$ imply that if $x_{11}^{d-i} \per_i \notin \overline{\GL_m \det_d}$, then there exists a HWV $f$ such that
$f(\overline{\GL_m \det_d})=\{0\}$
and for a random $g \in \GL_m$ we have $f(g(x_{11}^{d-i} \per_i)) \neq 0$.
So this separation is always provable by HWVs.
This follows from the fact that HWVs uniquely classify the irreducible representations of $\GL_m$.

B\"urgisser et al.~\cite{BIP:19} proved that occurrence obstructions are not sufficient to prove Mulmuley and Sohoni's conjecture.
Hence, multiplicity obstructions are a focus of recent research \cite{DBLP:conf/icalp/DorflerIP19, IK:19}. To compute multiplicities,
it is import to understand the complexity of the evaluation of highest weight vectors.  

To calculate a multiplicity $\mult_\la \IC[X]$, a common approach is to generate a basis of all HWVs of weight $\la$ and evaluate them at enough points from $X$ (points from all $\GL_m$-varieties in GCT are efficiently samplable) and observe the dimension of their linear span, which equals $\mult_\la \IC[X]$. For this to work, one needs an algorithm to evaluate HWVs at points.
An evaluation algorithm is even more important to make the following approach work:
We know that if $X \not\subseteq Y$, then there exists a HWV $f$ of some weight $\la$ such that $f(Y)=\{0\}$ and $f(x) \neq 0$ for almost all points $x \in X$ \cite[Cor.~11.4.2]{BI:17}.
This evaluation is a challenging problem in algebraic geometry that is related to deep combinatorics, see~\cite{LK:15, cheung2017symmetrizing, AIR:16}.

\section{Our contributions}

To our best knowledge, we systematically study the complexity of evaluating highest weight vectors for the first time.
In Section~\ref{sec:base} we first present a known combinatorial method of exactly evaluating HWVs without expanding all the monomials explicitly which has been used to to evaluate HWVs at points of small Waring rank as in \cite{AIR:16, BI:17b}.
Additionally there have been attempts to improve the running time for evaluating at products of linear forms -- the so called Chow variety -- via dynamic programming \cite{DBLP:conf/icalp/DorflerIP19}.
We generalize both approaches in Section~\ref{sec:abp} to allow evaluation on all points with partial derivative spaces of small dimension, i.e., small noncommutative algebraic branching program width complexity.
\begin{theorem*thm:evalabp}
The evaluation of a degree $n$ highest weight vector $f_{\hat T}$ (given by a Young tableau ${\hat T}$ with $r$ rows) at a homogeneous degree $d$ polynomial $p$ in $m$ variables whose noncommutative ABP width complexity is at most $w$ can be computed in time $O(w^{n + r}\poly(n,d,m))$.
\end{theorem*thm:evalabp}
In particular, by Theorem~\ref{thm:oursaxenaduality} this includes for the first time all points of small border Waring rank:
\begin{theorem*thm:oursaxenaduality}
For all polynomials $f$ the noncommutative ABP width of $f$ is less or equal to the border Waring rank of $f$. This also holds for commutative ABP width complexity.
\end{theorem*thm:oursaxenaduality}
Theorem~\ref{thm:oursaxenaduality} is proved using the noncommutative algebraic branching program width complexity as a tool, which shows that it is not just a notion useful for algorithmic purposes, but a natural notion of independent interest.
Note that our algorithms are particularly useful, because the noncommutative algebraic branching program width complexity can be determined in polynomial time, whereas determining the Waring rank of a polynomial is NP-hard, even when it is given explicitly as a list of coefficients, see \cite{Shi:16}.

A HWV can be encoded as a linear combination of Young tableaux, see e.g.~\cite[\S3.9]{Ott:13} or \cite[Sec.~4.3]{ike:12b}.
All current evaluation algorithms have a running time exponentially dependent on the size of the Young tableau.
We improve this in Section~\ref{sec:treewidth} and establish an algorithm that only depends exponentially on the treewidth of the Young tableau:
\begin{theorem*thm:evaltreewidth}
The evaluation of a degree $n$ highest weight vector $f_{\hat T}$ given by a Young tableau $\hat{T}$ 
at a homogeneous degree $d$ polynomial $p$ in $m$ variables with noncommutative ABP width complexity $w$ can be computed in time $w^{\omega(\tau+1)}\poly(n,d,m,|\CT|)$,
where $\CT$ is a tree decomposition of $\hat T$ of width $\tau$ and size $|\CT|$ and $\omega$ is the matrix multiplication exponent.
\end{theorem*thm:evaltreewidth}

Our paper is the first that formally connects the running time of algorithms in representation theory with a graph parameter. An implementation of the algorithm in Theorem~\ref{thm:evaltreewidth} might make it possible to compute the multiplicities for examples that were out of reach before, which is potentially useful for implementing the geometric complexity theory approach.

Lastly we show in Section~\ref{sec:hardness} that this dependency is basically optimal as we show two lower bounds under the exponential time hypothesis.
A lower bound of $2^{o(n)}$ for the vanishing evaluation decision problem when the HWV $f \in \Sym^n\Sym^d V$ is given by an arbitrary two row Young tableau and a lower bound of $2^{o(\sqrt{n})}$ when it is given by a semistandard Young tableau.
Additionally we show $\NP$-hardness for both versions of the decision problem and even $\#\PP$-hardness for exact evaluations.

\begin{theorem*thm:evaluationhardness}[HWVs from two-row tableaux]
    Deciding whether a degree $n$ highest weight vector $f_{\hat{T}}$ (given by a two-row Young tableau $\hat{T}$) evaluates to zero at a point of constant degree at least 8 and of Waring rank $3$ is $\NP$-hard.
    Assuming $\ETH$ no $2^{o(n)}$ algorithm for this evaluation can exist.
\end{theorem*thm:evaluationhardness}

\begin{theorem*thm:semistdhardness}[HWVs from semistandard tableaux]
    Deciding whether or not the evaluation of a degree $n$ highest weight vector $f_{\hat{T}}$ (given by a 5-row semistandard Young tableau $\hat{T}$) vanishes at a point of constant degree $d \geq 16$ with $16~|~d$ and of Waring rank 5 is $\NP$-hard.
    Additionally this evaluation can not be computed in time $2^{o(\sqrt{n})}$ unless ETH fails.
\end{theorem*thm:semistdhardness}

\begin{theorem*thm:evaluationsharpphardness}[$\#\PP$-hardness]
    Evaluating a highest weight vector $f_{\hat{T}}$ (given by a two-row Young tableau $\hat{T}$) at a point of Waring rank $3$ and degree $d \geq 18$ is $\#\PP$-hard.
\end{theorem*thm:evaluationsharpphardness}
We remark that it is quite surprising that these results can be obtained using points of small constant Waring rank.

\section{Related work}
The approach to lower bounds via evaluating HWVs was used in \cite{BI:11, BI:13} in the tensor setting to obtain lower bounds on the border rank of matrix multiplication. This also led to multiplicity obstructions (even occurrence obstructions). Our complexity results can be interpreted as limitations on how far such an approach via explicit evaluations can be pushed.

Combinatorics on tableaux for describing highest weight vectors has a rich history dating back to the early invariant theory.
This tableau calculus is equivalent to the classical \emph{Feynman diagram calculus} explained in \cite{Abd:02}, see also \cite{Ott:13}. Highest weight vectors of a $\GL_m$-representation $W$ are also called \emph{covariants}, since they correspond to the invariants of $W \otimes (S_\la \IC^{m})^*$, see e.g.~\cite[Def.~3.9]{OR:11}.
Recently, these methods have been applied in various areas, see \cite{Kum:11, BO:11, Rai:13, DHO:14, MM:14, AIR:16,BI:17b, CHILO:18}, to name a few.
If we restrict ourselves to two-row Young diagrams, then inheritance principles from representation theory \cite[Sec.~5.3]{ike:12b} let us replace $V$ with $\IC^2$. Then $\Sym^d \IC^2$ is the
Hilbert space corresponding to a system of $d$ indistinguishable photons distributed among two modes, which is used in the study of 2-mode linear optical circuits on $d$ indistinguishable particles.

Waring rank and border Waring rank are classical notions studied in algebraic geometry in the language of higher secant varieties \cite{Lan:15}.
More generally, border complexity is classically studied in algebraic geometry, see \cite{Lan:11}. Bini et al~\cite{bcrl:79} (see also \cite{Bini:80}) used it in their construction of fast matrix multiplication algorithms.
Studying border complexity in algebraic circuit complexity started with \cite{bue:01, gct1} and recently caught momentum \cite{GMQ:16, BIZ:18, Kumar2018}.

Kronecker coefficients and plethysm coefficients are the dimensions of specific highest weight vector spaces.
Algorithms for their computation or theorems about their positivity and value that depend heavily on the shape of the input Young tableau have a long history.
For example, if the number of rows of all parameters is constant, then the Kronecker coefficient can be computed in polynomial time \cite{CDW:12}.
A similar statement is true for plethysm coefficients, see \cite{FI:20}.
The software \texttt{LiE} \cite{LiE} performs all representation theoretic computations with a fixed number of rows.
In \cite{Ike:15}, positivity of Kronecker coefficients depends on comparing Young diagrams with respect to the dominance order, and in \cite{BB:04} the main parameter is the so-called \emph{Durfee size} of the Young diagram, which is the side length of largest square that can be embedded into the Young diagram, see also the very recent \cite{BBP:20}.
The shape of the Young diagram also plays a crucial role in the recent breakthrough proof of Stembridge's stability conjecture \cite{SS:16}.
For two-row Young diagrams much additional structure is known, for example Hermite's classical reciprocity law for plethysm coefficients \cite{Her:54}, which makes our lower bound for two-row Young tableaux quite surprising.

Treewidth has been intensely studied by Robertson and Seymour and has been applied numerous times to construct faster graph algorithms for cases where the treewidth is bounded by a function $o(n)$, most notably some algorithms for $\NP$-hard problems restricted to planar graphs, for example $3$-coloring.
See \cite{cygan2015parameterized} for an introduction to treewidth algorithms.

\section{Border Waring rank and Algebraic Branching Programs}\label{sec:prelude}
In this section we introduce noncommutative ABP width complexity for polynomials and use it to prove Theorem~\ref{thm:oursaxenaduality}.
Noncommutative ABP width complexity will play a central role in Sections~\ref{sec:abp} and~\ref{sec:treewidth}.

An algebraic branching program (ABP) is a layered directed acyclic graph (the vertex set is partitioned into numbered layers and edges only go from the $i$-th layer to the $(i+1)$-th layer) with two distinguished nodes, the \emph{source} and the \emph{sink}, and the edges are labeled with homogeneous linear polynomials.
The weight $w(P)$ of a path $P$ with edge labels $\ell_1, \ldots, \ell_d$ is defined as the product $w(P) := \ell_1 \cdots \ell_d.$
We say that the ABP \emph{computes} the sum $\sum_{\text{source-sink-path}\ P} w(P)$.
We can view the same ABP both over commuting variables or noncommuting variables.
If we interpret it over noncommuting variables, we call it an ncABP. If we want to stress that the variables commute, we call it a cABP.
The size of an ABP is the number of its vertices. The width of an ABP is the largest number of vertices in any layer.
For a homogeneous degree $d$ polynomial let the \emph{ABP width complexity} $\wW(f)$ be defined as the smallest width of a cABP computing~$f$. 
A sequence $(f_n)$ of polynomials is called a \emph{p-family} if the number of variables
and the degree of each $f_n$ are polynomially bounded in $n$. p-families are the object of study
in Valiant's algebraic complexity framework.
Let $\VBP$ denote the set of all p-families $(f_i)$ with polynomially bounded ABP width complexity $\wW(f_i)$.

The Waring rank $\WR(f)$ of a homogeneous degree $d$ polynomial is the smallest $r$ such that $f$ can be written as a sum of $r$ powers of homogeneous linear polynomials.
Let $\VWaring$ be the set of all p-families with polynomially bounded Waring rank.

Clearly, $\wW(f) \leq \WR(f)$, because from a Waring rank $r$ decomposition we can construct a width $r$ cABP that computes $f$ in the straightforward way: The cABP contains exactly $r$ disjoint source-sink-paths (vertex-disjoint up to source and sink) so that on each path all edges have the same label. Therefore $\VWaring \subseteq \VBP$.

There is a natural way to associate to every algebraic complexity measure a corresponding border complexity measure:
We define the \emph{border Waring rank} $\underline{\WR}(f)$ as the smallest $r$ such that $f$ can be approximated arbitrarily closely (coefficient-wise) by polynomials with $\WR(f)\leq r$, or equivalently, the smallest $r$ such that $f$ lies in the closure (Zariski closure and Euclidean closure coincide) of the set $\{f \mid \WR(f) \leq r\}$.
Clearly $\underline{\WR}(f)\leq \WR(f)$.
Let $\overline{\VWaring}$ denote the set of sequences of polynomials with polynomially bounded border Waring rank. Clearly $\VWaring \subseteq \overline{\VWaring}$.

Analogously we can define the \emph{border ABP width complexity} $\underline{\wW}(f)$ from $\wW$. Clearly $\underline{\wW}(f)\leq \wW(f)$.
Let $\overline{\VBP}$ be the set of polynomials with polynomially bounded border ABP width complexity. Clearly $\VBP \subseteq \overline{\VBP}$.

For noncommutative polynomials we define the analogous versions $\ncw$ and $\underline{\ncw}$. It follows from Nisan's work \cite{nisan1991lower} that $\ncw(f)=\underline{\ncw}(f)$.

In general, it is unknown by how much an algebraic complexity class grows when applying the closure. In particular, it is open whether $\VWaring = \overline{\VWaring}$ or whether $\VBP = \overline{\VBP}$.
But the following result in this direction is known.

\begin{theorem}\label{thm:borderwaringinVBP}
$\overline{\VWaring} \subseteq \VBP$.
\end{theorem}
We quickly sketch the standard proof.
We will need the following concept only for this proof.
A \emph{read-once oblivious ABP} is a layered ABP whose edge labels have univariate polynomials in $x_i$ on each edge in layer $i$.
The first step in the proof is Saxena's duality trick \cite[Lemma~1]{Sax:08}:
\begin{equation*}
\begin{minipage}{13cm}
If $f \in \IC[x_1,\ldots,x_m]_d$ has $\WR(f)\leq s$, then there is a read-once oblivious ABP computing $f$ with width at most $s\cdot(md+d+1)$.
\end{minipage}
\end{equation*}
The proof uses a power series argument.
The next crucial step is to use a variant of Nisan's result \cite{nisan1991lower} to see that the \emph{border} read-once oblivious ABP width equals the read-once oblivious ABP width, so approximations can be removed \cite[Sec.~4.5.2]{For:14}:
\begin{equation*}
\begin{minipage}{13cm}
If $f \in \IC[x_1,\ldots,x_m]_d$ has $\WRbar(f)\leq s$, then there is an read-once oblivious ABP computing $f$ with width at most $s\cdot(md+d+1)$.
\end{minipage}
\end{equation*}
We can unfold this read-once oblivious ABP, i.e., replace each edge (remember, each label is a univariate degree $\leq d$ polynomial) with a (non-layered) ABP computing it, where each edge has an affine linear label. If done properly, this requires $d-1$ additional vertices per edge.
Making the ABP layered and homogeneous blows up the ABP's width by a factor of $d+1$.
We conclude:
\begin{equation}\label{thm:theirsaxenaduality}
\begin{minipage}{13cm}
For all $f \in \IC[x_1,\ldots,x_m]_d$ we have $\wW(f) \leq \WRbar(f)\cdot(md+d+1)\cdot(d+1)$.
\end{minipage}
\end{equation}

Eq.~\eqref{thm:theirsaxenaduality} proves Theorem~\ref{thm:borderwaringinVBP} when we assume that $m$ and $d$ are polynomially bounded (which is usually assumed).
We now strengthen eq.~\eqref{thm:theirsaxenaduality} with the following clean statement that is independent of $m$ and $d$.
\begin{theorem}\label{thm:oursaxenaduality}
For all $f \in \IC[x_1,\ldots,x_m]_d$ we have $\wW(f) \leq \WRbar(f)$.
\end{theorem}
In fact, we prove
$\wW(f) \leq \ncw(f) \leq \WRbar(f)$, but we have not yet defined what we mean by an ncABP computing a polynomial.
The rest of Section~\ref{sec:prelude} is devoted to the proof of Theorem~\ref{thm:oursaxenaduality} and to this definition. We start with introducing several main multilinear algebra concepts of this paper. The actual proof of Theorem~\ref{thm:oursaxenaduality} is then very short and natural.

When talking about homogeneous multivariate noncommutative polynomials, we use the standard language of multilinear algebra: An order $d$ \emph{tensor} in $\otimes^d \IC^m$ is a $d$-dimensional $m \times m \times \cdots \times m$ array of numbers.
There is a canonical vector space isomorphism between the vector space of $m$-variate homogeneous degree $d$ noncommutative polynomials $\IC\langle x_1,\ldots,x_m\rangle_d$ and $\otimes^d \IC^m$, which is defined on monomials as
\[
x_{i_1} x_{i_2} \cdots x_{i_d} \stackrel{\sim}{\longrightarrow} E_{i_1,\ldots,i_d},
\]
where $E_{i_1,\ldots,i_d}$ is the tensor that is 0 everywhere, but has a single 1 at position $(i_1,\ldots,i_d)$.
Let $(e_i)$ be the standard basis of $\IC^m$. We use the notation $e_{i_1}\otimes e_{i_2}\otimes \cdots \otimes e_{i_d} := E_{i_1,\ldots,i_d}$.
More generally, for $v_1,\ldots,v_d \in \IC^m$, we write $v_1 \otimes v_2 \otimes \cdots \otimes v_m$ to be the tensor whose entry at position $(i_1,\ldots,i_d)$ is the product $(v_1)_{i_1}\cdot (v_2)_{i_2} \cdots (v_d)_{i_d}$.

A tensor $T$ is called \emph{symmetric} if $T_{i_1,\ldots,i_d} = T_{i_{\pi(1)},\ldots,i_{\pi(d)}}$ for all permutations $\pi \in \aS_d$. Let $\Sym^d \IC^m \subseteq \otimes^d \IC^m$ denote the linear subspace of symmetric tensors.
There is a canonical vector space isomorphism between the vector space of $m$-variate homogeneous degree $d$ \emph{commutative} polynomials $\IC[x_1,\ldots,x_m]_d$ and $\Sym^d \IC^m$, which is defined on monomials as
\[
x_{i_1} x_{i_2} \cdots x_{i_d} \stackrel{\sim}{\longrightarrow} \sum_{\pi \in \aS_d} \tfrac{1}{d!} E_{\pi(i_1),\ldots,\pi(i_d)},
\]
For example, the polynomial $x_1^2 x_2$ corresponds to the tensor $\frac 1 3 (e_1 \otimes e_1 \otimes e_2 + e_1 \otimes e_2 \otimes e_1 + e_2 \otimes e_1 \otimes e_1)$. \ \footnote{This tensor is called the W-state in quantum information theory.}
We use $e_i$ and $x_i$ interchangeably.

It is crucial to note that \emph{noncommutative ABPs can compute symmetric tensors}.
An example is given in Figure~\ref{fig:ncabp}, where we used $x:=x_1$ and $y:=x_2$.
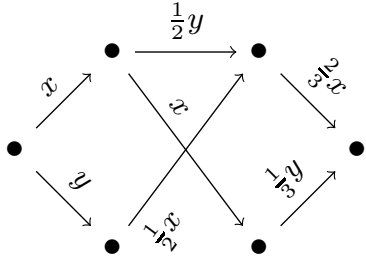
\begin{figure}
\centering
\scalebox{1.3}{%
\begin{tikzpicture}
\node (0) at (0,0) {$\bullet$};
\node (1up) at (1,1) {$\bullet$};
\node (1down) at (1,-1) {$\bullet$};
\node (2up) at (2.5,1) {$\bullet$};
\node (2down) at (2.5,-1) {$\bullet$};
\node (3) at (3.5,0) {$\bullet$};
\draw[->] (0) to[sloped, above] node {\footnotesize $x$} (1up);
\draw[->] (0) to[sloped, above] node {\footnotesize $y$} (1down);
\draw[->] (1up) to[sloped, above] node {\footnotesize $\tfrac 1 2 y$} (2up);
\draw[->] (1up) to[pos=0.3, sloped, above] node {\footnotesize $x$} (2down);
\draw[->] (1down) to[pos=0.1, sloped, below] node {\footnotesize $\tfrac 1 2 x$} (2up);
\draw[->] (2up) to[sloped, above] node {\footnotesize $\tfrac 2 3 x$} (3);
\draw[->] (2down) to[sloped, above] node {\footnotesize $\tfrac 1 3 y$} (3);
\end{tikzpicture}
}
\caption{An ncABP computing the symmetric tensor $\frac 1 3 (x\otimes x\otimes y + x\otimes y\otimes x + y\otimes x\otimes x)$, which corresponds to the polynomial $x^2 y$.
If we reinterpret the ncABP as a cABP, it computes $\frac 1 3 (xxy + xyx + yxx)= x^2 y$.
Such ncABPs can be efficiently constructed using Nisan's construction technique \cite{nisan1991lower}. Interestingly, in this example the width is only 2, while the Waring rank of $x^2 y$ is 3.}
\label{fig:ncabp}
\end{figure}
As before with cABPs, it is easy to see that every Waring rank $r$ decomposition of $f$ can be converted into a width $r$ ncABP computing $f$ in the straightforward way: The ncABP contains exactly $r$ disjoint source-sink-paths (vertex-disjoint up to source and sink) so that on each path all edges have the same label.
Every ncABP can be reinterpreted as a cABP by letting the variables commute. If the ncABP computes a symmetric tensor, then clearly this cABP computes the corresponding polynomial.
Now we can prove Theorem~\ref{thm:oursaxenaduality} in a very natural and short way as follows.

Given $f$ with a border Waring rank $s$ decomposition. We construct the corresponding border ncABP with $s$ many edge-disjoint source-sink-paths, so $\underline{\ncw}(f)\leq s$.
Using Nisan's result \cite{nisan1991lower} that $\underline{\ncw}=\ncw$, it follows $\ncw(f)\leq s$.
This gives a width $s$ ncABP that computes $f$.
Reinterpreting this ncABP as a cABP finishes our proof of Theorem~\ref{thm:oursaxenaduality}.

\section{Highest Weight Vectors and their combinatorial evaluation}
\label{sec:base}
\label{sec:prel}
Let $V = \IC^m$ be a finite dimensional complex vector space with standard basis $e_1, e_2, \ldots, e_m$.
There is a canonical action of $g \in \GL(V)$ on the tensor power $\otimes^d V$ via
$g(p_1 \otimes \dots \otimes p_d) := (gp_1) \otimes \dots \otimes (gp_{d})$ and linear continuation.
This action can be lifted to a linear action on $\Sym^n\otimes^d V$ via
\[
    (gf)(p) := f(g^{t}p)\;\text{for $f \in \Sym^n\otimes^d V$ and $p \in \otimes^d V$}
\]
Note that this makes $\Sym^n\otimes^d V$ a $\GL(V)$-representation.
We denote by $\Sym^d V \subseteq \otimes^d V$ the vector space of symmetric tensors over $V$ of order $d$ and by $p_1 \odot \dots \odot p_d := \sum_{\pi \in \aS_d}\frac{1}{d!} p_{\pi(1)} \otimes \dots \otimes p_{\pi(d)}$ the symmetric tensor product of $p_1, \ldots, p_d \in V$.
The linear subspace $\Sym^d V \subseteq \otimes^d V$ is closed under the action of $\GL(V)$.
This action can be lifted to a linear action on $\Sym^n\Sym^d V$ via
\[
    (gf)(p) = f(g^{t}p)\;\text{for $f \in \Sym^n\Sym^d V$ and $p \in \Sym^d V$}
\]
Note that this makes $\Sym^n\Sym^d V$ a $\GL(V)$-representation.

We call a sequence $\lambda = (\lambda_1, \lambda_2, \dots)$ a \emph{partition of $N \in \IN$} if $\lambda_1 \geq \lambda_2 \geq \lambda_3 \geq \ldots \geq 0$ and $\sum_{i \geq 1}\lambda_i = N$. In our case we will usually have $N=nd$.
We denote the transpose partition $\lambda^t$ by $\mu$ and define it as $\mu_i = |\{j \mid \lambda_j \geq i\}|$.
Note that $\mu$ is also a partition of $N$.
We will write partitions as finite sequences and omit all the trailing zeros.

For any $\GL_m$ representation $W$, a \emph{highest weight vector} $f \in W$ of type $\lambda$ is a vector that satisfies
\begin{enumerate}
    \item $f$ is invariant under the action of any $g \in \GL_m$ when $g$ is upper triangular with $1$s on the diagonal.
    \item $\diag(\alpha_1, \ldots, \alpha_m) f = \alpha_1^{\lambda_1} \cdot \cdots \cdot \alpha_m^{\lambda_m} f$ where $\diag(\alpha_1, \ldots, \alpha_m)$ is the diagonal matrix with $\alpha_1, \ldots, \alpha_m \in \IC$ on the diagonal.
\end{enumerate}
The highest weight vectors of type $\lambda$ form a vector space which we call $\HWV_\lambda(W)$.
We denote by $\HWV(W)$ the vector space of all HWVs in $W$ without any weight restriction.

The smallest example is the discriminant polynomial $b^2-4ac$ in $\Sym^2 \Sym^2 \IC^2$, see \cite[Exa.~9.1.4]{BI:17} for which we have $g(b^2-4ac) = \det(g)^2 (b^2-4ac)$.

We first derive a combinatorial description of the evaluation of highest weight vectors. We follow \cite{cheung2017symmetrizing, BI:17b}.

We can describe the highest weight vectors of $\Sym^n\Sym^d V$ in terms of so called Young tableaux (see also \cite[\S3.9]{Ott:13}).
\begin{definition}
    A \emph{Young tableau} $T$ of shape $\lambda = (\lambda_1, \ldots, \lambda_r)$ where $\lambda$ is a partition is a left justified array of boxes where row $i$ contains $\lambda_i$ boxes and each box contains a positive integer.
    If the tableau contains the numbers $1$ through $n$ each $d$ times it is said to have \emph{(rectangular) content} $n \times d$, for example {\Yvcentermath1\tiny\young(1231,23)} has content $3\times 2$.
    A Young tableaux is said to be \emph{semistandard} if the entries are strictly increasing in each column and non-decreasing in each row, for example {\Yvcentermath1\tiny\young(1123,23)} is semistandard, while {\Yvcentermath1\tiny\young(1231,23)} is not.
    A Young tableaux is said to be \emph{standard} if the entries are strictly increasing in each column and row and every entry occurs exactly once. For example, {\Yvcentermath1\tiny\young(134678,259)} is standard.
\end{definition}

Fix a tableau $T$ of shape $\la$ with content $(nd)\times 1$ and fix a tensor $p = \sum_{i=1}^r \ell_{i,1}\otimes\cdots\otimes\ell_{i,d} \in \otimes^d \IC^m$.
We use arithmetic modulo $d$ with the system of representatives $\{1,\ldots,d\}$, so $a \mod d \in \{1,\ldots,d\}$.
Each of the sets $\{1,\ldots,d\}, \{d+1,\ldots,2d\},\ldots$ is called a \emph{block}.
We define $k(a) := \lceil a / d\rceil$. We define $j(a) := a \mod d$, which gives the position of the element $a$ in its block.
A placement
\[
\vartheta : \{1, \ldots, nd\} \to \{\ell_{i,j} \mid 1 \leq i \leq r, 1 \leq j \leq d\}
\]
is called \emph{proper} if
there is a map $\varphi:\{1,\ldots,n\}\to\{1,\ldots,r\}$
such that
$\vartheta(a) = \ell_{\varphi(k(a)),j(a)}$.
We define the determinant of a matrix that has more rows than columns as the determinant of its largest top square submatrix.

We define the polynomial $f_T$ via its evaluation on $p$:
\begin{equation}\label{eq:sumpropertheta}
    f_T(p)
    := \sum_{\text{proper\ }\vartheta} \prod_{c = 1}^{\lambda_1}\det{}_{\vartheta, c} \ \text{ with } \ \det{}_{\vartheta, c} := \det\left(\vartheta(T(1, c)) \dots \vartheta(T(\mu_c, c))\right)
\end{equation}

Pictorially $\varphi$ chooses one of the rank $1$ tensors for each block of $d$ numbers and places those onto $T$.
Then we take the product of the columnwise determinants.
The evaluation $f(p)$ is now the sum over all possible choices.

It is a classical result from multilinear algebra that this construction yields a well-defined polynomial of weight $\la$ on $\otimes^d \IC^m$.
If $T$ is the column-standard tableau, then $f_T \in \HWV_\la(\Sym^n\otimes^d \IC^m)$ is not hard to verify.
Schur-Weyl duality states that
$\otimes^n \otimes^d \IC^m = \bigoplus_{\la} S_\la(V) \otimes [\la]$,
where the sum goes over all partitions $\la$ of $nd$ into at most $m$ parts,
and where $S_\la(V)$ is the irreducible $\GL_m$-representation of type $\la$ (called the Schur module) and $[\la]$ is the irreducible $\aS_{dn}$-representation of type $\la$ (called the Specht module).
Since a basis of $[\la]$ is given by the standard tableaux of shape $\la$,
this immediately implies that
\begin{equation}\label{eq:span}
\text{$\HWV_\la(\Sym^n\otimes^d \IC^m)$ is the linear span of the $f_T$, where $T$ is standard of shape $\la$.}
\end{equation} See for example \cite{Ott:13} or \cite[Ch.~19]{BI:17} for a detailed exposition.

The following Lemmas~\ref{lem:orderedentries} and~\ref{lem:nodoubleentry} follow from eq.~\eqref{eq:sumpropertheta}.

\begin{lemma}\label{lem:orderedentries}
Let $T$ and $T'$ be Young tableaux of the same shape with content $(nd)\times 1$ such that $T'$ can be obtained from $T$ by performing permutations within the blocks.
The functions $f_T$ and $f_{T'}$ coincide after restricting their domains of definition from $\otimes^d\IC^m$ to $\Sym^d\IC^m$.
\end{lemma}
\begin{proof}
If $p$ is symmetric, then $p$ has a Waring rank decomposition, i.e., there exists $r\in\IN$ and homogeneous linear forms $p_1,\ldots,p_r$ such that $p=\sum_{i=1}^r p_i^{\otimes d}$.
Using this decomposition for $p$, we see that the summands of $f_T(p)$ and $f_{T'}(p)$ in \eqref{eq:sumpropertheta} coincide.
\end{proof}

Lemma~\ref{lem:orderedentries} implies that in order to define the restriction of $f_T$ to symmetric tensors we only need to define the blocks in $T$, but not the internal structure of the blocks.
Thus for a tableau with content $(nd)\times 1$ we define the tableau $\hat T$ by replacing all entries $a \in\{1,\ldots,nd\}$ by $k(a)$.
The resulting tableau $\hat T$ has content $n \times d$.
For example, if $n=2$, $d=4$,  $T={\Yvcentermath1\tiny\young(134678,25)}$, then $\hat T = {\Yvcentermath1\tiny\young(111222,12)}$.
For a tableau $\hat T$ with content $n \times d$ we define $f_{\hat T} \in \Sym^d \Sym^n \IC^m$ as the restriction of $f_T$ to $\Sym^n\IC^m$.

\begin{lemma}\label{lem:nodoubleentry}
Let $T$ be a Young tableau that has a column in which there are two or more entries from the same block. Then $f_T=0$.
\end{lemma}
\begin{proof}
Let $c$ be the column in $T$ in which there are two or more entries from the same block.
As in Lemma~\ref{lem:orderedentries}, consider the evaluation of $f_T$ at a point $p$ in its Waring rank decomposition. We observe that every summand in eq.~\eqref{eq:sumpropertheta} is zero, because the determinant corresponding to the column $c$ has a repeated column.
\end{proof}

In other words, Lemma~\ref{lem:nodoubleentry} says that $f_{\hat T}=0$ if $\hat T$ contains a column in which a number appears at least twice.
Combining this insight with eq.~\eqref{eq:span}, we conclude that
\begin{equation}\label{eq:spanssyt}
\begin{split}
\HWV_\la(\Sym^n\Sym^d \IC^m) & \text{ is the linear span of the $f_{\hat T}$,} \\
& \text{ where $\hat T$ is semistandard of shape $\la$ with content $n \times d$.}
\end{split}
\end{equation}

\begin{remark}
\label{rem:waringevaluation}
From eq.~\ref{eq:sumpropertheta} and writing $p$ in its Waring rank decomposition, we immediately get an $O(\WR(p)^n \cdot \poly(n,d,m))$ algorithm to evaluate $f_{\hat T}(p)$.
\end{remark}

\section{Non-commutative algebraic branching programs}
\label{sec:abp}
For an in-depth formal study of ncABPs we now introduce additional notation (cp.~Section~\ref{sec:prelude}).

\begin{definition}\label{def:ncabp}
    Let $V$ be a vector space.
    \begin{itemize}
        \item
            A non-commutative algebraic branching program (ncABP) $A$ is an acyclic directed graph with two distinguished nodes $s$ and $t$ and edges labeled with elements from $V$ and every path from $s$ to $t$ having the same length.
            This makes $A$ layered, with layer $k$ containing all vertices of distance $k$ from $s$.
        \item
            The weight $w(P)$ of a path $P$ with edge labels $\ell_1, \ldots, \ell_d \in V$ is defined as
            $
                w(P) := \ell_1 \otimes \dots \otimes \ell_d\,.
            $
        \item
            The tensor computed at a node $v$ in $A$ is
            $
                \hat{w}(v) = \sum_{s-v\ path\ P}w(P)\,.
            $
            By convention the tensor computed at $s$ is $1$.
        \item
            The tensor computed by $A$ is the tensor computed at $t$.
        \item
            The size of an ncABP is the number of vertices.
        \item
            The width of an ncABP is the largest number of vertices in any layer.
    \end{itemize}
\end{definition}
In particular we will be looking at ncABPs computing symmetric tensors $p$ and the evaluation of highest weight vectors at $p$.
An example is given in Figure~\ref{fig:ncabp}.

Each node in layer $k$ computes a tensor in $V^{\otimes k}$.
We show in Proposition~\ref{prop:minimalsym} that there is always a minimal ncABP where all these computed tensors are also symmetric and whose size is exactly the size of the partial derivative space of $p$.
An example is given in Figure~\ref{fig:ncabp}.

We can now use the ``overlapping structure of the paths through ncABPs'' to our advantage in evaluating HVWs by using dynamic programming.
\begin{theorem}
    \label{thm:evalabp}
    The evaluation $f_T(p)$ of a highest weight vector $f_T \in \Sym^n\Sym^d \IC^m$ given by a Young tableau $T$ with content $(nd) \times 1$ and $r$ rows and a symmetric tensor $p \in \Sym^d\IC^m$ given by an ncAPB of width $w$ can be computed in time $O(w^{n + r}\poly(n,d,m))$.
\end{theorem}
\begin{proof}
Let $A$ be an ncABP with source $s$, sink $t$, and width $w$ computing a symmetric tensor $p \in \Sym^d \IC^m$.
W.l.o.g.\ let the numbers $i \cdot d + 1, \ldots, i \cdot d + d$ occur in order left to right in $T$ for any $i \in \{0, \ldots, n-1\}$, see Lemma~\ref{lem:orderedentries}.
Note that left to right is a unique ordering since if one column contains multiple of these numbers we already know $f_T = 0$, see Lemma~\ref{lem:nodoubleentry}.

Combining eq.~\eqref{eq:sumpropertheta} with
$
    p = \hat{w}(t) = \sum_{s-t\ path\ P}w(P)
$ (see Def.~\ref{def:ncabp})
we see that
\begin{equation}\label{eq:sumproperthetapaths}
    f_T(p)
    := \sum_{\text{proper\,}\vartheta} \prod_{c = 1}^{\lambda_1}\det{}_{\vartheta, c} \ \text{ with } \ \det{}_{\vartheta, c} := \det\left(\vartheta(T(1, c)) \dots \vartheta(T(\mu_c, c))\right),
\end{equation}
where here $\vartheta:\{1,\ldots,nd\}\to V$ is called \emph{proper} if there exists $\varphi:\{1,\ldots,d\}\to\{s-t\  \text{path } P\}$ such that $\vartheta(a)=\text{the label of the $j(a)$-th edge of } \varphi(k(a))$ (see the definitions of $j$ and $k$ in Section~\ref{sec:base}).

We now calculate partial evaluations in a column by column fashion from right to left.
In order to do this we define a partial placement $\vartheta|_{\leq k}$ to be the restriction of $\vartheta$ to the boxes in the first $k$ columns of $T$.

We now observe a common factor for a fixed partial placement $\vartheta|_{\leq k}$:
\begin{align*}
    \sum_{\text{proper $\vartheta$ extending $\vartheta|_{\leq k}$}}\prod_{c = 1}^{\lambda_1}\det{}_{\vartheta,c}
    &=
    \left(\prod_{c = 1}^{k}\det{}_{\vartheta|_{\leq k},c}\right) \underbrace{\left(\sum_{\text{proper $\vartheta$ extending $\vartheta|_{\leq k}$}}\prod_{c = k+1}^{\lambda_1}\det{}_{\vartheta,c}\right)}_{=:\alpha(\vartheta|_{\leq k})}\,.
\end{align*}
Each $\vartheta|_{\leq k}$ defines a set of $n$ partial $s-t$ paths (of potentially different lengths, one path for each block), where $\alpha(\vartheta|_{\leq k})$ only depends on the endpoints of these paths.
These paths are connected from $s$ up to these endpoints due to the nature of $T$ being ordered from left to right for each block of $n$ numbers.
This crucial observation allows us to store and reuse these values of $\alpha$ whenever two partial assignments correspond to lists of $n$ paths ending in the same vertices of $A$.

We can now calculate the evaluation as $
    f_T(p) = \alpha(\vartheta|_{\leq 0}) = \alpha(\emptyset).
$

Since the length of each of the paths defined by any $\vartheta|_{\leq k}$ are fixed for fixed $k$, there are at most $w^n$ possible different values for $\alpha$ that need to be computed.
So in total this evaluation algorithm has running time $O(w^{n+r} \poly(n,d,m))$.
The $w^{r}$ term comes from all the possibilities to extend a given $\vartheta|_{\leq k}$ by one column of $T$.
\end{proof}

\begin{remark}
    Note that Theorem~\ref{thm:evalabp} is a generalisation of the dynamic programming used in \cite{DBLP:conf/icalp/DorflerIP19} to evaluate HWVs at the Chow variety $\Ch^d_m$.
    The Chow variety $\Ch^d_m$ consists of products of $d$ linear forms $\ell_1 \odot \dots \odot \ell_d \in \Sym^d \IC^m$.
    Here the minimal ncABP $A$ of $\ell_1 \odot \dots \odot \ell_d$ corresponds to having subsets of $\{1, \ldots, d\}$ as vertices where two vertices $U, V \subseteq \{1, \ldots n\}$ are connected by an edge labeled $\ell_i$ iff $U \setminus V = \{i\}$ and $U \supset V$.
    Then $A$ has size exactly $2^d$ and width $\binom{d}{k}$ on layer $k$ while layer $k$ contains all the sets of size $k$.
\end{remark}

We now give the connection between the width of ncABPs, and the dimension of the partial derivative spaces of the symmetric tensors computed by the ncAPB.
We additionally show that ncABPs can efficiently compute partial derivatives.

First note that the following equivalence between partial derivatives and polynomial contractions is well known for fields of characteristic $0$, see for example \cite[Equation 1.1.2]{iarrobino1999power}.
We reformulate this as an equivalence between partial derivatives and tensor contractions instead. The tensor contraction $\langle \cdot, \cdot \rangle : \bigotimes^r V  \times \bigotimes^s V \to \bigotimes^{s-r} V$ is defined for any $r < s$ on the basis vectors via 
\[
    \langle e_{i_1} \otimes \cdots \otimes  e_{i_{r}}, e_{j_1} \otimes \cdots \otimes e_{j_{s}}\rangle = \begin{cases}
    e_{j_{r+1}} \otimes e_{j_{r+2}} \otimes \cdots \otimes e_{j_s} & \text{if $i_k = j_k$ for all $1 \leq k \leq r$}\\
    0 & \text{otherwise}\\
    \end{cases}
\]
and extended via linear continuation in both parameters.
\begin{lemma}
    \label{lem:derivativecontr}
    Let $\varphi$ be the canonical isomorphism between $\Sym^d \IC^m$ and $\IC[x_1,\ldots,x_m]_d$ defined via $\varphi\left(e_{i_1} \odot \dots \odot e_{i_d}\right) = x_{i_1} \cdot \dots \cdot x_{i_d}$.
    Then the partial derivative $\frac{\partial^k}{\partial\ell_1 \cdots \partial\ell_k}t$ of a symmetric tensor $t \in \Sym^d V$ is given by the tensor contraction $\frac{d!}{(d-k)!}\langle \ell_1 \otimes \dots \otimes \ell_k, t \rangle$.

    Since $t$ is symmetric the partial derivative $\frac{\partial^k}{\partial\ell_1 \cdots \partial\ell_k}t$ is also given by $\frac{d!}{(d-k)!} \langle \ell_1 \odot \dots \odot \ell_k, t \rangle$.
\end{lemma}
\begin{proof}
    It suffices to prove this for the case $k=1$ since repeated tensor contraction is the same as one big tensor contraction and the same holds for partial derivatives.
    Since both tensor contraction and taking derivates are linear operations in both parameters we can restrict ourselves to the derivative $\frac{\partial}{\partial e_i}(e_{j_1} \odot \dots \odot e_{j_d})$ and prove that $\frac{\partial}{\partial e_i}(e_{j_1} \odot \dots \odot e_{j_d}) = d \cdot \langle \ell_1 \otimes \dots \otimes \ell_k, t \rangle$.
    The factor of $\frac{d!}{(d-k)!}$ is then the result of repeatedly taking the derivative.

    In case $e_i$ is not any of $e_{j_1}, \dots, e_{j_d}$ clearly
    \[
        \frac{\partial}{\partial e_i}(e_{j_1} \odot \dots \odot e_{j_d}) = 0 = \frac{d!}{(d-k)!} \langle e_i, e_{j_1} \odot \dots \odot e_{j_d} \rangle
    \]
    so w.l.o.g.\ we can now assume due to symmetry $e_{j_1} = e_i$.

    We can write $\varphi\left(e_i \odot e_{j_2} \odot e_{j_3} \odot \dots \odot e_{j_d}\right) = x_i^h \cdot q$ for some monomial $q \in \IC[x_1, \dots,x_m]$ not containing $x_i$.
    \begin{align*}
        \varphi\left(\frac{\partial}{\partial e_i}\left(e_i \odot e_{j_2} \odot e_{j_3} \odot \dots \odot e_{j_d}\right)\right)
        &= h \cdot x_i^{h-1} \cdot q\\
        &= \varphi\left(h \cdot e_{j_2} \odot e_{j_3} \odot \dots \odot e_{j_d}\right)\\
        &= \varphi\left(\langle e_i, h \cdot e_i \otimes (e_{j_2} \odot e_{j_3} \odot \dots \odot e_{j_d})\rangle\right)\\
        &= \varphi\left(\langle e_i, d \cdot e_i \odot e_{j_2} \odot e_{j_3} \odot \dots \odot e_{j_d}\rangle\right)
    \end{align*}
    The last equality follows from the fact that all terms of the symmetric tensor not containing $e_i$ as the first component of the tensor vanish under the tensor contraction.
\end{proof}

\begin{lemma}
    \label{lem:derivativeabps}
    If $A$ is an ncABP computing a symmetric tensor $p \in \Sym^d V$, then the $k$-th derivatives are linear combinations of the tensors computed at the $(d-k)$-th layer of $A$.
\end{lemma}
\begin{proof}
    As proven in Lemma~\ref{lem:derivativecontr} the derivatives are just tensor contractions. A tensor contraction on an ncABP replaces the last $k$ edges on each $s$-$t$ path by constants\footnote{due to the symmetry of $p$ we could even choose any $k$ layers and all outgoing edges out of these chosen layers would be replaced by constants for the derivative.}, thus directly proving the claim.
\end{proof}

We will now characterize the minimal size of ncABPs via the dimension of the partial derivative spaces.
For this we denote by $\partial^{=k}(t)$ the partial derivative space of $k$-th order for $t \in \Sym^d V$:
\[
    \partial^{=k}(t) := \{ \langle q, t \rangle \mid q \in \Sym^k V\}
\]
Analogously we define
\[
    \partial^{\leq k}(t) := \spn \bigcup_{i = 0}^k \partial^{=i}(t)\,.
\]
Note that the usage of tensor contractions instead of derivatives is just for simplicity.

For a list $q \in \{1,\ldots,m\}^k$ let $e_q := e_{q_1} \otimes \cdots \otimes e_{q_k}$.
For a tensor $p \in \otimes^d \IC^m$ we define the $m^k \times m^{d-k}$ matrix $M_k(p)$ whose rows are indexed by elements $q \in \{1,\ldots,m\}^k$ and whose columns are indexed by elements in $q' \in \{1,\ldots,m\}^{d-k}$ via
\begin{equation}\label{eq:Mk}
M_k(p)[q, q'] := \text{ the coefficient of $e_q \otimes e_{q'}$ in $p$}.
\end{equation}

\begin{proposition}
    \label{prop:minimalsym}
    If $A$ is an ncABP computing a symmetric tensor $p \in \Sym^d V$, then there is an ncABP $B$ with the following properties:
    \begin{enumerate}
        \item $B$ also computes $p$.
        \item Each layer of $B$ has at most as many vertices as the same layer in $A$.
        \item Each node of $B$ computes a symmetric tensor.
        \item The $k$-th layer of $B$ has precisely $\dim\partial^{=k}(p)$ many vertices which is the optimal width.
    \end{enumerate}
\end{proposition}
\begin{proof}
    We mainly follow Nisan \cite{nisan1991lower} with this contruction who constructed minimal ncABPs and extend this to also compute symmetric tensors at each node and establishing the connection to the dimensions of the partial derivative spaces.
    For an example of a minimal ncAPB with symmetric tensors computed at each node can be seen in Figure~\ref{fig:ncabp}.

    Let $v_1, \ldots, v_t$ be the vertices in a fixed layer $k$. Let $M_k[q, q'] := M_k(p)[q,q']$ from eq.~\eqref{eq:Mk}.
    Note that the row of $M_k$ corresponding to $q$ is given precisely by the tensor contraction $\langle q, p \rangle$ and it is thus by Lemma~\ref{lem:derivativecontr} a partial derivative of $k$-th order.
    Therefore $\rank M_k = \dim \partial^{=k}(p)$.

    Now we can construct two matrices $L_k$ and $R_k$.
    Here $L_k[q, i]$ for indices $q \in \{e_1, \dots, e_{\dim V}\}^{\otimes k}$ is defined as the coefficient of $q$ in $\hat{w}(v_i)$ and $R_k[i, q']$ for indices $q' \in \{e_1, \dots, e_{\dim V}\}^{\otimes (d-k)}$ is defined as the coefficient of $q'$ in the tensor computed by the restricted ncABP with source $v_i$.
    It is easy to verify $M_k = L_k R_k$.

    Hence if $t > \rank L_k$ there must be some vertices $v_i$ computing a linear combination of the other vertices in the same layer, thus all outgoing edges of $v_i$ can be replaced by precisely this linear combination, allowing us to remove $v_i$.
    In this way we can remove some $v_i$ as long as $t > \rank R_k$.

    After this process finishes we have $t = \rank L_k = \rank R_k = \rank M_k = \dim \partial^{=k}(p)$ proving the claims on the width of the layers.

    Since by Lemma~\ref{lem:derivativeabps} all the $(d-k)$-th partial derivatives are linear combinations of restrictions of the ncABP to the first $k$ levels we can now replace all vertices on the $k$-th level by $t$ vertices computing a symmetric tensor basis of the $k$-th partial derivatives thus proving the remaining claim.
\end{proof}

From this characterization of ncABP size as the rank of the partial derivative matrices we can also see that ncABP size is preserved under approximation.
This was remarked by Michael Forbes \cite{forbesWACT16}, but we give a proof for the sake of completeness.
\begin{corollary}
    \label{cor:abpclosure}
    Let $p \in \Sym^d V$ and $(A_i)_{i \in \IN}$ be ncABPs s.t.\ $A_i$ computes $p_i \in \otimes^d V$ and has size $s_i \leq s$ and width $w_i \leq w$ with $\lim_{i \to \infty}p_i = p$.
    Then there is an ncABP $A$ computing $p$ with size at most $s$ and width at most $w$.
\end{corollary}
\begin{proof}
    Let the matrices $M_{k,p_i} := M_k(p_i)$ from eq.~\eqref{eq:Mk}.
    We have
    \[
        M_{k,p} = \lim_{i \to \inf} M_{k, p_i}\,.
    \]
    Since each $A_i$ has width at most $w$, we know that $\rank M_{k,p_i} \leq w_i \leq w$.
    This is characterized by all determinants of $(w+1) \times (w+1)$ minors of $M_{k,p_i}$ vanishing.
    So by continuity of the determinant also all $(w+1) \times (w+1)$ minors of $M_{k,p}$ vanish and thus $\dim \partial^{=k}(p) =\rank M_{k,p} \leq w$ and there is an ncABP $A$ with width at most $w$ by Proposition~\ref{prop:minimalsym}.

    This constructed $A$ directly has size at most $s$. For this we note that the partial derivatives of different orders are linearly independent, so $\dim \partial^{\leq d}(p) = \sum_{j=0}^d \dim \partial^{=j}(p) = s$.
    This is the same as looking at the rank of the direct sum $\oplus_{j=0}^d M_{j,p}$, so the bound on the size of $A$ follows from the same continuity argument.
\end{proof}

From this we can conclude an order of inclusion on the sets of symmetric tensors of small Waring rank, small border Waring rank and small non-commutative ncABP size.

\begin{corollary}
    Let $k \in \IN$ and
    \begin{align*}
        W_{k,d} &:= \{p \in \Sym^d V \mid \WR(p) \leq k\}\,,\\
        \overline{W_{k,d}} &:= \{p \in \Sym^d V \mid \WRbar(p) \leq k\}\,,\\
        B_{k,d} &:= \{p \in \Sym^d V \mid \ncw(p) \leq k\}\,.\\
        \overline{B_{k,d}} &:= \{p \in \Sym^d V \mid \underline{\ncw}(p) \leq k\}\,.\\
    \end{align*}
    Then
    \[
        W_{k, d} \subseteq \overline{W_{k, d}} \subseteq B_{k,d} = \overline{B_{k,d}}
    \]
and there exist $k,d$ for which the inclusions are strict.
\end{corollary}
\begin{proof}
    The inclusion $W_{k, d} \subseteq \overline{W_{k, d}}$ is trivial and $B_{k,d} = \overline{B_{k,d}}$ is proven in Corollary~\ref{cor:abpclosure}.
    To show $W_{k, d} \subsetneq \overline{W_{k, d}}$ is strict we refer to \cite{carlini2012solution} showing that $x^{d-1}y$ has Waring rank $d$ while it is known\footnote{Technically we need here that our base field is algebraically closed in order for this to be a border Waring rank decomposition, but $\IC$ satisfies this.} that $x^{d-1}y = \lim_{\varepsilon \to 0}\frac{1}{\varepsilon d}((x+\varepsilon y)^{d} - x^{d})$ and thus $x^{d-1}y$ has border Waring rank at most $2$.
    For the inclusion $W_{k, d} \subseteq B_{k,d}$ we can embed the $k$ summands $\ell_i^d$ of the Waring rank decomposition as disjoint $s-t$ paths in an ncABP of width $k$ and depth $d$.
    Here every edge on the path corresponding to $\ell_i^d$ has the label $\ell_i$.
    Since $B_{k,d}$ is closed this immediately proves $\overline{W_{k, d}} \subseteq B_{k,d}$.
    An example for $W_{k, d} \neq B_{k,d}$ is given by the $2 \times 2$ matrix multiplication polynomial $p=x_{1,1}^3 + 3 x_{1,1}x_{1,2}x_{2,1} + 3 x_{1,2}x_{2,2}x_{2,1} + x_{2,2}^3$ that is studied in \cite{CHILO:18}: We have $\ncw(p)=4$, but $\WRbar(p)\geq 5$, which can be seen using \emph{Young flattenings}. This representation theoretic technique is explained for example in \cite{Far:16}. The \texttt{Macaulay2} code

\smallskip

\scalebox{0.8}{
\begin{minipage}{13cm}%
\texttt{%
loadPackage "PieriMaps"\\
MX = pieri ($\{$4,3,2,2$\}$ , $\{$1,2,4$\}$ , QQ [x11,x12,x21,x22])\\
p = x11*x11*x11 + 3*x11*x12*x21 + 3*x12*x22*x21 + x22*x22*x22\\
rank(diff(p,MX))/rank(diff(x11\^{}3,MX))
}
\end{minipage}
}

\smallskip

\noindent outputs 5, which is the lower bound on the border Waring rank.
\end{proof}

Note that the following is still unknown:
\begin{question}
    Is there a polynomial $q$, such that $B_{k,d} \subseteq W_{q(k),d}$ or $B_{k,d} \subseteq \overline{W_{q(k),d}}$?
\end{question}

\section{Treewidth of Young tableaux}
\label{sec:treewidth}
Let $S$ be an arbitrary Young tableau containing the numbers $\{1, \ldots, n\}$.
We can associate with $S$ the undirected graph $G_S = (V_S, E_S)$ where $V_S = \{1, \ldots, n\}$ and $\{i, j\} \in E_S$ iff $i$ and $j$ are contained in some common column in $S$, see Figure~\ref{fig:treewidthevalexample}(a) and (b).

We are now going to study how we can use the graph parameter treewidth of $G_S$ to speed up the evaluation of highest weight vectors.
Treewidth has been intensely studied by Robertson and Seymour and has been applied numerous times to construct faster graph algorithms for cases where the treewidth is bounded by a function $o(n)$, most notably some algorithms for $\NP$-hard problems restricted to planar graphs, for example $3$-coloring.
See \cite{cygan2015parameterized} for an introduction to treewidth algorithms.

\begin{definition}\label{def:treedecomp}
    A \emph{tree decomposition} of a graph $G = (V, E)$ is a tree $\CT$ with vertices $X_1, X_2, \ldots, X_t$ called bags where $X_i \subseteq V$ and the following properties hold:
    \begin{itemize}
        \item $\cup_{i=1}^t X_i = V$
        \item For every edge $\{u, v\} \in E$ there is some bag $X_i$, s.t.\ $\{u, v\} \subseteq X_i$.
        \item For every vertex $v \in V$ the bags containing $v$ form a subtree of $\CT$.
    \end{itemize}

    The \emph{width} of a tree decomposition is the size of the largest bag minus one.
    The \emph{treewidth} of $G$ is then the smallest possible width of a tree decomposition for $G$.
\end{definition}

Often solving problems on graphs of bounded treewidth is easier then the general problem and indeed this is also the case for evaluating the highest weight vector corresponding to a graph if the graph $G_{\hat{T}}$ has bounded or low treewidth.

\begin{theorem}
    \label{thm:evaltreewidth}
    The evaluation $f_{\hat T}(p)$ for a highest weight vector $f_{\hat T} \in \Sym^n\Sym^d \IC^m$ given by a Young tableau $\hat{T}$ with content $n \times d$ and  a symmetric tensor $p \in \Sym^d \IC^m$ given by an ncABP $A$ of width $w$ can be computed in time $w^{\omega(\tau+1)}\poly(n,d,m,|\CT|)$ if a tree decomposition $\CT$ of $G_{\hat{T}}$ of width $\tau$ and size $|\CT|$ is given and given that we can multiply two matrices of size $\leq k \times k$ in time $O(k^\omega)$.
\end{theorem}
\begin{proof}
Let $A$ be a ncABP with source $v_\text{source}$ and sink $v_\text{sink}$. The label on the edge from $v$ to $w$ shall be called $A_{(v,w)}$.

A tableau $\hat T$ with its corresponding graph $G_{\hat T}$ is depicted in Figure~\ref{fig:treewidthevalexample}(a) and (b).
It is well known that every clique of a graph is fully contained in some bag of its tree decomposition.
Every column $c$ of $\hat{T}$ corresponds to a clique in $G_{\hat{T}}$, so there is some bag $X_i$ of $\CT$ which contains all the vertices corresponding to the entries of $c$.
We modify $\CT$ by adding a new vertex which is only adjacent to $X_i$.
This vertex is from now on associated with the column $c$ and contains all the entries contained in $c$ as its bag.
An example is given in Figure~\ref{fig:treewidthevalexample}(c).
Without loss of generality we can assume that if a vertex has only one child, then the vertex has the same bag as the child.
From now on we only need the structure of the subtree $\CT'$ of $\CT$ whose leaves are the vertices associated with columns and every vertex removed that is not on the path between two of these vertices.
We interpret $\CT'$ as an ordered binary tree rooted at an arbitrary internal\footnote{i.e.\ a non-leaf} vertex $r$.
In case any vertex $v$ of $\CT'$ has more than two children, we replace $v$ by a binary tree, where each added vertex has the same bag as $v$, see Figure~\ref{fig:treewidthevalexample}(e).
Since $\CT'$ is also a tree decomposition, for every vertex $v$ with two children $v_\text{left}$ and $v_\text{right}$ we have
\begin{equation}\label{eq:treedecompinclusion}
X_{v_\text{left}} \cap X_{v_\text{right}}\subseteq X_v.
\end{equation}

We start with a few observations.
We sort the leaves of $\CT'$ according to when they are visited by depth-first search that always takes the left child first. In this way every leaf $\CT'$ gets assigned an index from 1 to $\la_1$, which we call the \emph{traversal index} of the leaf.
Let $\nu_i$ denote the length of the column of $\la$ with traversal index $i$.
For any internal vertex $v$ of $\CT'$ let $\textsf{leftmost}(v)$ denote the traversal index of the leftmost leaf of the subtree rooted at $v$.
Analogously, let $\textsf{rightmost}(v)$ denote the traversal index of the rightmost leaf of the subtree rooted at $v$. For a leaf $v$ we define $\textsf{leftmost}(v)=\textsf{rightmost}(v)$ to be the traversal index of $v$.
For any internal vertex $v$ of $\CT'$ with two children $v_\text{left}, v_\text{right}$ by definition we have
\begin{equation}\label{eq:leftmostv}
\textsf{leftmost}(v)=\textsf{leftmost}(v_\text{left}) \text{ and } \textsf{rightmost}(v)=\textsf{rightmost}(v_\text{right})
\end{equation}
and
\begin{equation}\label{eq:rightmostleft}
\textsf{rightmost}(v_\text{left})=\textsf{leftmost}(v_\text{right})-1.
\end{equation}
We define $\text{leaves}(v)$ to be the set of leaves in $\CT'$ with traversal index at least $\textsf{leftmost}(v)$ and at most $\textsf{rightmost}(v)$.

For $1 \leq i \leq n$, $0 \leq t \leq \la_1$, define $\kappa_t(i)$ to be the number of times the number $i$ appears in columns with traversal index at most $t$. Figure~\ref{fig:treewidthevalexample}(f) shows diamond separators that mark the values for $t$ so that $\kappa_t(i)$ is the number of times the number $i$ appears in columns left of the diamond $t$. For an internal vertex $v$ with children $v_\text{left}$ and $v_\text{right}$ we define $\textsf{mid}(v) := \textsf{rightmost}(v_\text{left})$. Pictorially, this is the number of the diamond separator between the left and right subtree of $v$. If $v$ only has one child $w$, then $\textsf{mid}(v) = \textsf{mid}(w)$.
Let the ncABP $A$ have layers $L_0, \ldots, L_d$, $|L_0|=|L_d|=1$. We assume that in $A$ all edges between any layers $L_i$ and $L_{i+1}$ exist, hence we allow edges that are labelled with 0.
For $0 \leq t \leq \la_1$ define
\[
\CF_t := L_{\kappa_{t}(1)} \times \cdots \times L_{\kappa_{t}(n)}
\]
Let $t_\text{start} \leq t_\text{end}$ and let $\Phi_\text{start} \in \CF_{t_\text{start}}$ and $\Phi_\text{end} \in \CF_{t_\text{end}}$.
A \emph{$\Phi_\text{start}$-$\Phi_\text{end}$-multiwalk} is defined as a finite sequence
\[
\CW := (\Phi_{t_\text{start}}, \Phi_{t_\text{start}+1}, \Phi_{t_\text{start}+2}, \ldots, \Phi_{t_\text{end}})
\]
such that each $\Phi_{t} \in \CF_t$, and for all $t,i$ with $\kappa_t(i)=\kappa_{t+1}(i)$ we have $\Phi_{t}(i)=\Phi_{t+1}(i)$.
To explain this notion more pictorially, we define
a \emph{lazy walk} in a digraph to be a walk that as a step can remain at its vertex instead of advancing over an edge. If a digraph does not have any loops, then for every finite lazy walk there is a corresponding walk on the digraph that is obtained if we add all loops: Remaining at a vertex has the same effect as taking the loop. This is also true in the reverse direction.
The $i$-th walk of a $\Phi_\text{start}$-$\Phi_\text{end}$-multiwalk
$(\Phi_{t_\text{start}}, \Phi_{t_\text{start}+1}, \ldots, \Phi_{t_\text{end}})$
is defined as the sequence
$(\Phi_{t_\text{start}}(i), \Phi_{t_\text{start}+1}(i), \ldots, \Phi_{t_\text{end}}(i))$,
which is a lazy walk in $A$ from $\Phi_{t_\text{start}}(i)$ to $\Phi_{t_\text{end}}(i)$.

We now define the determinant $\det(\CW)$.
Note that $\Phi_{t}$ and $\Phi_{t+1}$ differ in exactly $\nu_{t+1}$ positions.
We define $\det{}_{\CW,t+1}$ as the determinant of the $\nu_{t+1}\times \nu_{t+1}$-matrix obtained from taking the top $\mu_{t+1}$ of the edge labels in $A$ that connect $\Phi_{t}$ with $\Phi_{t+1}$.
We define $\det(\CW)$ as
\[
\det(\CW) := \prod_{c=t_\text{start}+1}^{t_\text{end}}\det{}_{\CW,c}
\]
For $\Phi_\text{start} \in \CF_{\textsf{leftmost}(v)-1}$ and $\Phi_\text{end} \in \CF_{\textsf{rightmost}(v)}$ we define
\begin{equation}\label{def:D}
D[v,\Phi_\text{start},\Phi_\text{end}] :=
\sum_{\Phi_\text{start}\text{-}\Phi_\text{end}\text{-multiwalk } \CW} \det(\CW) \end{equation}

Note that
\begin{equation}\label{def:DII}
D[v,\Phi_\text{start},\Phi_\text{end}] =
\sum_{\Phi_\text{start}\text{-}\Phi_\text{end}\text{-multiwalk } \CW} \prod_{c=\textsf{leftmost}(v)}^{\textsf{rightmost}(v)} \det{}_{\CW,c}.
\end{equation}

We claim that
\begin{equation}\label{eq:fTDst}
f_{\hat T}(p) = D[r, \vv{v_\textup{source}}, \vv{v_{\textup{sink}}}],
\end{equation}
where $\vv{v_{\textup{source}}} = (v_{\textup{source}},v_{\textup{source}},\ldots,v_{\textup{source}})$ and $\vv{v_{\textup{sink}}} = (v_{\textup{sink}},v_{\textup{sink}},\ldots,v_{\textup{sink}})$. To see this,
we observe that left left-to-right ordering of the leaves of $\CT'$ defines an ordering on the columns of tableaux of shape $\la$. We call this ordering the leaf-ordering.
Let $T$ be the following tableau of shape $\la$ and content $(nd) \times 1$ that is a preimage of $\hat T$ under the $\hat.$-operation (see Section~\ref{sec:base}):
We greedily go through the columns of $\hat T$ from left to right in the leaf-order and replace each entry $i$ by the smallest still unused number from $\{(i-1) \cdot d + 1, (i-1) \cdot d + 2, \ldots, i \cdot d\}$, see Figure~\ref{fig:treewidthevalexample}(d) for an example.
Then $f_{\hat{T}}(p) = f_{T}(p)$ is given as
\begin{align*}
    f_{T}(p) \stackrel{\eqref{eq:sumpropertheta}}{=} \sum_{\text{proper\ }\vartheta} \prod_{c = 1}^{\lambda_1}\det{}_{\vartheta, c} \stackrel{(\ast)}{=} D[r, \vv{v_{\textup{source}}}, \vv{v_{\textup{sink}}}]\,.
\end{align*}
$(\ast)$ can be seen from the fact that there is a natural 1:1 correspondence between proper $\vartheta$ and $\vv{v_{\textup{source}}}$-$\vv{v_{\textup{sink}}}$-multiwalks $\CW$: A tensor assigned to the $i$-th block of $T$
is given by an $v_{\textup{source}}$-$v_{\textup{sink}}$-path in $A$, which uniquely specifies the $i$-th path of the multiwalk $\CW$.
If $\vartheta$ is mapped to $\CW$ under this bijection, then $\det(\CW) = \prod_{c = 1}^{\lambda_1}\det{}_{\vartheta, c}$.
This proves \eqref{eq:fTDst}.

We now explain how to compute $D[r, \vv{v_{\textup{source}}}, \vv{v_{\textup{sink}}}]$ recursively over the tree structure of~$\CT'$. We separate the explanation into several claims.
First, we claim that for every internal vertex $v \in \CT'$ with only one child $v'$ we have
\begin{align}\label{eq:claimonlychild}
    D[v, \Phi_\text{start}, \Phi_\text{end}] = D[v', \Phi_\text{start}, \Phi_\text{end}].
\end{align}
The right-hand side is well-defined, because $\textsf{leftmost}(v)=\textsf{leftmost}(v')$ and $\textsf{rightmost}(v)=\textsf{rightmost}(v')$. The equality follows directly from the definition: \eqref{def:D}.
Next, we claim that for every leaf vertex $v \in \CT'$ with corresponding column $c$ we have
\begin{align}\label{eq:claimleaf}
    D[v, \Phi_\text{start}, \Phi_\text{end}] = \det(A_{(\Phi_\text{start}(c_1), \Phi_\text{end}(c_1))}, A_{(\Phi_\text{start}(c_2), \Phi_\text{end}(c_2))}, \cdots, A_{(\Phi_\text{start}(c_{|c|}), \Phi_\text{end}(c_{|c|}))}).
\end{align}
This follows from the fact that in this case there is exactly one $\Phi_\text{start}$-$ \Phi_\text{end}$-multiwalk $\CW$, and 
$\textsf{leftmost}(v)=\textsf{rightmost}(v)$ is the traversal index of $v$.
The crucial claim is the following.
For any inner vertex $v$ of $\CT'$ with two children $v_\text{left}, v_\text{right}$ we claim
\begin{align}\label{eq:claimdistributive}
    D[v, \Phi_\text{start}, \Phi_\text{end}] = \sum_{\Phi_\text{mid} \in \CF_{\textsf{mid}(v)}} D[v_\text{left},\Phi_\text{start}, \Phi_\text{mid}] \cdot D[v_\text{right},\Phi_\text{mid}, \Phi_\text{end}].
\end{align}
Before proving this, first note that $D[v_\text{left},\Phi_\text{start}, \Phi_\text{mid}]$ on the right-hand side is well-defined, as can be seen by combining \eqref{eq:leftmostv} and \eqref{eq:rightmostleft} with the definition of $\textsf{mid}(v)$.
Analogously, $D[v_\text{right},\Phi_\text{mid}, \Phi_\text{end}]$ is well-defined.

The key tool in the proof of \eqref{eq:claimdistributive} is the bijection
\begin{equation}\label{eq:pathbijection}
\{\Phi_\text{start}\text{-}\Phi_\text{end}\text{-multiwalk}\} \simeq \bigcup_{\Phi_\text{mid} \in \CF_{\textsf{mid}(v)}} \left(\{\Phi_\text{start}\text{-}\Phi_\text{mid}\text{-multiwalk}\} \times \{\Phi_\text{mid}\text{-}\Phi_\text{end}\text{-multiwalk}\}\right)
\end{equation}
given by splitting the multiwalk into two multiwalks, where the inverse map is given by contatenating two multiwalks. The union on the right-hand side is a disjoint union.
\eqref{eq:claimdistributive} is now proved by a direct calculation as follows.
\begin{eqnarray*}
&&\sum_{\Phi_\text{mid} \in \CF_{\textsf{mid}(v)}} D[v_\text{left},\Phi_\text{start}, \Phi_\text{mid}] \cdot D[v_\text{right},\Phi_\text{mid}, \Phi_\text{end}]
\\
&\stackrel{\eqref{def:DII}}{=}&
\sum_{\Phi_\text{mid} \in \CF_{\textsf{mid}(v)}}
\left(\sum_{\Phi_\text{start}\text{-}\Phi_\text{mid}\text{-multiwalk } \CW_\text{left}} \ 
\prod_{c_\text{left}=\textsf{leftmost}(v_\text{left})}^{\textsf{rightmost}(v_\text{left})}\det{}_{\CW,c_\text{left}}
\right)\\
&&\hspace{2cm}
\cdot\left(\sum_{\Phi_\text{mid}\text{-}\Phi_\text{end}\text{-multiwalk } \CW_\text{right}} \ 
\prod_{c_\text{right}=\textsf{leftmost}(v_\text{right})}^{\textsf{rightmost}(v_\text{right})}\det{}_{\CW,c_\text{right}}\right)
\\
&\stackrel{\eqref{eq:pathbijection}}{=}&
\sum_{\Phi_\text{start}\text{-}\Phi_\text{end}\text{-multiwalk } \CW} \left(\prod_{c_\text{left}=\textsf{leftmost}(v_\text{left})}^{\textsf{rightmost}(v_\text{left})}\det{}_{\CW,c_\text{left}}\right)
\left(\prod_{c_\text{right}=\textsf{leftmost}(v_\text{right})}^{\textsf{rightmost}(v_\text{right})}\det{}_{\CW,c_\text{right}}\right)
\\
&\stackrel{\eqref{eq:rightmostleft}}{=}&
\sum_{\Phi_\text{start}\text{-}\Phi_\text{end}\text{-multiwalk } \CW} \left(\prod_{c=\textsf{leftmost}(v_\text{left})}^{\textsf{rightmost}(v_\text{right})}\det{}_{\CW,c}\right)
\stackrel{\eqref{eq:leftmostv},\eqref{def:DII}}{=} D[v,\Phi_\text{start},\Phi_\text{end}].
\end{eqnarray*}
This proves \eqref{eq:claimdistributive}.

Equations \eqref{eq:claimonlychild}, \eqref{eq:claimleaf}, and \eqref{eq:claimdistributive} give us a procedure to compute $D[r, \vv{s}, \vv{t}]$ by induction over the structure of the tree $\CT'$.
But we can improve the running time significantly as follows.

We arbitrarily order the vertices within each layer such that every vertex $v \in L_j$ has an index $\iota(v) \in \{1,\ldots,|L_j|\}$. Of course the vertex $v$ with $\iota(v)=1$ plays no special role, but $v$ is useful to find a normal form for paths from $L_j$ to $L_j$ of length 0: They go from $v$ to $v$.

For a set of $X \subseteq \IN$
define the subset $\CF_{t}^{X}\subseteq \CF_{t}$ as
\[
\CF_{t}^{X} := \{\Phi \in \CF_{t} \mid \text{ for all } i \notin X: \ \Phi(i)=v \text{ with } \iota(v)=1\}
\]
For every $\Phi \in \CF_{t}$ define $\Phi^X \in \CF_{t}^X$ via
\begin{equation}\label{eq:defsupX}
\Phi^X(i) := \begin{cases}
\Phi(i) & \text{ if } i \in X \\
v \text{ with } \iota(v)=1 & \text{ otherwise}.
\end{cases}
\end{equation}
Clearly $\vv{v_{\textup{source}}}=\vv{v_{\textup{source}}}^{X_r}$ and $\vv{v_{\textup{sink}}}=\vv{v_{\textup{sink}}}^{X_r}$, hence $D[r, \vv{v_{\textup{source}}}, \vv{v_{\textup{sink}}}]=D[r, \vv{v_{\textup{source}}}^{X_r}, \vv{v_{\textup{sink}}}^{X_r}]$.
We claim that
\begin{equation}\label{eq:speedup}
D[v,\Phi_\text{start}, \Phi_\text{end}]=
\begin{cases}
0 & \text{if there exists $i \notin X_v$ with $\iota(\Phi_\text{start}(i))\neq\iota(\Phi_\text{end}(i))$} \\
D[v,\Phi_\text{start}^{X_v}, \Phi_\text{end}^{X_v}]
& \text{otherwise}.
\end{cases}
\end{equation}
To see this, first assume that there is $i \notin X_v$ with $\iota(\Phi_\text{start}(i))\neq\iota(\Phi_\text{end}(i))$.
Since $i \notin X_v$, we either have \emph{all} instances of $i$ in $\text{leaves}(v_\text{left})$ or \emph{all} instances of $i$ in $\text{leaves}(v_\text{right})$ or \emph{no} instance of $i$ in $\text{leaves}(v)$. The first two cases are impossible, because in those cases we have $\iota(\Phi_\text{start}(i))=1=\iota(\Phi_\text{end}(i))$, because the first and last layer only have one vertex each. In the third case, there is no $\Phi_\text{start}$-$\Phi_\text{end}$-multiwalk, because the $i$-th path in the multiwalk would have to start at a vertex and end at a different vertex without using any edge. This proves the first case of \eqref{eq:speedup}.
Now, assume that for all $i \notin X_v$ we have $\iota(\Phi_\text{start}(i))=\iota(\Phi_\text{end}(i))$.
We have a distinction into the same three cases as above.
If $i \notin X_v$ has all instances of $i$ in $\text{leaves}(v_\text{left})$ or in $\text{leaves}(v_\text{right})$, then
$\iota(\Phi_\text{start}(i))=1=\iota(\Phi_\text{end}(i))$, which implies
$\Phi_\text{start}^{X_v}(i)=\Phi_\text{start}(i)$ and $\Phi_\text{end}^{X_v}(i)=\Phi_\text{end}(i)$.
If $i \notin X_v$ has no instance of $i$ in $\text{leaves}(v)$, then the value of $\iota(\Phi_\text{start}(i))$ and $\iota(\Phi_\text{end}(i))$ does not affect $D[v,\Phi_\text{start},\Phi_\text{end}]$, as long as $\iota(\Phi_\text{start}(i))=\iota(\Phi_\text{end}(i))$, because the $i$-th path in any $\Phi_\text{start}$-$\Phi_\text{end}$-multiwalk is unique and of length 0.
This proves \eqref{eq:speedup}.

We use the short notation $\iota(\Phi_\text{start}|_{\overline{X_v}})\neq\iota(\Phi_\text{end}|_{\overline{X_v}})$ for the first condition in \eqref{eq:speedup}: The $\iota$-values of the vectors $\Phi_\text{start}$ and $\Phi_\text{end}$ are different when restricted to the complement of $X_v$.

To compute $D[v, \Phi_\text{start}^{X_v}, \Phi_\text{end}^{X_v}]$ recursively we observe the following.
\begin{eqnarray}\label{eq:claimdistributiveII}
D[v, \Phi_\text{start}^{X_v}, \Phi_\text{end}^{X_v}] &=& 
\sum_{\Phi_\text{mid} \in \CF_{\textsf{mid}(v)}} D[v_\text{left},\Phi_\text{start}^{X_v}, \Phi_\text{mid}] \cdot D[v_\text{right},\Phi_\text{mid}, \Phi_\text{end}^{X_{v}}]\\\nonumber
&\stackrel{\eqref{eq:speedup}}{=}& 
\sum_{\Phi_\text{mid}} D[v_\text{left},\Phi_\text{start}^{X_{v}}, \Phi_\text{mid}] \cdot D[v_\text{right},\Phi_\text{mid}, \Phi_\text{end}^{X_{v}}],
\end{eqnarray}
where the second sum is over all those $\Phi_\text{mid} \in \CF_{\textsf{mid}(v)}$ that satisfy both
\begin{itemize}
\item[$\bullet$] $\Phi_\text{start}^{X_v}(i) = \Phi_\text{mid}(i)$ for all $i \notin X_{v_\text{left}}$ and
\item[$\bullet$] $\Phi_\text{end}^{X_v}(i) = \Phi_\text{mid}(i)$ for all $i \notin X_{v_\text{right}}$.
\end{itemize}
It follows from \eqref{eq:defsupX} that all these summation indices $\Phi_\text{mid}$ satisfy $\iota(\Phi_\text{mid}(i))=1$ for all $i \in (\overline{X_v} \cap \overline{X_{v_\text{left}}})\cup(\overline{X_v} \cap \overline{X_{v_\text{right}}})$, where the bar denotes the set complement.
But $(\overline{X_v} \cap \overline{X_{v_\text{left}}})\cup(\overline{X_v} \cap \overline{X_{v_\text{right}}}) = \overline{X_v \cup (X_{v_\text{left}} \cap X_{v_\text{right}})}$, which equals $\overline{X_v}$ by \eqref{eq:treedecompinclusion}. This implies that $\Phi_\text{mid} \in \CF_{\textsf{mid}(v)}^{X_v}$.
Therefore we can rewrite \eqref{eq:claimdistributiveII} as
\begin{eqnarray*}
D[v, \Phi_\text{start}^{X_v}, \Phi_\text{end}^{X_v}] &=& 
\sum_{\Phi_\text{mid} \in \CF^{X_v}_{\textsf{mid}(v)}} D[v_\text{left},\Phi_\text{start}^{X_v}, \Phi_\text{mid}] \cdot D[v_\text{right},\Phi_\text{mid}, \Phi_\text{end}^{X_{v}}]
\end{eqnarray*}
\begin{eqnarray*}
\stackrel{\eqref{eq:speedup}}{=}\sum_{\Phi_\text{mid} \in \CF^{X_v}_{\textsf{mid}(v)}}
\left(
0 \ \text{ if } \ \iota(\Phi_\text{start}^{X_v}|_{\overline{X_{v_\text{left}}}})\neq\iota(\Phi_\text{mid}|_{\overline{X_{v_\text{left}}}}),
\atop
D[v_\text{left},(\Phi_\text{start}^{X_v})^{X_{v_\text{left}}}, \Phi_\text{mid}^{X_{v_\text{left}}}] \text{ otherwise}
\right)
\cdot
\left(
0 \ \text{ if } \ \iota(\Phi_\text{mid}|_{\overline{X_{v_\text{right}}}})\neq\iota(\Phi_\text{end}^{X_v}|_{\overline{X_{v_\text{right}}}}),
\atop
D[v_\text{right},\Phi_\text{mid}^{X_{v_\text{right}}}, (\Phi_\text{end}^{X_v})^{X_{v_\text{right}}}]
\right)
\end{eqnarray*}

Using this equality we can compute the $|\CF^{X_v}|\times|\CF^{X_v}|$ matrix $D[v,\CF^{X_v},\CF^{X_v}]$ as the product of two matrices of dimensions $|\CF^{X_v}|\times|\CF^{X_v}|$ whose entries can be computed recursively: They are either 0 or an entry in $D[w,\CF^{X_w},\CF^{X_w}]$, where $w$ is a child of $v$. The entries in $D[w,\CF^{X_w},\CF^{X_w}]$ are not computed individually, but recursively as a product of matrices.
For vertices that have only one child we use the assumption that they have the same bag, so that we can apply \eqref{eq:claimonlychild}. leaves are treated via \eqref{eq:claimleaf}.

If we can multiply two matrices of size $\leq k \times k$ in time $O(k^\omega)$, then the total running time to compute $D[r,\CF^{X_r},\CF^{X_r}]$ is $w^{\omega(\tau+1)}\poly(n,d,m, |\CT|)$. Note that $D[r,\CF^{X_r},\CF^{X_r}]$ is a $1 \times 1$ matrix whose entry is the desired $D[r,\vv{v_{\textup{source}}}, \vv{v_{\textup{sink}}}]$.
\end{proof}

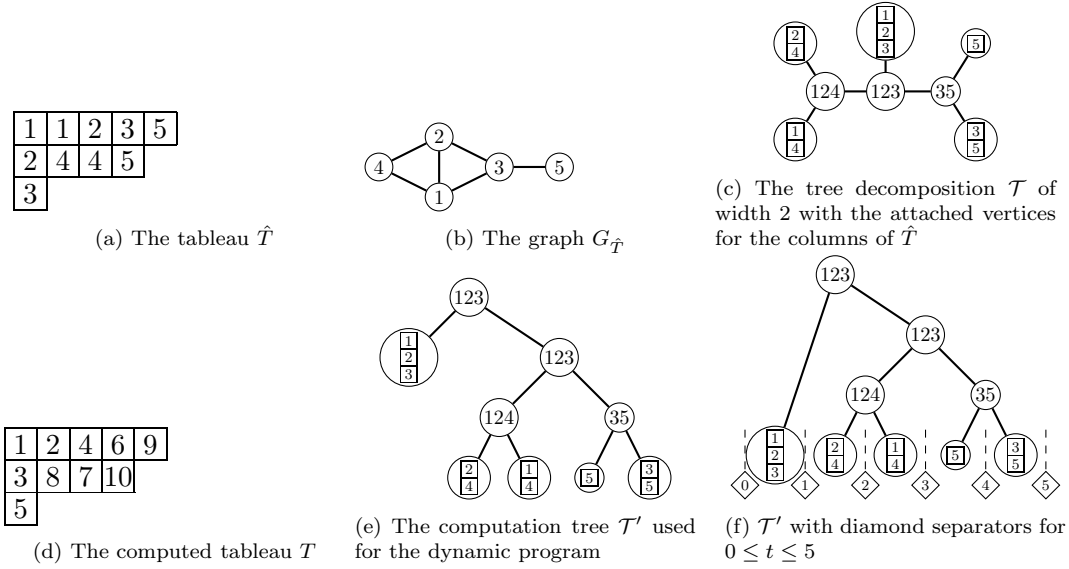
\begin{figure}[tpb]
    \centering
    \begin{subfigure}[b]{0.3\textwidth}
        \young(11235,2445,3)
        \caption{The tableau $\hat{T}$}
    \end{subfigure}
    \begin{subfigure}[b]{0.3\textwidth}
        \begin{tikzpicture}[scale=0.8,
            vertex/.style={circle, draw, scale=0.7, inner sep=1pt, minimum size=15pt}]
            \node [vertex] at (1, 1.5) (4) {$4$};
            \node [vertex] at (2, 1) (1) {$1$};
            \node [vertex] at (2, 2) (2) {$2$};
            \node [vertex] at (3, 1.5) (3) {$3$};
            \node [vertex] at (4, 1.5) (5) {$5$};

            \foreach \from/\to in {1/2, 1/3, 2/3, 2/4, 1/4, 3/5}{%
                \draw[thick] (\from) -- (\to);
            }
        \end{tikzpicture}
        \caption{The graph $G_{\hat{T}}$}
    \end{subfigure}
    \begin{subfigure}[b]{0.3\textwidth}
        \centering
        \begin{tikzpicture}[scale=0.8,
            vertex/.style={circle, draw, scale=0.7, inner sep=1pt, minimum size=15pt}]
            \node [vertex] at (1, 1) (124) {$124$};
            \node [vertex] at (2, 1) (123) {$123$};
            \node [vertex] at (3, 1) (35) {$35$};

            \node [vertex] at (2, 2) (c123) {\scriptsize $\young(1,2,3)$};
            \node [vertex] at (0.5, 0.2) (c14) {\scriptsize $\young(1,4)$};
            \node [vertex] at (0.5, 1.8) (c24) {\scriptsize $\young(2,4)$};
            \node [vertex] at (3.5, 0.2) (c35) {\scriptsize $\young(3,5)$};
            \node [vertex] at (3.5, 1.8) (c5) {\scriptsize $\young(5)$};

            \foreach \from/\to in {124/123, 123/35, 124/c14, 124/c24, 123/c123, 35/c35, 35/c5}{%
                \draw[thick] (\from) -- (\to);
            }
        \end{tikzpicture}
        \caption{The tree decomposition $\CT$ of width $2$ with the attached vertices for the columns of $\hat{T}$}
    \end{subfigure}
    
    \begin{subfigure}[b]{0.3\textwidth}
        \def\zehn{10}
        \young(12469,387\zehn,5)
        \caption{The computed tableau $T$}
    \end{subfigure}
    \begin{subfigure}[b]{0.3\textwidth}
        \centering
        \begin{tikzpicture}[scale=0.8,
            vertex/.style={circle, draw, scale=0.7, inner sep=1pt, minimum size=15pt}]
            \node [vertex] at (2, 3) (1) {$123$};
            \node [vertex] at (3.5, 2) (2) {$123$};
            \node [vertex] at (2.5, 1) (3) {$124$};
            \node [vertex] at (4.5, 1) (4) {$35$};

            \node [vertex] at (1, 2) (c123) {\scriptsize $\young(1,2,3)$};
            \node [vertex] at (3, 0) (c14) {\scriptsize $\young(1,4)$};
            \node [vertex] at (2, 0) (c24) {\scriptsize $\young(2,4)$};
            \node [vertex] at (5, 0) (c35) {\scriptsize $\young(3,5)$};
            \node [vertex] at (4, 0) (c5) {\scriptsize $\young(5)$};

            \foreach \from/\to in {1/c123, 1/2, 2/3, 2/4, 3/c14, 3/c24, 4/c35, 4/c5}{%
                \draw[thick] (\from) -- (\to);
            }
        \end{tikzpicture}
        \caption{The computation tree $\CT'$ used for the dynamic program}
    \end{subfigure}\quad
    \begin{subfigure}[b]{0.3\textwidth}
        \centering
        \begin{tikzpicture}[scale=0.8,
            vertex/.style={circle, draw, scale=0.7, inner sep=1pt, minimum size=15pt},separator/.style={diamond, draw, fill=white, scale=0.7, inner sep=1pt, minimum size=15pt}]
            \node [vertex] at (2, 3) (1) {$123$};
            \node [vertex] at (3.5, 2) (2) {$123$};
            \node [vertex] at (2.5, 1) (3) {$124$};
            \node [vertex] at (4.5, 1) (4) {$35$};

            \node [vertex] at (1, 0) (c123) {\scriptsize $\young(1,2,3)$};
            \node [vertex] at (3, 0) (c14) {\scriptsize $\young(1,4)$};
            \node [vertex] at (2, 0) (c24) {\scriptsize $\young(2,4)$};
            \node [vertex] at (5, 0) (c35) {\scriptsize $\young(3,5)$};
            \node [vertex] at (4, 0) (c5) {\scriptsize $\young(5)$};

            \draw [dashed] (0.5, -0.5) -- (0.5, 0.55);
            \draw [dashed] (1.5, -0.5) -- (1.5, 0.55);
            \draw [dashed] (2.5, -0.5) -- (2.5, 0.55);
            \draw [dashed] (3.5, -0.5) -- (3.5, 0.55);
            \draw [dashed] (4.5, -0.5) -- (4.5, 0.55);
            \draw [dashed] (5.5, -0.5) -- (5.5, 0.55);
            \node [separator] at (0.5, -0.5) {\scriptsize $0$};
            \node [separator] at (1.5, -0.5) {\scriptsize $1$};
            \node [separator] at (2.5, -0.5) {\scriptsize $2$};
            \node [separator] at (3.5, -0.5) {\scriptsize $3$};
            \node [separator] at (4.5, -0.5) {\scriptsize $4$};
            \node [separator] at (5.5, -0.5) {\scriptsize $5$};
            
            \foreach \from/\to in {1/c123, 1/2, 2/3, 2/4, 3/c14, 3/c24, 4/c35, 4/c5}{%
                \draw[thick] (\from) -- (\to);
            }
        \end{tikzpicture}
        \caption{$\CT'$ with diamond separators for $0 \leq t \leq 5$}
    \end{subfigure}

    \caption{An example execution of the preparation of the algorithm of Theorem~\ref{thm:evaltreewidth} for the Young tableau~$\hat{T}$.}
    \label{fig:treewidthevalexample}
\end{figure}

\begin{remark}
    Even though only the size of the largest bag of the tree decomposition influences the asymptotic running time, it is advisable for an actual implementation of this algorithm to minimize the size of the individual bags.
    This can be achieved by removing each number from any bag which is not on a direct path between columns that contain that particular number or even splitting bags in some cases.
\end{remark}

This dependency on the treewidth instead of $n$ is significant, since for example the graphs of semistandard Young tableaux with only two rows are planar and thus have a treewidth of $O(\sqrt{n})$.
Additionally this dependency is tight: we can construct semistandard Young tableaux with two rows and rectangular content which induce multigraph versions of the $n \times n$ grid-graphs and thus have treewidth $\Omega(\sqrt{n})$.
We prove both these observations in Proposition~\ref{prop:semistdtreewidth}.

Let $p$ be given as a Waring rank decomposition of rank $r$.
From this we can easily construct an ncABP of width $w=r$, in the same way we did to prove Theorem~\ref{thm:oursaxenaduality}.
Therefore the evaluation algorithm in Theorem~\ref{thm:evaltreewidth} now takes time $O(w^{\omega(\tau+1)}) \cdot \poly(n,d,m) = O(r^{\omega(\tau+1)}) \cdot \poly(n,d,m)$.
Comparing this to the naive algorithm in Remark~\ref{rem:waringevaluation}, we get a faster evaluation in the case $\tau \in o(n)$, which for example is achieved for all semistandard tableaux with two rows which we will now prove.

As a first step, we will prove that in this case the corresponding graphs are always planar.
\begin{proposition}
    \label{prop:2rowplanar}
    Let $S$ be a semistandard Young tableau with two rows.
    Then $G_S$ is planar.
\end{proposition}
\begin{proof}
    Let $S$ contain the numbers $\{1, \ldots, n\}$.
    We first start by constructing a different graph $G'_S = (L_S \dot{\cup} R_S, E'_S)$ which is a bipartite graph consisting of two copies of vertices $L_S = \{1_L, \ldots, n_L\}$, $R_S = \{1_R, \ldots, n_R\}$.
    Now $\{i_L, j_R\} \in E'_S$ iff {\scriptsize $\young(i,j)$} is a column in $S$.
    Here the vertical order in $S$ matters, so due to $S$ being semistandard we know $i < j$.
    Ordering the vertices, s.t.\ the vertices on the left and those on the right are each ordered in ascending order, we will now prove that $G'_S$ is outerplanar and can be drawn with straight lines.
    So let $\{i,j\}, \{k,l\} \in E'_S$ be two different edges where the column {\scriptsize $\young(i,j)$} appears to the left of the column {\scriptsize $\young(k,l)$} in $T$.
    Due to $T$ being semistandard this implies $i \leq k$ and $j \leq l$, which means those two edges do not cross.
    Since the edges were arbitrary no two edges intersect and $G'_S$ is outerplanar.

    Because both sets of vertices are ordered in ascending order we can now continuously rotate both vertex sets by 180 degrees and move them on top of each other, in this way unifying both copies of each vertex while still keeping the graph planar (the edges are not straight lines anymore, but they have the shape of a spiral).
    This resulting graph is precisely $G_S$, thus proving the claim.
\end{proof}

Now we can commence to prove the upper bound on the treewidth of Young tableaux with two rows.
Additionally we prove that this bound is tight.

\begin{proposition}
    \label{prop:semistdtreewidth}
    \begin{enumerate}
        \item
            Let $S_n$ be a semistandard Young tableau with two rows containing the numbers $\{1, \ldots, n\}$.
            Then $G_{S_n}$ has treewidth at most $O(\sqrt{n})$.
        \item
            Additionally there is a family $(S'_n)$ of semistandard Young tableaux with two rows containing the numbers $\{1, \ldots, n\}$ exactly $4$ times each and $G_{S'_n}$ having treewidth $\Omega(\sqrt{n})$.
    \end{enumerate}
\end{proposition}
\begin{proof}
    Let $S_n$ be a semistandard Young tableau with $2$ rows containing the numbers $\{1, \ldots, n\}$.
    Then $G_{S_n}$ is a planar graph with $n$ vertices by Proposition~\ref{prop:2rowplanar}.
    The fact that planar graphs on $n$ vertices have treewidth bounded by $O(\sqrt{n})$ follows directly from the famous planar excluded grid theorem \cite{DBLP:journals/jct/RobertsonST94}.

    W.l.o.g.\ we can restrict $n$ to be of the form $(2k)^2$ with $k \in \IN \setminus \{0\}$, since we can always extend the tableau without increasing the treewidth by appending four columns containing only a single cell with the number $i+1$ to the end of $S'_i$ to get $S'_{i+1}$.
    This change corresponds to adding a new isolated vertex to $G_{S'_i}$.
    We repeat this until $N=(2k)^2$, which scales $n$ up by at most a factor of 8.

    Every layered multigraph $G = (V, E)$ with the following properties is the graph $G_S$ corresponding to some semistandard tableaux $S$ where each number $i$ appears exactly as often as the degree of $i$ in $G$:
    \begin{enumerate}
        \item $V = \{1, \ldots, n\}$
        \item Edges in $G$ only go from one layer to the next.
        \item Edges between any two layers can be drawn with straight lines without crossing when the vertices in each layer are placed in ascending order.
        \item All vertices in any layer $j$ are labeled smaller than those in layer $j+1$ and each form a consecutive sequence of integers. \label{property:ascendinglayers}
    \end{enumerate}
    Some examples are provided in Figure~\ref{fig:regularmultigrid}.
    This can be shown constructively and separately for every pair of layers $j$ and $j+1$.
    Since the edges between two layers are not crossing, there is a unique ordering on the set of edges from left to right.
    Adding columns corresponding to the edges in exactly this order to $S$ forms exactly the wanted semistandard tableaux:
    For $\{u, v\} \in E$ we add the column {\scriptsize $\young(u,v)$} to $S$.
    Thus the entries corresponding to layer $j$ are only in the first row while those corresponding to layer $j+1$ only appear in the second row.
    Because of property~(\ref{property:ascendinglayers}) the columns of edges from layer $j$ to layer $j+1$ can directly be concatenated to the columns of edges from layer $j+1$ to layer $j+2$ without violating the property of being semistandard.
    Clearly $S$ contains each number $i$ exactly once for each incident edge of $i$ in $G$.

    We now take the $2k \times 2k$ grid $\boxplus_{2k} = (V_{2k}, E_{2k})$ where
    \begin{align*}
        V_{2k} &= \{(x,y) \mid x, y \in \{1, \ldots, 2k\}\}\\
        E_{2k} &= \{\{(x_1,y_1), (x_2,y_2)\} \mid |x_1 - x_2| + |y_1 - y_2| = 1\}
    \end{align*}
    This graph is known to have treewidth exactly $2k$ \cite{cygan2015parameterized}.
    We now create a multigraph by doubling all the edges $\{(1, 2i-1), (1, 2i)\}, \{(2k, 2i-1), (2k, 2i)\}, \{(2i-1, 1), (2i, 1)\}$ and $\{(2i-1, 2k), (2i, 2k)\}$ for every $i \in \{1, \ldots, k\}$ which results in each vertex having degree exactly $4$ while not changing the treewidth.
    To now apply the previous observations we now treat each diagonal $\{(x,y) \mid x+y = j+1\}$ as layer $j$ and label them by increasing $x$, thus proving the claim of the lower bound.
    The resulting graphs are also visualized in Figure~\ref{fig:regularmultigrid}.
\end{proof}
\begin{figure}[tpb]
    \centering
    \begin{tikzpicture}[scale=0.9]
        \node [circle, draw] at (1, -1) (1) {$\phantom{0}1$};

        \node [circle, draw] at (1, -2) (2) {$\phantom{0}2$};
        \node [circle, draw] at (2, -1) (3) {$\phantom{0}3$};

        \node [circle, draw] at (2, -2) (4) {$\phantom{0}4$};

        \foreach \from/\to in {1/2, 1/3, 2/4, 3/4}{%
            \draw[thick, double, double distance=2pt] (\from) -- (\to);
        }
    \end{tikzpicture}
    \hskip 20pt
    \begin{tikzpicture}[scale=0.9]
        \node [circle, draw] at (1, -1) (1) {$\phantom{0}1$};

        \node [circle, draw] at (1, -2) (2) {$\phantom{0}2$};
        \node [circle, draw] at (2, -1) (3) {$\phantom{0}3$};

        \node [circle, draw] at (1, -3) (4) {$\phantom{0}4$};
        \node [circle, draw] at (2, -2) (5) {$\phantom{0}5$};
        \node [circle, draw] at (3, -1) (6) {$\phantom{0}6$};

        \node [circle, draw] at (1, -4) (7) {$\phantom{0}7$};
        \node [circle, draw] at (2, -3) (8) {$\phantom{0}8$};
        \node [circle, draw] at (3, -2) (9) {$\phantom{0}9$};
        \node [circle, draw] at (4, -1) (10) {$10$};

        \node [circle, draw] at (2, -4) (11) {$11$};
        \node [circle, draw] at (3, -3) (12) {$12$};
        \node [circle, draw] at (4, -2) (13) {$13$};

        \node [circle, draw] at (3, -4) (14) {$14$};
        \node [circle, draw] at (4, -3) (15) {$15$};

        \node [circle, draw] at (4, -4) (16) {$16$};

        \foreach \from/\to in {2/4, 2/5, 3/5, 3/6, 4/8, 5/8, 5/9, 6/9, 8/11, 8/12, 9/12, 9/13, 11/14, 12/14, 12/15, 13/15}{%
            \draw[thick] (\from) -- (\to);
        }

        \foreach \from/\to in {1/2, 1/3, 4/7, 6/10, 7/11, 10/13, 14/16, 15/16}{%
            \draw[thick, double, double distance=2pt] (\from) -- (\to);
        }
    \end{tikzpicture}
    \hskip 20pt
    \begin{tikzpicture}[scale=0.9]

        \node [circle, draw] at (1, -1) (1) {$\phantom{0}1$};

        \node [circle, draw] at (1, -2) (2) {$\phantom{0}2$};
        \node [circle, draw] at (2, -1) (3) {$\phantom{0}3$};

        \node [circle, draw] at (1, -3) (4) {$\phantom{0}4$};
        \node [circle, draw] at (2, -2) (5) {$\phantom{0}5$};
        \node [circle, draw] at (3, -1) (6) {$\phantom{0}6$};

        \node [circle, draw] at (1, -4) (7) {$\phantom{0}7$};
        \node [circle, draw] at (2, -3) (8) {$\phantom{0}8$};
        \node [circle, draw] at (3, -2) (9) {$\phantom{0}9$};
        \node [circle, draw] at (4, -1) (10) {$10$};

        \node [circle, draw] at (1, -5) (11) {$11$};
        \node [circle, draw] at (2, -4) (12) {$12$};
        \node [circle, draw] at (3, -3) (13) {$13$};
        \node [circle, draw] at (4, -2) (14) {$14$};
        \node [circle, draw] at (5, -1) (15) {$15$};

        \node [circle, draw] at (1, -6) (16) {$16$};
        \node [circle, draw] at (2, -5) (17) {$17$};
        \node [circle, draw] at (3, -4) (18) {$18$};
        \node [circle, draw] at (4, -3) (19) {$19$};
        \node [circle, draw] at (5, -2) (20) {$20$};
        \node [circle, draw] at (6, -1) (21) {$21$};

        \node [circle, draw] at (2, -6) (22) {$22$};
        \node [circle, draw] at (3, -5) (23) {$23$};
        \node [circle, draw] at (4, -4) (24) {$24$};
        \node [circle, draw] at (5, -3) (25) {$25$};
        \node [circle, draw] at (6, -2) (26) {$26$};

        \node [circle, draw] at (3, -6) (27) {$27$};
        \node [circle, draw] at (4, -5) (28) {$28$};
        \node [circle, draw] at (5, -4) (29) {$29$};
        \node [circle, draw] at (6, -3) (30) {$30$};

        \node [circle, draw] at (4, -6) (31) {$31$};
        \node [circle, draw] at (5, -5) (32) {$32$};
        \node [circle, draw] at (6, -4) (33) {$33$};

        \node [circle, draw] at (5, -6) (34) {$34$};
        \node [circle, draw] at (6, -5) (35) {$35$};

        \node [circle, draw] at (6, -6) (36) {$36$};

        \foreach \from/\to in {%
            2/4, 2/5, 3/5, 3/6,
            4/8, 5/8, 5/9, 6/9, 7/11, 7/12, 8/12, 8/13, 9/13, 9/14, 10/14, 10/15,
            11/17, 12/17, 12/18, 13/18, 13/19, 14/19, 14/20, 15/20,
            17/22, 17/23, 18/23, 18/24, 19/24, 19/25, 20/25, 20/26,
            22/27, 23/27, 23/28, 24/28, 24/29, 25/29, 25/30, 26/30,
            28/31, 28/32, 29/32, 29/33,
            31/34, 32/34, 32/35, 33/35}{%
            \draw[thick] (\from) -- (\to);
        }

        \foreach \from/\to in {1/2, 1/3, 4/7, 6/10, 11/16, 15/21, 16/22, 21/26, 27/31, 30/33, 34/36, 35/36}{%
            \draw[thick, double, double distance=2pt] (\from) -- (\to);
        }
    \end{tikzpicture}
    \caption{The grid graphs $\boxplus_{2}, \boxplus_{4}$ and $\boxplus_{6}$ after doubling the correct edges around the border and relabeling the vertices. The layers in $\boxplus_{2}$ are $\{1\}$, $\{2,3\}$, $\{4\}$.
    The layers in $\boxplus_{4}$ are $\{1\}$, $\{2,3\}$, $\{4,5,6\}$, $\{7,8,9,10\}$, $\{11,12,13\}$, $\{14,15\}$, $\{16\}$.
    The layers in $\boxplus_{6}$ are $\{1\}$, $\{2,3\}$, $\{4,5,6\}$, $\{7,8,9,10\}$, $\{11,12,13,14,15\}$, $\{16,17,18,19,20,21\}$, $\{22,23,24,25,26\}$, $\{27,28,29,30\}$, $\{31,32,33\}$, $\{34,35\}$, $\{36\}$.
    }
    \label{fig:regularmultigrid}
\end{figure}
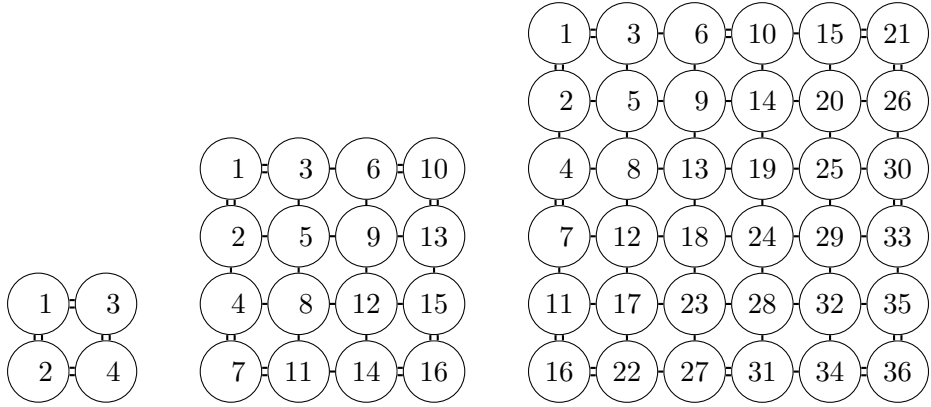

\begin{question}
    It is open whether the bound of $O(\sqrt{n})$ on the treewidth can be extended to any other constant number of rows, but starting at $3$ rows $G_S$ becomes non-planar\footnote{For example, $G_S$ is the complete graph on 5 vertices for $S={\tiny\Yvcentermath1\young(11133,22244,34555)}$.}, so another approach to solving this problem would be needed.
    Additionally, if the number of rows is arbitrary $G_S$ can contain an arbitrarily big clique, so it can have arbitrarily high treewidth.
    For example for any $S$ with a first column with $n$ distinct entries the graph $G_S$ contains a clique on $n$ vertices and thus has treewidth at least $n-1$.
\end{question}

\section{Hardness of evaluation}
\label{sec:hardness}

We will show that deciding whether a highest weight vector $f_{\hat{T}}$ of $\Sym^n\Sym^d \IC^m$ vanishes at a point of Waring rank $k$ for suiting parameters $n, d, m, k$ is \NP-hard.
In particular we prove the $\NP$-hardness of evaluating highest weight vectors given by Young tableaux with two rows in Theorem~\ref{thm:evaluationhardness}.

We can prove a similar -- slightly weaker -- result in Theorem~\ref{thm:semistdhardness}, when the tableau $\hat{T}$ is restricted to be semistandard.
In this case we have to increase the number of rows, the inner degree of the symmetric tensors and the Waring rank of the points of evaluation.
Furthermore we don't prove hardness for all constant $d$ in this case, but only for those divisible by $16$.
This still rules out polynomial evaluation algorithms which allow $d$ to be part of the input under $\PP \neq \NP$.
These reductions also yield more explicit lower bounds under the exponential time hypothesis (ETH) in Theorems~\ref{thm:evaluationhardness} and~\ref{thm:semistdhardness}.
As a reminder, the exponential time hypothesis states, that $\SAT$ can not be solved in time $2^{o(n)}$.
Finally we show in Theorem~\ref{thm:evaluationsharpphardness} that if we want to calculate the exact value of the evaluation we can even prove $\#\PP$-hardness for evaluating highest weight vectors given as Young tableaux.

Most of these reductions start with the same base that deciding whether a graph admits a proper $3$-coloring a graph is $\NP$-hard even when restricted to planar graphs of maximum degree $4$.
This was originally proven by Garey, Johnson and Stockmeyer \cite{DBLP:journals/tcs/GareyJS76} and a modified version can be found in Lemmas~\ref{lem:3colhardness} and~\ref{lem:3colgrid}.

\begin{theorem}
    \label{thm:evaluationhardness}
    Deciding whether a highest weight vector $f_{\hat{T}}$ of $\Sym^n\Sym^d \IC^m$ given as a Young tableau $\hat{T}$ evaluates to zero at a point $p \in \Sym^d \IC^m$ of Waring rank $3$ is $\NP$-hard for constant $d \geq 8, m \geq 2$.

    Assuming $\ETH$ no $2^{o(n)}$ algorithm for this evaluation can exist.
\end{theorem}
\begin{proof}
    We use the $\NP$-hardness of $3$-coloring graphs of maximum degree at most $4$, see \cite{DBLP:journals/tcs/GareyJS76} or Lemma~\ref{lem:3colhardness}.

    Let $G = (V, E)$ be a graph of maximum degree at most $4$.
    Assume w.l.o.g.\ that $V = \{1, \ldots, n\}$.
    We now construct a Young tableau $T$ with content $n \times d$ as follows:
    For every edge $\{u,v\} \in E$ we add two columns of the form {\scriptsize $\young(u,v)$} to $\hat{T}$.
    Now for every vertex $v \in V$ add $d-2\cdot\deg(v)$ single-box columns {$\young(v)$} to $\hat{T}$.
    It is easy to see that $\hat{T}$ has content $n \times d$ and is not necessarily semistandard.

    We now choose to evaluate the highest weight vector $f_{\hat{T}}$ at $p = \ell_1^d + \ell_2^d + \ell_3^d$ with $\ell_1 = (1, 0, 0, \ldots), \ell_2 = (1, 1, 0, \ldots), \ell_3 = (1, 2, 0, \ldots) \in \IC^m$.
    Note that the determinant of any two distinct linear forms of these is a real number, so its square is a positive real number.

    Recall from \eqref{eq:sumpropertheta} that
\begin{align*}
    f_{\hat{T}}(p) &= \sum_{\text{proper\ }\vartheta} \prod_{c = 1}^{\lambda_1}\det{}_{\vartheta,c}.
\end{align*}
    We now show a $1$-to-$1$ correspondence between summands of the evaluation and arbitrary -- not necessarily proper -- $3$-colorings of $G$.
    A summand will be non-zero iff the corresponding $3$-coloring is proper.
    Due to evaluating at $p$ in its Waring decomposition, $\vartheta$ will be proper iff boxes with the same number $j$ get assigned the same $\ell_i$.
    We interpret this as vertex $j$ receiving color $i$.
    Additionally every $3$-coloring of $G$ corresponds to some placement in this way.

    We now take the product of determinants for each column.
    Since each column with two boxes is repeated twice, this product is a product of squares, and hence will always be positive iff none of the determinants is zero. This idea was first used in \cite{bci:10}.
    A determinant is non-zero iff different vectors $\ell_i$ and $\ell_j$ are chosen for both of the boxes, corresponding to coloring both vertices of this column with different colors.
    So a summand will be non-zero iff $\vartheta$ corresponds to a proper $3$-coloring of $G$.

    Note that any algorithm deciding whether $f_{\hat{T}}(p)$ is non-zero in time $2^{o(n)}$ can now be used to decide whether $G$ allows for a proper $3$-coloring in time $\poly(|V|)2^{o(|V|)}$ which is a contradiction unless ETH fails as proven in \cite{DBLP:journals/jcss/ImpagliazzoPZ01}.
\end{proof}

Note that our algorithms for evaluation described in Theorems~\ref{thm:evalabp} and~\ref{thm:evaltreewidth} both achieve a running time of $2^{O(n)}$ for evaluations at points of constant Waring rank with constant $m$ and $d$.
So Theorem~\ref{thm:evaluationhardness} gives a matching lower bound under ETH.

The proof for $\#\PP$-hardness is pretty similar and reduces from counting the number of $3$-colorings of a graph with maximum vertex degree $3$ which is known to be $\#\PP$-complete \cite{DBLP:journals/siamcomp/BubleyDGJ99}.
The main idea is to use a more carefully chosen point of evaluation to ensure that every summand that corresponds to a proper $3$-coloring will be exactly $1$.

\begin{theorem}
    \label{thm:evaluationsharpphardness}
    Evaluating a highest weight vector $f_{\hat{T}}$ of $\Sym^n\Sym^d \IC^m$ given as a Young tableau $\hat{T}$ at a point $p \in \Sym^d \IC^m$ of Waring rank $3$ is $\#\PP$-hard for constant $d \geq 18, m \geq 2$.
\end{theorem}
\begin{proof}
    We reduce from counting the number of $3$-colorings of a graph $G=(V, E)$ where every vertex has degree at most $3$ which is known to be $\#\PP$-complete \cite{DBLP:journals/siamcomp/BubleyDGJ99}.
    We proceed in a similar manner as in the $\NP$-hardness proof in Theorem~\ref{thm:evaluationhardness}.
    We construct $\hat{T}$ by adding the columns {\scriptsize $\young(u,v)$} for $\{u,v\} \in E$ $6$-times each and for every vertex $v \in V$ add $d-6\cdot\deg(v)$ columns {\scriptsize $\young(v)$} to $\hat{T}$.
    This time we evaluate $f_{\hat{T}}$ at $p = \ell_1^d + \ell_2^d + \ell_3^d$ with $\ell_1 = (1, 0, 0, \ldots), \ell_2 = (1, e^{\frac{i\pi}{3}}, 0, \ldots), \ell_3 = (1, e^{\frac{2i\pi}{3}}, 0, \ldots) \in \IC^m$.
    Note that the determinant of any two distinct linear forms of these is a $6$-th root of unity, so its $6$-th power is always exactly $1$.
    If we now analyse the summands of the evaluation again we see that each term contributes exactly $1$ if it corresponds to a proper $3$-coloring and $0$ otherwise.
    Thus the evaluation $f_{\hat{T}}(p)$ counts exactly the number of $3$-colorings of $G$.
\end{proof}

Extending this result to semistandard Young tableaux now proceeds in multiple steps, which we devote the rest of this section towards.

We first extend the $\NP$-hardness of $3$-coloring to a subclass of planar graphs which we call \emph{grid-like layered} graphs.
More specifically we prove $\NP$-hardness for $8$-regular, i.e.\ each vertex has degree exactly $8$, grid-like layered graphs in Lemma~\ref{lem:gridlikehardnessregular}, while we show a lower bound of $2^{o\left(\sqrt{|V|}\right)}$ under ETH using Lemma~\ref{lem:gridlikeethhardness}.
\enlargethispage{2ex}
\begin{definition}
    \label{def:gridlike}
    We call a planar multigraph $G = (V, E)$ \emph{grid-like layered} if there are disjoint layers $L_1, \ldots, L_k \subseteq V$ of vertices and an embedding $e: V \to \IN \times \{1,\ldots,k\}$, s.t.
    \begin{enumerate}
        \item $e$ is injective.
        \item For every $i \in \{1, \ldots, k\}$ we have $e^{-1}(\IN \times \{i\}) = L_i$
        \item Edges between layers only exist between layer $L_i$ and $L_{i+1}$ for all $i \in \{1, \ldots, k-1\}$.
        \item Edges inside layers only exist for vertices $v, u \in L_i$ where $e(v) = e(u) \pm (1,0)$ for some $i \in \{1, \ldots, k\}$.\label{itm:gridlikeeasylayer}
        \item All edges can be drawn as straight lines without crossing when vertices are placed according to $e$ in $\IR^2$ and the graph is treated as being simple.\label{itm:gridlikestraightplanar}
        \item Every vertex has a neighbour in a different layer.\label{itm:gridlikeneighbours}
    \end{enumerate}
\end{definition}
Note that grid-like layered graphs are not necessarily subgraphs of a grid-graph, see Figure~\ref{fig:gridlikeexample} for an example.
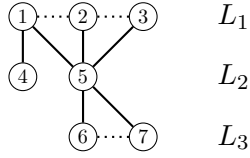
\begin{figure}[tpb]
    \centering
    \begin{tikzpicture}[scale=0.8,
        vertex/.style={circle, draw, scale=0.7, inner sep=1pt, minimum size=15pt}]
        \node [vertex] at (1, -1) (1) {$1$};
        \node [vertex] at (2, -1) (2) {$2$};
        \node [vertex] at (3, -1) (3) {$3$};
        \node [vertex] at (1, -2) (4) {$4$};
        \node [vertex] at (2, -2) (5) {$5$};
        \node [vertex] at (2, -3) (6) {$6$};
        \node [vertex] at (3, -3) (7) {$7$};
        \node [] at (4.5, -1) (L1) {$L_1$};
        \node [] at (4.5, -2) (L2) {$L_2$};
        \node [] at (4.5, -3) (L3) {$L_3$};

        \foreach \from/\to in {1/4, 1/5, 2/5, 3/5, 5/6, 5/7}{%
            \draw[thick] (\from) -- (\to);
        }
        \foreach \from/\to in {1/2, 2/3, 6/7}{%
            \draw[thick, dotted] (\from) -- (\to);
        }
    \end{tikzpicture}
    \caption{Example of a grid-like layered graph with $3$ layers. Edges between layers are drawn as solid lines and edges inside layers as dotted lines.}
    \label{fig:gridlikeexample}
\end{figure}
The crucial property about grid-like layered graphs is, that they can be decomposed into two graphs over the same vertices each corresponding to a semistandard Young tableau with two rows.
This decomposition is essential to encode the $3$-coloring of such graphs into a single combined semistandard Young tableau.
\begin{lemma}
    \label{lem:gridlikedecomp}
    Let $G = (V, E)$ be a grid-like layered graph.
    Then $G = (V, E(G_{\Thor}) \cup E(G_{\Tvert}))$ for two semistandard tableaux $\Thor, \Tvert$ for some relabeling of the vertices $V$.
    Additionally $\Tvert$ contains every number from $1$ to $|V|$ at least once.
\end{lemma}
\begin{proof}
    Let $e$ be the embedding of $G$.
    We relabel the vertices in increasing order inside each layer according to $e$ and then increasing order from layer $L_i$ to layer $L_{i+1}$ for every $i$, like in Figure~\ref{fig:gridlikeexample}.
    Let $\Evert$ now be the edges between different layers and $\Ehor$ those inside the layers, see Figure~\ref{fig:gridlikeexampledecomp}.
    Clearly $\Ehor \cup \Evert = E$ and every vertex is incident to some edge in $\Evert$ by condition~\ref{def:gridlike}.\ref{itm:gridlikeneighbours}, so if we can construct semistandard tableaux $\Thor, \Tvert$ with $\Ehor = E(G_{\Thor})$ and $\Evert = E(G_{\Tvert})$ we are done.

    We start with $\Tvert$.
    Since the labeling of the vertices is increasing from one layer to the next it suffices to show that we can create $\Tvert$ for a single pair of consecutive layers and afterwards concatenate them.
    Condition~\ref{def:gridlike}.\ref{itm:gridlikestraightplanar} gives a unique order of the edges between these layers from left to right.
    So for the edge $\{u, v\} \in \Evert$ with $u < v$ we add the column {\scriptsize $\young(u,v)$} to $\Tvert$.
    Assume two columns
    {\ytableausetup{boxsize=1.1em}\scriptsize
\ytableaushort{{u},{v}}}
    and
    {\ytableausetup{boxsize=1.1em}\scriptsize
\ytableaushort{{u'},{v'}}}
would violate the semistandard property.
    Then either $u' < u$ in which case the edge $\{u', v'\}$ would start left of $\{u, v\}$ or $v' < v$ in which case the edge $\{u', v'\}$ would end left of $\{u, v\}$, both a contradiction to our unique ordering from left to right.
    So $\Tvert$ is semistandard.

    We continue with $\Thor$.
    Again we only have to consider $\Thor$ for a single layer as we can just concatenate the resulting tableaux afterwards.
    If we direct the edges in $\Ehor$ to only go from the smaller vertex to the larger one we see with condition~\ref{def:gridlike}.\ref{itm:gridlikeeasylayer} that each vertex can only be the first vertex of an edge once, and those edges have the form $\{v, v+1\}$.
    So the only columns in $\Thor$ are of the form
{\ytableausetup{boxsize=1.6em}\scriptsize
\ytableaushort{{v},{v\!\!+\!\!1}}}.
    Those can clearly just be combined in order to make $\Thor$ semistandard.

    Note that since $G$ is a multigraph we add every column to the tableaux $k$ times if the edge appears with multiplicity $k$ in $G$.
\end{proof}
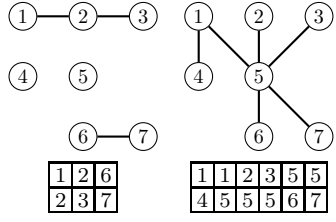
\begin{figure}[tpb]
    \centering
    $\begin{array}[]{cc}
    \begin{tikzpicture}[scale=0.8,
        vertex/.style={circle, draw, scale=0.7, inner sep=1pt, minimum size=15pt}]
        \node [vertex] at (1, -1) (1) {$1$};
        \node [vertex] at (2, -1) (2) {$2$};
        \node [vertex] at (3, -1) (3) {$3$};
        \node [vertex] at (1, -2) (4) {$4$};
        \node [vertex] at (2, -2) (5) {$5$};
        \node [vertex] at (2, -3) (6) {$6$};
        \node [vertex] at (3, -3) (7) {$7$};

        \foreach \from/\to in {1/2, 2/3, 6/7}{%
            \draw[thick] (\from) -- (\to);
        }
    \end{tikzpicture}
    &
    \begin{tikzpicture}[scale=0.8,
        vertex/.style={circle, draw, scale=0.7, inner sep=1pt, minimum size=15pt}]
        \node [vertex] at (1, -1) (1) {$1$};
        \node [vertex] at (2, -1) (2) {$2$};
        \node [vertex] at (3, -1) (3) {$3$};
        \node [vertex] at (1, -2) (4) {$4$};
        \node [vertex] at (2, -2) (5) {$5$};
        \node [vertex] at (2, -3) (6) {$6$};
        \node [vertex] at (3, -3) (7) {$7$};

        \foreach \from/\to in {1/4, 1/5, 2/5, 3/5, 5/6, 5/7}{%
            \draw[thick] (\from) -- (\to);
        }
    \end{tikzpicture}
    \\
    {\def\lr#1{\multicolumn{1}{|@{\hspace{.6ex}}c@{\hspace{.6ex}}|}{\raisebox{-.3ex}{$#1$}}}
        \scriptsize\raisebox{-.6ex}{$\begin{array}[b]{*{3}c}\cline{1-3}
                \lr{1}&\lr{2}&\lr{6}\\\cline{1-3}
                \lr{2}&\lr{3}&\lr{7}\\\cline{1-3}
        \end{array}$}
    }
    &
    {\def\lr#1{\multicolumn{1}{|@{\hspace{.6ex}}c@{\hspace{.6ex}}|}{\raisebox{-.3ex}{$#1$}}}
        \scriptsize\raisebox{-.6ex}{$\begin{array}[b]{*{6}c}\cline{1-6}
                \lr{1}&\lr{1}&\lr{2}&\lr{3}&\lr{5}&\lr{5}\\\cline{1-6}
                \lr{4}&\lr{5}&\lr{5}&\lr{5}&\lr{6}&\lr{7}\\\cline{1-6}
        \end{array}$}
    }
    \\
    \end{array}$
    \caption{Example of $\Thor$ and $\Tvert$ and the two graphs $G_\Thor$ and $G_\Tvert$ according to Lemma~\ref{lem:gridlikedecomp} for the grid-like layered graph given in Figure~\ref{fig:gridlikeexample}.}
    \label{fig:gridlikeexampledecomp}
\end{figure}

We can now give an elegant proof of the $\NP$-hardness of deciding whether a given $8$-regular grid-like layered graph $G = (V, E)$ admits a proper $3$-coloring.
With this elegance comes the caveat, that this proof only yields a lower bound of $2^{o\left(\sqrt[4]{|V|}\right)}$ under ETH, which we improve to $2^{o\left(\sqrt{|V|}\right)}$ with a more technical proof in Lemma~\ref{lem:gridlikeethhardness}.

For this we need the notion of a graph minor model.
We call a collection of subsets of vertices $(V_h)_{h \in V(G)}$ a \emph{graph minor model} of embedding a graph $G$ into a graph $H$ if the $V_h \subseteq V(H)$ induce disjoint non-empty connected subgraphs of $H$ for every $h \in V(G)$ and if for every edge $\{u, v\} \in V(G)$ there is an edge between some vertices of $V_u$ and $V_v$.
See~\cite[Section 6.3]{cygan2015parameterized} for a more detailed introduction to graph minors.

\begin{lemma}
    \label{lem:gridlikehardnessregular}
    Deciding whether a given graph $G = (V, E)$ admits a proper $3$-coloring is $\NP$-hard, even if the graph is restricted to be grid-like layered and $8$-regular.
    
    Unless ETH fails, $3$-coloring doesn't admit an $2^{o\left(\sqrt[4]{|V|}\right)}$ time algorithm for grid-like layered graphs.
\end{lemma}
\begin{proof}
    For this we reduce from the decision problem whether a planar graph $G$ admits a proper $3$-coloring.

    In order to achieve this we find a graph minor model $(V_h)_{h\in V(G)}$ of embedding $G$ into a grid $\boxplus$ with $O(|V(G)|^2)$ vertices in linear time \cite{tamassia1989planar}.
    Let $G_1$ be the grid $\boxplus$ after removing any vertices and edges which do not correspond to vertices or edges in $G$, i.e.
    \[
        V(G_1) = \bigcup_{h \in V(G)} V_h
    \]
        and
    \[
        E(G_1) = \bigcup_{h \in V(G)} E(\boxplus[V_h]) \cup \bigcup_{uv \in E(G)} E(\boxplus[V_u \cup V_v])
    \]
    where $\boxplus[V]$ denotes the subgraph of $\boxplus$ induced by the vertices $V$.

    We can now transform any $3$-coloring of $G$ into a $3$-coloring of $G_1$ by coloring every vertex in $V_h$ with the same color as $h$ for every $h \in V(G)$.
    The property of a coloring of $G$ being proper now translates to enforcing that for each $h \in V(G)$ all the vertices inside the component $V_h$ are colored with the same color and vertices in neighbouring components $V_u, V_v$ for $\{u,v\} \in E(G)$ are colored with different colors.

    \begin{figure}[tpb]
        \centering
        \begin{tikzpicture}[scale=0.8,
            vertex/.style={circle, draw, scale=0.7, inner sep=1pt, minimum size=15pt}]
            \node [vertex] at (1, 1) (1) {$v_1$};

            \node [vertex] at (2, 1) (2) {};
            \node [vertex] at (2, 2) (3) {};

            \node [vertex] at (3, 1) (4) {};

            \node [vertex] at (4, 1) (5) {};
            \node [vertex] at (4, 2) (6) {};

            \node [vertex] at (5, 1) (7) {$v_2$};

            \foreach \from/\to in {1/3, 2/3, 3/4, 4/6, 5/6, 6/7}{%
                \draw[thick] (\from) -- (\to);
            }
            \foreach \from/\to in {1/2, 2/4, 4/5, 5/7}{%
                \draw[thick, dotted] (\from) -- (\to);
            }
        \end{tikzpicture}
        \hspace{1cm}
        \begin{tikzpicture}[scale=0.8,
            vertex/.style={circle, draw, scale=0.7, inner sep=1pt, minimum size=15pt}]
            \node [vertex] at (1, 1) (1) {$v_2$};

            \node [vertex] at (1, 2) (2) {};
            \node [vertex] at (2, 2) (3) {};

            \node [vertex] at (1, 3) (4) {};

            \node [vertex] at (1, 4) (5) {};
            \node [vertex] at (2, 4) (6) {};

            \node [vertex] at (1, 5) (7) {$v_1$};

            \foreach \from/\to in {1/2, 2/4, 4/5, 5/7, 1/3, 3/4, 4/6, 6/7}{%
                \draw[thick] (\from) -- (\to);
            }
            \foreach \from/\to in {2/3, 5/6}{%
                \draw[thick, dotted] (\from) -- (\to);
            }
        \end{tikzpicture}
        \hspace{1cm}
        \begin{tikzpicture}[scale=0.8,
            vertex/.style={circle, draw, scale=0.7, inner sep=1pt, minimum size=15pt}]
            \node [vertex] at (1, 1) (1) {$v_1$};

            \node [vertex] at (2, 1) (2) {};
            \node [vertex] at (2, 2) (3) {};

            \node [vertex] at (3, 1) (4) {};
            \node [vertex] at (3, 2) (4p) {};

            \node [vertex] at (4, 1) (5) {};
            \node [vertex] at (4, 2) (6) {};

            \node [vertex] at (5, 1) (7) {$v_2$};

            \foreach \from/\to in {1/3, 2/3, 3/4, 4/4p, 4p/5, 5/6, 6/7}{%
                \draw[thick] (\from) -- (\to);
            }
            \foreach \from/\to in {1/2, 2/4, 4p/6, 5/7}{%
                \draw[thick, dotted] (\from) -- (\to);
            }
        \end{tikzpicture}
        \hspace{1cm}
        \begin{tikzpicture}[scale=0.8,
            vertex/.style={circle, draw, scale=0.7, inner sep=1pt, minimum size=15pt}]
            \node [vertex] at (1, 1) (1) {$v_2$};

            \node [vertex] at (1, 2) (2) {};
            \node [vertex] at (2, 2) (3) {};

            \node [vertex] at (1, 3) (4) {};
            \node [vertex] at (2, 3) (4p) {};

            \node [vertex] at (1, 4) (5) {};
            \node [vertex] at (2, 4) (6) {};

            \node [vertex] at (1, 5) (7) {$v_1$};

            \foreach \from/\to in {1/2, 2/4, 4p/6, 5/7, 1/3, 3/4, 4p/5, 6/7}{%
                \draw[thick] (\from) -- (\to);
            }
            \foreach \from/\to in {2/3, 4/4p, 5/6}{%
                \draw[thick, dotted] (\from) -- (\to);
            }
        \end{tikzpicture}

        \caption{The equality and inequality gadgets $H^=_1, H^=_2$ and $H^{\neq}_1, H^{\neq}_2$ used in Lemma~\ref{lem:gridlikehardnessregular}}
        \label{fig:eqneqgadget}
    \end{figure}
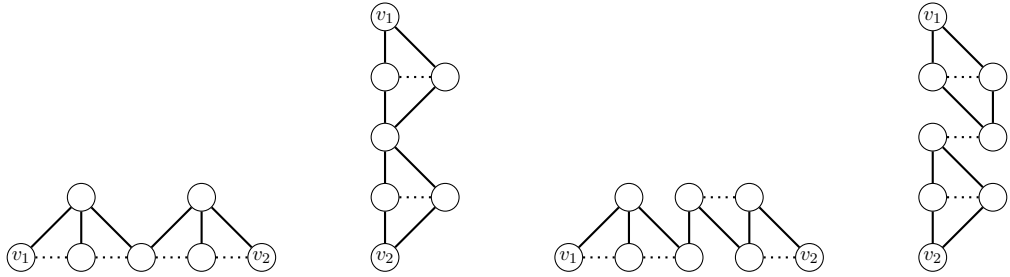

    In order to enforce these constraints on $G_1$ we construct a new graph $G_2$ by replacing each edge inside any $V_h$ by the equality gadgets $H^=_1$ or $H^=_2$ and replacing each edge between neighbouring components $V_u, V_v$ by the inequality gadgets $H^{\neq}_1$ or $H^{\neq}_2$.
    These gadgets are shown in Figure~\ref{fig:eqneqgadget}.
    If an edge is horizontal in the canonical embedding of $G_1$ into the plane we choose variant $1$ of the gadgets. If an edge is vertical we choose variant $2$.

    It can be easily checked that the only way to properly $3$-color these gadgets is such that the colors of $v_1$ and $v_2$ are the same for the equality gadgets and different for the inequality gadgets.
    
    Clearly $G$ is now properly $3$-colorable iff $G_2$ is properly $3$-colorable.

    Secondly all those gadgets are designed as grid-like layered graphs.
    It can be easily checked that replacing all edges in a subgraph of a grid yields a grid-like layered graph, so $G_2$ is grid-like layered.

    So the only thing remaining to do is make the graph $8$-regular by adding copies of existing edges to the graph.
    In order to achieve this it is sufficient to show that multigraph versions of $H^=_1, H^=_2, H^{\neq}_1$ and $H^{\neq}_2$ exist which are $8$-regular except for the vertices $v_1$ and $v_2$, which can independently have a degree of $2, 4, 6$ or $8$ each.
    This is sufficient since every vertex of the grid graph has a degree between $1$ and $4$, so it has between $1$ and $4$ of these gadgets attached to it.
    The multigraph variations of the gadgets are shown for $H^=_1$ and $H^{\neq}_1$ in Figure $\ref{fig:eqgadgetmult}$, for the other two gadgets these are constructed similarly.

    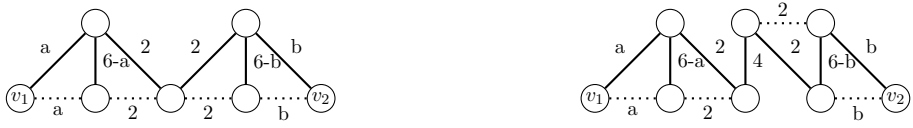
\begin{figure}[tpb]
        \centering
        \begin{tikzpicture}[scale=1,
            vertex/.style={circle, draw, scale=0.7, inner sep=1pt, minimum size=15pt},
            nl/.style={scale=0.7}]
            \node [vertex] at (1, 1) (1) {$v_1$};

            \node [vertex] at (2, 1) (3) {};
            \node [vertex] at (2, 2) (2) {};

            \node [vertex] at (3, 1) (4) {};

            \node [vertex] at (4, 1) (6) {};
            \node [vertex] at (4, 2) (5) {};

            \node [vertex] at (5, 1) (7) {$v_2$};

            \draw[thick] (1) --  node[nl, above left] {a} (2);
            \draw[thick, dotted] (1) --  node[nl, below] {a} (3);

            \draw[thick] (2) --  node[nl, right] {6-a} (3);

            \draw[thick] (2) --  node[nl, above right] {2} (4);
            \draw[thick, dotted] (3) --  node[nl, below] {2} (4);

            \draw[thick] (4) --  node[nl, above left] {2} (5);
            \draw[thick, dotted] (4) --  node[nl, below] {2} (6);

            \draw[thick] (5) --  node[nl, right] {6-b} (6);

            \draw[thick] (5) --  node[nl, above right] {b} (7);
            \draw[thick, dotted] (6) --  node[nl, below] {b} (7);
        \end{tikzpicture}
        \hspace{3cm}
        \begin{tikzpicture}[scale=1,
            vertex/.style={circle, draw, scale=0.7, inner sep=1pt, minimum size=15pt},
            nl/.style={scale=0.7}]
            \node [vertex] at (1, 1) (1) {$v_1$};

            \node [vertex] at (2, 1) (3) {};
            \node [vertex] at (2, 2) (2) {};

            \node [vertex] at (3, 1) (4) {};
            \node [vertex] at (3, 2) (4p) {};

            \node [vertex] at (4, 1) (6) {};
            \node [vertex] at (4, 2) (5) {};

            \node [vertex] at (5, 1) (7) {$v_2$};

            \draw[thick] (1) --  node[nl, above left] {a} (2);
            \draw[thick, dotted] (1) --  node[nl, below] {a} (3);

            \draw[thick] (2) --  node[nl, right] {6-a} (3);

            \draw[thick] (2) --  node[nl, above right] {2} (4);
            \draw[thick, dotted] (3) --  node[nl, below] {2} (4);

            \draw[thick] (4) --  node[nl, right] {4} (4p);

            \draw[thick, dotted] (4p) --  node[nl, above] {2} (5);
            \draw[thick] (4p) --  node[nl, above right] {2} (6);

            \draw[thick] (5) --  node[nl, right] {6-b} (6);

            \draw[thick] (5) --  node[nl, above right] {b} (7);
            \draw[thick, dotted] (6) --  node[nl, below] {b} (7);
        \end{tikzpicture}
        \caption{The nearly $8$-regular versions of the gadgets $H^=_1$ and $H^{\neq}_1$ used in Lemma~\ref{lem:gridlikehardnessregular}.
        The edge labels denote the multiplicity of the edges in the multigraph and $2a = \deg(v_1)$ and $2b = \deg(v_2)$.}
        \label{fig:eqgadgetmult}
    \end{figure}

    Note that $\boxplus$ has $O(|V(G)|^2)$ many vertices so we can conclude that $G_2$ also has $O(|V(G)|^2)$ many vertices.
    If we can decide whether the $8$-regular grid-like layered graph $G_2$ allows for a proper $3$-coloring in time $2^{o\left(\sqrt[4]{|V(G_2)|}\right)}$ we can decide via this reduction whether the planar graph $G$ allows for a proper $3$-coloring in time $\poly(|V(G)|) \cdot 2^{o\left(\sqrt{|V(G)|}\right)}$.
    This contradicts that planar $3$-coloring is not solveable in time $2^{o\left(\sqrt{|V(G)|}\right)}$ unless ETH fails which was essentially observed by Cai and Juedes \cite{cai2001subexponential} and is also mentioned in \cite[Theorem 14.9]{cygan2015parameterized}.
\end{proof}

Looking at the reduction from $3$-satisfiability to $3$-coloring more closely we can improve the ETH lower bound of the previous proof.
The fourth root was necessary because first embedding the $\SAT$ formula into a planar graph and then into a grid graph each resulted in quadratic blow-up.
By abusing the structure of the intermediate graphs more closely we reduce the size of the grid graph to be only quadratic in the number of variables of the $\SAT$ formula and thus show a better lower bound in Lemma~\ref{lem:gridlikeethhardness}.

The proof uses similar gadgets to the standard textbook reduction of $\SAT$ to $3$-coloring, which we show again for reference.
\begin{lemma}
    \label{lem:3colhardness}
    Deciding whether a given graph $G = (V, E)$ admits a proper $3$-coloring is $\NP$-hard.
\end{lemma}
\begin{proof}
    We reduce from $3$-satisfiability.
    So let $\phi = C_1 \land \ldots \land C_m$ be a formula in $3$-CNF on $n$ variables $x_1, \ldots, x_n$.
    We construct a graph $G$ as follows. We start with the graph $H_1$ shown in Figure~\ref{fig:3colgadgets} (left) and call the three vertices $\top$, $\bot$, and $z$.
    Then for each $1 \leq i \leq n$ we add a vertex $x_i$ and a vertex $\overline{x_i}$ and add three edges: $\{x_i,\overline{x_i}\}$, $\{x_i,z\}$, $\{\overline{x_i},z\}$. This is depicted in Figure~\ref{fig:3colgadgets} (middle).
    For each $1 \leq j \leq m$ we now add 6 vertices and connect them with the existing vertices as shown in Figure~\ref{fig:3colgadgets} (right): The vertices labeled $l_1$, $l_2$, $l_3$ in the figure stand for the vertices corresponding to the three literals (elements in $\{x_1,\ldots,x_n,\overline{x_1},\ldots,\overline{x_n}\}$) in the clause $C_j$.
    
    We now analyze potential proper 3-colorings of $G$. Our colors will conveniently be called $\top$, $\bot$, and $z$ and we assume from now on w.l.o.g.\ that the three vertices in $H_1$ are colored according to their names. It follows from $H_2$ that in every proper 3-coloring the vertices corresponding to literals are colored with $\top$ or $\bot$, but never with $z$.
    It is easy to see that $H_3$ has no proper 3-coloring if $l_1$, $l_2$, and $l_3$ all are colored with $\bot$.
    Moreover, if at least one of $l_1$, $l_2$, and $l_3$ is colored with $\top$ and the others are colored with $\bot$, then a proper 3-coloring of $H_3$ exists.
    
    Hence from a proper 3-coloring of $G$ we can easily reconstruct a satisfying assignment of $\phi$ and vice versa.
    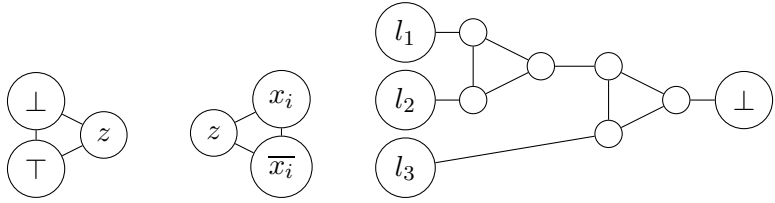
\begin{figure}[tpb]
        \centering
        \begin{tikzpicture}[scale=0.9]
            \node [circle, draw] at (1, -1) (1) {$\bot$};
            \node [circle, draw] at (1, -2) (2) {$\top$};
            \node [circle, draw] at (2, -1.5) (3) {$z$};

            \foreach \from/\to in {1/2, 2/3, 1/3}{%
                \draw (\from) -- (\to);
            }
        \end{tikzpicture}
        \hskip 20pt
        \begin{tikzpicture}[scale=0.9]
            \node [circle, draw] at (1, -1) (1) {$z$};
            \node [circle, draw] at (2, -0.5) (2) {$x_i$};
            \node [circle, draw] at (2, -1.5) (3) {$\overline{x_i}$};

            \foreach \from/\to in {1/2, 2/3, 1/3}{%
                \draw (\from) -- (\to);
            }
        \end{tikzpicture}
        \hskip 20pt
        \begin{tikzpicture}[scale=0.9]
            \node [circle, draw] at (1, -1) (1) {$l_1$};
            \node [circle, draw] at (1, -2) (2) {$l_2$};
            \node [circle, draw] at (1, -3) (3) {$l_3$};
            \node [circle, draw] at (2, -1) (4) {};
            \node [circle, draw] at (2, -2) (5) {};
            \node [circle, draw] at (3, -1.5) (6) {};
            \node [circle, draw] at (4, -1.5) (7) {};
            \node [circle, draw] at (4, -2.5) (8) {};
            \node [circle, draw] at (5, -2) (9) {};
            \node [circle, draw] at (6, -2) (10) {$\bot$};

            \foreach \from/\to in {1/4, 2/5, 3/8, 4/5, 4/6, 5/6, 6/7, 7/8, 7/9, 8/9, 9/10}{%
                \draw (\from) -- (\to);
            }
        \end{tikzpicture}
        \caption{The gadgets $H_1$, $H_2$ and $H_3$ used in the proof of Lemma~\ref{lem:3colhardness}.}
        \label{fig:3colgadgets}
    \end{figure}
\end{proof}

In Lemma~\ref{lem:gridlikehardnessregular} we then proceeded with a planar version of this theorem due to \cite{DBLP:journals/tcs/GareyJS76} and embedded these resulting graphs as minors of a grid.
In essence we used a variant of $3$-coloring where the graph is a subset of a grid graph and every edge can either be an equality or inequality edge, i.e.\ vertices connected by an equality edge have to be colored by the same color and vertices connected by an inequality edge have to be colored with different colors.
We already implicitly showed $\NP$-hardness of this variant which we call \emph{relational $3$-coloring on subgraphs of grids} in the proof of Lemma~\ref{lem:gridlikehardnessregular}.

Note that equality edges are a necessity, since any subgraph of a grid graph is bipartite and thus can be $2$-colored.

\begin{lemma}
    \label{lem:3colgrid}
    Unless ETH fails, relational $3$-coloring on subgraphs $G$ of grids can not be solved in time $2^{o\left(\sqrt{|V(G)|}\right)}$.
\end{lemma}
\begin{proof}
    We reduce from $3$-satisfiability.
    So let $\phi = C_1 \land \ldots \land C_m$ be a formula in $3$-CNF on $n$ variables $x_1, \ldots, x_n$.

    We again start with the color choosing gadget $H_1'$ from Figure~\ref{fig:3colgridgadgets} and assume that each vertex of $H_1'$ is colored with its label to simplify the analysis.
    Note that vertices with the same labels will be connected by a path of equality edges, so they have the same color in each proper $3$-coloring.
    $H_1'$ forms a border of width $\leq 2$ around the rest of the graph.

    Connected to the vertices labeled with $z$ are the variable gadgets $H_2'$.
    The vertices with labels $x_i$ and $\overline{x_i}$ corresponding to the literals of $\phi$ appear exactly as often as each of the literals appears in $\phi$.
    Connected to the bottom vertices $\bot$ are the clause gadgets $H_3'$.
    Both of these gadgets can be found in Figure~\ref{fig:3colgridgadgets}.

    The only thing left is connecting the vertices corresponding to literals in the clause gadgets to those in the variable gadgets via an equality edge.
    Unfortunately this would make the graph not be a subgraph of a grid, so we need the crossing gadget $H_4'$ from Figure~\ref{fig:3colgridcrossing} which is an embedding of the crossing gadget used in \cite{DBLP:journals/tcs/GareyJS76}.
    In $H_4'$ the vertices pairs labeled $a, a'$ and $b, b'$ each have the same color in every proper $3$-coloring.
    Additionally there is a proper $3$-coloring for every choice of colors of $a$ and $b$.

    We now need to ``sort'' the vertices corresponding to literals into the order $(l_{1,1}, l_{1,2}, l_{1,3}, l_{2,1}, l_{2,2}, l_{2,3}, \ldots, l_{m,1}, l_{m,2}, l_{m,3})$ where $C_i = l_{i, 1} \lor l_{i, 2} \lor l_{i, 3}$.
    We do this via an iterative procedure.
    We add the crossing gadget $H_4'$ between every two consecutive vertices $l_i, l_j$ which are in the wrong order in each step.
    In case some vertices could be part of multiple swaps choose the pairs in a way that maximizes the number of possible swaps per iteration.
    We connect $l_i$ to the vertex $a$ of $H_4'$ via an equality edge and similarly $l_j$ to $b$.
    The vertices $a'$ and $b'$ now form the next step in the ordering process and have essentially swapped adjacent $l_i$ and $l_j$.
    All the vertices $l_i$ which do not change their position will be extended via paths made out of equality edges to be on the same layer as the outlets of the crossing gadgets.
    After at most $O(m)$ of these steps the vertices are sorted in our desired order and can be directly connected to the corresponding vertices of the clause gadgets.

    We call this resulting graph $G$.
    Note that $H_4'$ enforces a finer subdivision of the grid than $H_1', H_2'$ and $H_3'$, but we can always split an equality edge into two equality edges connected by a vertex or split vertices into two vertices connected by an equality edge to stretch these gadgets, so $G$ is a subgraph of a grid graph.

    It can be easily checked that $H_1', H_2'$ and $H_3'$ behave in exactly the same way as their counterparts in the proof of Lemma~\ref{lem:3colhardness}, so the correctness of this reduction can easily be seen with the same reasoning as there together with the properties of $H_4'$.

    $G$ is a subgraph of an $O(m) \times O(m)$ grid, so $|V(G)| = O(m^2)$.
    If we could decide relational $3$-colorability on subgraphs of grids in size $2^{o\left(\sqrt{|V(G)|}\right)}$ we could thus decide $3$-satisfiability in time $2^{o\left(m\right)}$ which is a contradiction unless ETH fails, see \cite{DBLP:journals/jcss/ImpagliazzoPZ01} for this lower bound for $3$-satisfiability.
\end{proof}

    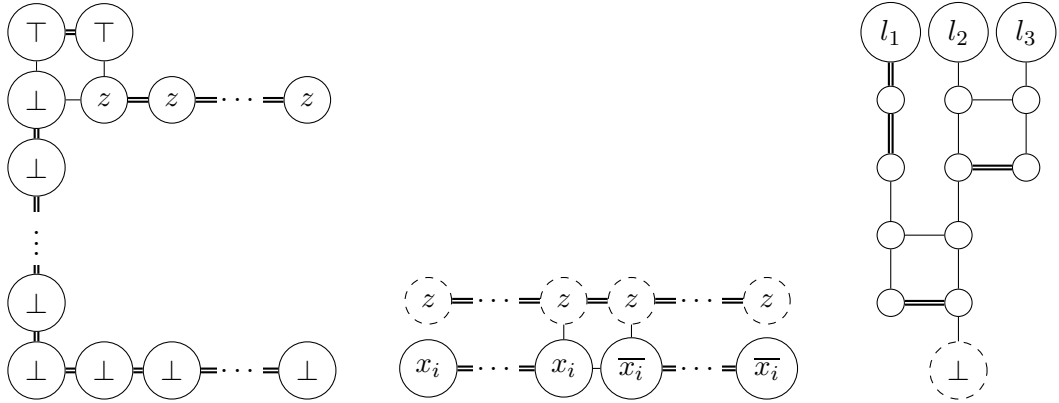
\begin{figure}[tpb]
        \centering
        \begin{tikzpicture}[scale=0.9]
            \node [circle, draw] at (0, 0) (t1) {$\top$};
            \node [circle, draw] at (1, 0) (t2) {$\top$};

            \node [circle, draw] at (1, -1) (z1) {$z$};
            \node [circle, draw] at (2, -1) (z2) {$z$};
            \node [circle, draw] at (4, -1) (z3) {$z$};

            \node [circle, draw] at (0, -1) (b1) {$\bot$};
            \node [circle, draw] at (0, -2) (b2) {$\bot$};
            \node [circle, draw] at (0, -4) (b3) {$\bot$};
            \node [circle, draw] at (0, -5) (b4) {$\bot$};
            \node [circle, draw] at (1, -5) (b5) {$\bot$};
            \node [circle, draw] at (2, -5) (b6) {$\bot$};
            \node [circle, draw] at (4, -5) (b7) {$\bot$};

            \foreach \from/\to in {t1/t2, z1/z2, b1/b2, b3/b4, b4/b5, b5/b6}{%
                \draw[thick, double] (\from) -- (\to);
            }
            \foreach \from/\to in {t1/b1, t2/z1, b1/z1}{%
                \draw (\from) -- (\to);
            }
            \foreach \from/\to in {z2/z3, b6/b7}{%
                \draw[thick, double, postaction={decorate, decoration={markings,
                        mark= at position 0.5
                              with
                              {%
                                \fill[white] (-0.3,-0.3) rectangle (0.3,0.3);
                              }
                        }
                    }] (\from) --node{\ldots} (\to);
            }
            \foreach \from/\to in {b2/b3}{%
                \draw[thick, double, postaction={decorate, decoration={markings,
                        mark= at position 0.5
                              with
                              {%
                                \fill[white] (-0.3,-0.3) rectangle (0.4,0.3);
                              }
                        }
                    }] (\from) --node{\vdots} (\to);
            }
        \end{tikzpicture}
        \hskip 20pt
        \begin{tikzpicture}[scale=0.9]
            \node [circle, draw, dashed] at (1, 0) (z1) {$z$};
            \node [circle, draw, dashed] at (3, 0) (z2) {$z$};
            \node [circle, draw, dashed] at (4, 0) (z3) {$z$};
            \node [circle, draw, dashed] at (6, 0) (z4) {$z$};

            \node [circle, draw] at (1, -1) (x1) {$x_i$};
            \node [circle, draw] at (3, -1) (x2) {$x_i$};
            \node [circle, draw] at (4, -1) (nx1) {$\overline{x_i}$};
            \node [circle, draw] at (6, -1) (nx2) {$\overline{x_i}$};

            \foreach \from/\to in {z2/z3}{%
                \draw[thick, double] (\from) -- (\to);
            }
            \foreach \from/\to in {z2/x2, x2/nx1, z3/nx1}{%
                \draw (\from) -- (\to);
            }
            \foreach \from/\to in {z1/z2, z3/z4, x1/x2, nx1/nx2}{%
                \draw[thick, double, postaction={decorate, decoration={markings,
                        mark= at position 0.5
                              with
                              {%
                                \fill[white] (-0.3,-0.3) rectangle (0.3,0.3);
                              }
                        }
                    }] (\from) --node{\ldots} (\to);
            }
        \end{tikzpicture}
        \hskip 20pt
        \begin{tikzpicture}[scale=0.9]
            \node [circle, draw] at (1, 0) (l1) {$l_1$};
            \node [circle, draw] at (2, 0) (l2) {$l_2$};
            \node [circle, draw] at (3, 0) (l3) {$l_3$};

            \node [circle, draw] at (1, -1) (a1) {};
            \node [circle, draw] at (2, -1) (a2) {};
            \node [circle, draw] at (3, -1) (a3) {};

            \node [circle, draw] at (1, -2) (b1) {};
            \node [circle, draw] at (2, -2) (b2) {};
            \node [circle, draw] at (3, -2) (b3) {};

            \node [circle, draw] at (1, -3) (c1) {};
            \node [circle, draw] at (2, -3) (c2) {};

            \node [circle, draw] at (1, -4) (d1) {};
            \node [circle, draw] at (2, -4) (d2) {};

            \node [circle, draw, dashed] at (2, -5) (bot) {$\bot$};

            \foreach \from/\to in {l1/a1, a1/b1, b2/b3, d1/d2}{%
                \draw[thick, double] (\from) -- (\to);
            }
            \foreach \from/\to in {b1/c1, c1/d1, l2/a2, a2/b2, b2/c2, c2/d2, d2/bot, l3/a3, a3/b3, a2/a3, c1/c2}{%
                \draw (\from) -- (\to);
            }
        \end{tikzpicture}
        \caption{The gadgets $H_1'$, $H_2'$ and $H_3'$ used in the proof of Lemma~\ref{lem:3colgrid}. Double lines denote equality edges and single lines denote inequality edges. Vertices corresponding to other gadgets are visualized with dashed outline to show how to connect the gadgets.}
        \label{fig:3colgridgadgets}
    \end{figure}
    \begin{figure}[tpb]
        \centering
        \begin{tikzpicture}[scale=0.9]
            \node [circle, draw] at (1, -1) (11) {$a$};
            \node [circle, draw] at (2, -1) (21) {};
            \node [circle, draw] at (3, -1) (31) {};
            \node [circle, draw] at (4, -1) (41) {};
            \node [circle, draw] at (5, -1) (51) {$b$};

            \node [circle, draw] at (1, -2) (12) {};
            \node [circle, draw] at (2, -2) (22) {};
            \node [circle, draw] at (3, -2) (32) {};
            \node [circle, draw] at (4, -2) (42) {};
            \node [circle, draw] at (5, -2) (52) {};

            \node [circle, draw] at (1, -3) (13) {};
            \node [circle, draw] at (2, -3) (23) {};
            \node [circle, draw] at (3, -3) (33) {};
            \node [circle, draw] at (4, -3) (43) {};
            \node [circle, draw] at (5, -3) (53) {};

            \node [circle, draw] at (1, -4) (14) {};
            \node [circle, draw] at (2, -4) (24) {};
            \node [circle, draw] at (3, -4) (34) {};
            \node [circle, draw] at (4, -4) (44) {};
            \node [circle, draw] at (5, -4) (54) {};

            \node [circle, draw] at (1, -5) (15) {$b'$};
            \node [circle, draw] at (2, -5) (25) {};
            \node [circle, draw] at (3, -5) (35) {};
            \node [circle, draw] at (4, -5) (45) {};
            \node [circle, draw] at (5, -5) (55) {$a'$};

            \foreach \from/\to in {11/12, 11/21, 22/23, 32/42, 41/51, 51/52, 43/44, 54/55, 55/45, 24/34, 14/15, 15/25}{%
                \draw[thick, double] (\from) -- (\to);
            }
            \foreach \from/\to in {21/31, 31/41, 12/22, 22/32, 32/31, 42/41, 12/13, 13/23, 23/33, 32/33, 33/43, 42/43, 52/53, 43/53, 13/14, 23/24, 33/34, 53/54, 34/44, 44/54, 24/25, 25/35, 35/45, 34/35}{%
                \draw (\from) -- (\to);
            }
        \end{tikzpicture}
        \caption{The gadget $H_4'$ used in the proof of Lemma~\ref{lem:3colgrid}. Double lines denote equality edges and single lines denote inequality edges.}
        \label{fig:3colgridcrossing}
    \end{figure}
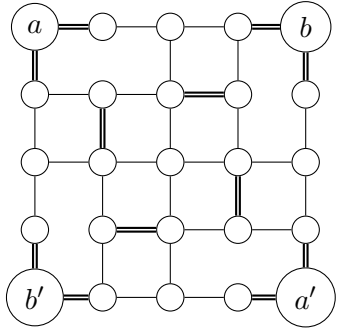

\begin{lemma}
    \label{lem:gridlikeethhardness}
    Unless ETH fails, $3$-coloring doesn't admit an $2^{o\left(\sqrt{|V|}\right)}$ time algorithm for $8$-regular grid-like layered graphs $G = (V, E)$.
\end{lemma}
\begin{proof}
    We reduce from relational $3$-coloring on subgraphs of grids.
    Let $G = (V, E)$ be a subgraph of a grid.
    We proceed in the same way as Lemma~\ref{lem:gridlikehardnessregular} did except that $\boxplus$ is replaced by $G$, each equality edge is replaced by the corresponding equality gadget and each inequality edge is replaced by the corresponding inequality gadget.
    
    Note that the obtained graph $G_2$ now has $O(|V(G)|)$ many vertices.
    If we can decide whether the $8$-regular grid-like layered graph $G_2$ allows for a proper $3$-coloring in time $2^{o\left(\sqrt{|V(G_2)|}\right)}$ we can decide via this reduction whether $G$ allows for a relational $3$-coloring in time $\poly(|V(G)|) \cdot 2^{o\left(\sqrt{|V(G)|}\right)}$, contradicting Lemma~\ref{lem:3colgrid} unless ETH fails.
\end{proof}

We now have all the necessary intermediate results to prove that even evaluation of highest weight vectors given by semistandard tableaux is $\NP$-hard.
We use the same general idea of coloring the cells of the Young tableau s.t.\ all cells with the same number receive the same color.
Additionally each column of the Young tableau has to be repeated often enough that any summands are guaranteed to be positive iff each column is colorful, i.e.\ does not contain any color multiple times and zero otherwise.

\begin{theorem}
    \label{thm:semistdhardness}
    The evaluation of highest weight vectors $f_{\hat{T}}$ of $\Sym^n\Sym^d \IC^m$ is $\NP$-hard for any constant $d \geq 16$ with $16~|~d$  and $m \geq 5$, when $f_{\hat{T}}$ is given as a semistandard Young tableau $\hat{T}$.
    This even holds if evaluation is restricted to points of Waring rank $5$ and the algorithm only has to decide whether the evaluation is non-zero.

    Additionally this evaluation can not be computed in time $2^{o(\sqrt{n})}$ unless ETH fails.
\end{theorem}
\begin{proof}
    \definecolor{symbolicA}{rgb}{0,0,1}
    \definecolor{symbolicB}{rgb}{1,0,0}
    \definecolor{symbolicC}{rgb}{0,0.5,0}
    \definecolor{symbolicD}{rgb}{0.5,0.3,0.2}
    \definecolor{symbolicE}{rgb}{0.5,0,0.5}

    We reduce from checking whether an $8$-regular grid-like layered graph allows for a proper $3$-coloring which was proven in Lemma~\ref{lem:gridlikehardnessregular} to be \NP-hard.

    Let $G = (V, \Ehor \cup \Evert)$ be an $8$-regular grid-like layered graph where $\Ehor$ denotes the edges inside layers and $\Evert$ denotes edges between layers.
    W.l.o.g.\ the vertex set $V$ are the numbers $1, \ldots, |V|$ assigned in a layer by layer and left to right fashion, given by the embedding of $G$.
    To ease the description of the constructed semistandard tableaux $\hat{T}$ we will describe it in $5$ parts $T_1, \ldots, T_5$ over the symbolic entries $\textcolor{symbolicA}{a_i}, \textcolor{symbolicB}{b_i}, \textcolor{symbolicC}{c_i}, \textcolor{symbolicD}{d_i}, \textcolor{symbolicE}{e_i}$.
    For better readability we will colorcode each of the symbolic entries in the constructions of this theorem.

    The point of evaluation is now $p = \sum_{i=1}^5 \ell_i^d$ with $\ell_i = (1, i, i^2, i^3, \ldots, i^m)$
    Then any determinants arising in the evaluation are determinants of Vandermonde matrices and thus are well known to be non-zero.
    
    $p$ is a point of Waring rank $5$, so analogously to Theorem~\ref{thm:evaluationhardness} the summands of the evaluation will consist of assigning one of the $5$ linear forms to each number and will be non-zero iff no column contains the same linear form twice.
    Since all vectors are real, any occuring determinants in the evaluation will also be real and hence every summand will be either $0$ or positive due to every column being repeated an even number of times.

    \begin{figure}[tpb]
        \centering
        \begin{tikzpicture}[scale=0.5]
            \draw (0,0) rectangle (6,5);
            \draw[dashed] (1.8,0) -- (1.8,5);
            \draw[dashed] (4.2,0) -- (4.2,5);
            \node at (0.9,2.5) {$T_{1,1}$};
            \node at (3,2.5) {$\dots$};
            \node at (5.1,2.5) {$T_{1,r}$};

            \draw (6,1) rectangle (8,5);
            \node at (7,3) {$T_2$};

            \draw (8,1) rectangle (14,5);
            \draw[dashed] (8,3) -- (14,3);
            \node at (11,4) {$T_{3,1}$};
            \node at (11,2) {$T_{3,2}$};

            \draw (14,3) rectangle (16,5);
            \node at (15,4) {$T_4$};

            \draw (16,3) rectangle (22,5);
            \node at (19,4) {$T_5$};
        \end{tikzpicture}
        \caption{The general structure of the $5$ row Young tableau $\hat{T}$ constructed in Theorem~\ref{thm:semistdhardness}}
        \label{fig:semistdhardnesstructure}
    \end{figure}
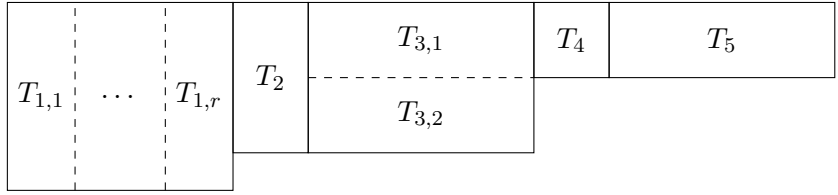

    The general structure of $\hat{T}$ can be seen in Figure~\ref{fig:semistdhardnesstructure} and we will first describe the main idea of each part.
    The parts $T_3$ and $T_5$ both encode the actual $3$-coloring restrictions of the edge sets $\Evert$ and $\Ehor$ respectively, in the entries $\textcolor{symbolicE}{e_i}$ in the same way as in Theorem~\ref{thm:evaluationhardness}.
    To ensure that only $3$ colors can be used for the graph coloring, the entries $\textcolor{symbolicC}{c_i}$ are added into $T_3$ to use up the remaining colors.
    The consistency of the $\textcolor{symbolicC}{c_i}$, i.e.\ that exactly two colors are used by all of the $\textcolor{symbolicC}{c_i}$, is then ensured in $T_1$ with the help of the entries $\textcolor{symbolicA}{a_i}$ and $\textcolor{symbolicB}{b_i}$.
    Everything else, i.e.\ the $\textcolor{symbolicD}{d_i}$ and the tableaux $T_2$ and $T_4$ are only used in order to make $\hat{T}$ semistandard and with rectangular content.

    $\hat{T}$ is then the concatenation $T_1 T_2 \dots T_5$ where we assign increasing numbers starting from $1$ first to all the $\textcolor{symbolicA}{a_i}$, then to all the $\textcolor{symbolicB}{b_i}, \textcolor{symbolicC}{c_i}, \textcolor{symbolicD}{d_i}$ and $\textcolor{symbolicE}{e_i}$ in order, increasing inside each group of symbolic entries with increasing index.
    This ensures that each of the $T_i$ individually, but also the concatenation $\hat{T}$ will be semistandard.
    The latter can be seen by looking at the symbolic entries at the left and right borders of the $T_i$ in the following descriptions.

    We first describe the construction of all the $T_i$:
    $T_1$ is built as a concatenation of smaller tableaux $T_{1,1}, \ldots T_{1, r}$ for $r = \lceil \frac{|\Evert| - 1}{24} \rceil$.
    The construction of $T_{1,1}$ and of $T_{1,i}$ for $1 < i \leq r$ can be seen in Figure~\ref{fig:semitabconstr1}.
    $T_2$ and $T_4$ are always the same and are given in Figure~\ref{fig:semitabconstr2}.
    $T_3$ is given as a left aligned vertical concatenation of $T_{3,1}$ and $T_{3,2}$, where $T_{3, 1}$ consists of the columns
    {\scriptsize \def\lr#1{\multicolumn{1}{|@{\hspace{.6ex}}c@{\hspace{.6ex}}|}{\raisebox{-.3ex}{\textcolor{symbolicC}{$#1$}}}}
\raisebox{-.6ex}{$\begin{array}[b]{*{1}c}\cline{1-1}
        \lr{c_1}\\\cline{1-1}
        \lr{c_2}\\\cline{1-1}
\end{array}$}
},
    {\scriptsize \def\lr#1{\multicolumn{1}{|@{\hspace{.6ex}}c@{\hspace{.6ex}}|}{\raisebox{-.3ex}{\textcolor{symbolicC}{$#1$}}}}
\raisebox{-.6ex}{$\begin{array}[b]{*{1}c}\cline{1-1}
        \lr{c_3}\\\cline{1-1}
        \lr{c_4}\\\cline{1-1}
\end{array}$}
},
$\ldots$,
    {\scriptsize \def\lr#1{\multicolumn{1}{|@{\hspace{.6ex}}c@{\hspace{.6ex}}|}{\raisebox{-.3ex}{\textcolor{symbolicC}{$#1$}}}}
\raisebox{-.6ex}{$\begin{array}[b]{*{1}c}\cline{1-1}
        \lr{c_{8r-3}}\\\cline{1-1}
        \lr{c_{8r-2}}\\\cline{1-1}
\end{array}$}
} each repeated $12$ times and
    {\scriptsize \def\lr#1{\multicolumn{1}{|@{\hspace{.6ex}}c@{\hspace{.6ex}}|}{\raisebox{-.3ex}{\textcolor{symbolicC}{$#1$}}}}
\raisebox{-.6ex}{$\begin{array}[b]{*{1}c}\cline{1-1}
        \lr{c_{8r-1}}\\\cline{1-1}
        \lr{c_{8r}}\\\cline{1-1}
\end{array}$}
} repeated $14$ times.
    $T_{3,2}$ is obtained from $\Tvert$ of Lemma~\ref{lem:gridlikedecomp} by doubling every column.
    Lastly $T_5$ is constructed in the exact same way as $T_{3,2}$, but is obtained from $\Thor$ of Lemma~\ref{lem:gridlikedecomp}.
    
    We first prove that $T$ is a semistandard Young tableau of rectangular content.
    $T_1$ fulfils the following properties which are easy to prove via induction:
    \begin{itemize}
        \item $\textcolor{symbolicA}{a_1}, \ldots, \textcolor{symbolicA}{a_{3r}}$ all appear exactly $16$ times each in $T_1$.
        \item $\textcolor{symbolicB}{b_1}$ and $\textcolor{symbolicB}{b_2}$ appear exactly twice in $T_1$.
        \item $\textcolor{symbolicC}{c_1}, \ldots, \textcolor{symbolicC}{c_{8r-2}}$ all appear exactly $4$ times in $T_1$.
        \item $\textcolor{symbolicC}{c_{8r-1}}$ and $\textcolor{symbolicC}{c_{8r}}$ appear exactly twice in $T_1$.
        \item If we replace the symbolic entries as previously described then $T_1$ is semistandard.
    \end{itemize}
    
    The only important properties of $T_2$ and $T_4$ are that $\textcolor{symbolicB}{b_1}, \textcolor{symbolicB}{b_2}, \textcolor{symbolicD}{d_1}, \textcolor{symbolicD}{d_2}$ all appear exactly $14$ times in $T_2$ and $\textcolor{symbolicD}{d_1}$ and $\textcolor{symbolicD}{d_2}$ each appear twice in $T_4$, while both are clearly semistandard.
    
    The properties of $T_3$ are now:
    \begin{itemize}
        \item $\textcolor{symbolicC}{c_1}, \ldots, \textcolor{symbolicC}{c_{8r-2}}$ all appear $12$ times in $T_3$.
        \item $\textcolor{symbolicC}{c_{8r-1}}$ and $\textcolor{symbolicC}{c_{8r}}$ appear exactly $14$ times in $T_3$.
        \item If we replace the symbolic entries as previously described, $T_3$ is semistandard.
        \item $T_{3,1}$ has at least as many columns as $T_{3,2}$ by our choice of $r = \lceil \frac{|\Evert| - 1}{24} \rceil$.
            $T_{3,1}$ has \[
                (4r-1) \cdot 12 + 14 \geq \left(\frac{|\Evert|-1}{6}-1\right)\cdot 12 + 14 = 2 \cdot |\Evert|
            \]
            columns while $T_{3,2}$ has exactly $2 \cdot |\Evert|$ columns.
    \end{itemize}
    The last property of $T_3$ is important in order for $T_3$ and thus $\hat{T}$ to be a of proper shape for a Young tableau, i.e.\ have non-decreasing row lengths.
    
    Combining all the properties we see that $\hat{T}$ contains every entry exactly $16$ times each and is semistandard after replacing the symbolic entries.
    Additionally each column is repeated an even number of times, so no summands of the evaluation can be negative.
    In case $d > 16$ we repeat every column of $T$ $\frac{d}{16}$ times in order to get the representation of a highest weight vector of $\Sym^n\Sym^d \IC^m$ as a semistandard Young tableau.
    
    Next we look at the effects of the gadgets on the possible non-zero summands of the evaluation.
    
    Any further considerations will now assume w.l.o.g.\ that $\textcolor{symbolicA}{a_1}, \textcolor{symbolicA}{a_2}, \textcolor{symbolicA}{a_3}$ get assigned the first three linear forms of $p$, all other cases are symmetric.
    These three entries all occur together in the very first column of $T_1$, so they have to be pairwise different in order to be part of a non-zero summand.
    $T_{1,1}$ then enforces $\textcolor{symbolicC}{c_1}, \ldots, \textcolor{symbolicC}{c_8}$ to all be assigned the last two linear forms of $p$.
    Since $T_{1, i}$ and $T_{1, i+1}$ share the entries of $\textcolor{symbolicC}{c_{8i-1}}$ and $\textcolor{symbolicC}{c_{8i}}$, inductively all of $\textcolor{symbolicA}{a_i}, \ldots, \textcolor{symbolicA}{a_{3r}}$ will be assigned the first three linear forms of $p$ in some order and all of $\textcolor{symbolicC}{c_1}, \ldots, \textcolor{symbolicC}{c_{8r}}$ will be assigned the last two linear forms of $p$ in some order.

    The last important property is, that $\textcolor{symbolicE}{e_1}, \ldots, \textcolor{symbolicE}{e_{|V|}}$ all appear at least once in $T_3$ since every vertex of a grid-like layered graph is incident to an edge going to another layer.
    This means that all the linear forms being chosen for any $\textcolor{symbolicE}{e_1}, \ldots, \textcolor{symbolicE}{e_{|V|}}$ can only be the first three linear forms of $p$ since the remaining two are already used for the $\textcolor{symbolicC}{c_i}$ of which two appear in every column.

    Now assume $G$ admits a proper $3$-coloring with the colors $1,2,3$.
    We can now construct a placement of the linear forms onto the entries of $\hat{T}$ as follows:
    \begin{itemize}
        \item The entries $\textcolor{symbolicA}{a_{3i+j}}$ get assigned the linear form $\ell_j$ for every $i \in \{0, \ldots, r-1\}$ and $j \in \{1, 2, 3\}$.
        \item The entries $\textcolor{symbolicB}{b_1}$ and $\textcolor{symbolicB}{b_2}$ get assigned the linear forms $\ell_4$ and $\ell_5$ respectively.
        \item The entries $\textcolor{symbolicC}{c_{2i+j}}$ get assigned the linear form $\ell_{3+j}$ for every $i \in \{0, \ldots, 4r-1\}$ and $j \in \{1, 2\}$.
        \item The entries $\textcolor{symbolicD}{d_1}$ and $\textcolor{symbolicD}{d_2}$ get assigned the linear forms $\ell_1$ and $\ell_2$ respectively.
        \item The entries $\textcolor{symbolicE}{e_i}$ get assigned the linear form $\ell_j$ if vertex $i$ was colored with color $j$ in $G$.
    \end{itemize}
    It is now easy to check that in $T_1, T_2$ and $T_4$ no column contains any linear form twice.
    To see that the same holds for $T_3$ and $T_5$ note that the only way any column could contain the same linear form twice would be for two entries $\textcolor{symbolicE}{e_u}$ and $\textcolor{symbolicE}{e_v}$ to appear in the same column and be assigned the same linear form.
    That would mean that $u$ and $v$ got colored the same in $G$, but by our construction there is also an edge $\{u, v\} \in \Evert \cup \Ehor$, a contradition to $G$ being properly $3$-colored.
    Since no column contains a repeated linear form this summand is strictly positive, making the whole evaluation $f_{\hat{T}}(p)$ non-zero.

    Conversely assume that the evaluation of $f_{\hat{T}}(p)$ is non-zero.
    Thus there must be a non-zero summand, placing linear forms on each entry.
    As by the previous discussion there are only $3$ different linear forms being placed on all of the $\textcolor{symbolicE}{e_i}$, directly inducing a $3$-coloring of $G$.
    This $3$-coloring is proper since every column can never contain the same linear form twice and every edge of $G$ is represented by a column.

    \begin{figure}[tpb]
        \centering
        {\scriptsize
            \begin{align*}T_{1,1} &= {%
                    \def\la#1{\multicolumn{1}{|@{\hspace{.6ex}}c@{\hspace{.6ex}}|}{\raisebox{-.3ex}{\textcolor{symbolicA}{$#1$}}}}
                    \def\lb#1{\multicolumn{1}{|@{\hspace{.6ex}}c@{\hspace{.6ex}}|}{\raisebox{-.3ex}{\textcolor{symbolicB}{$#1$}}}}
                    \def\lc#1{\multicolumn{1}{|@{\hspace{.6ex}}c@{\hspace{.6ex}}|}{\raisebox{-.3ex}{\textcolor{symbolicC}{$#1$}}}}
                    \raisebox{-.6ex}{$\begin{array}[b]{*{16}c}\cline{1-16}
                            \la{a_1}&\la{a_1}&\la{a_1}&\la{a_1}&\la{a_1}&\la{a_1}&\la{a_1}&\la{a_1}&\la{a_1}&\la{a_1}&\la{a_1}&\la{a_1}&\la{a_1}&\la{a_1}&\la{a_1}&\la{a_1}\\\cline{1-16}
                            \la{a_2}&\la{a_2}&\la{a_2}&\la{a_2}&\la{a_2}&\la{a_2}&\la{a_2}&\la{a_2}&\la{a_2}&\la{a_2}&\la{a_2}&\la{a_2}&\la{a_2}&\la{a_2}&\la{a_2}&\la{a_2}\\\cline{1-16}
                            \la{a_3}&\la{a_3}&\la{a_3}&\la{a_3}&\la{a_3}&\la{a_3}&\la{a_3}&\la{a_3}&\la{a_3}&\la{a_3}&\la{a_3}&\la{a_3}&\la{a_3}&\la{a_3}&\la{a_3}&\la{a_3}\\\cline{1-16}
                            \lb{b_1}&\lb{b_1}&\lc{c_1}&\lc{c_1}&\lc{c_1}&\lc{c_1}&\lc{c_3}&\lc{c_3}&\lc{c_3}&\lc{c_3}&\lc{c_5}&\lc{c_5}&\lc{c_5}&\lc{c_5}&\lc{c_7}&\lc{c_7}\\\cline{1-16}
                            \lb{b_2}&\lb{b_2}&\lc{c_2}&\lc{c_2}&\lc{c_2}&\lc{c_2}&\lc{c_4}&\lc{c_4}&\lc{c_4}&\lc{c_4}&\lc{c_6}&\lc{c_6}&\lc{c_6}&\lc{c_6}&\lc{c_8}&\lc{c_8}\\\cline{1-16}
                    \end{array}$}
                }\\
                T_{1,i} &=
                {%
                    \def\la#1{\multicolumn{1}{|@{\hspace{.4ex}}c@{\hspace{.4ex}}|}{\raisebox{-.3ex}{\textcolor{symbolicA}{$#1$}}}}
                    \def\lb#1{\multicolumn{1}{|@{\hspace{.4ex}}c@{\hspace{.4ex}}|}{\raisebox{-.3ex}{\textcolor{symbolicB}{$#1$}}}}
                    \def\lc#1{\multicolumn{1}{|@{\hspace{.4ex}}c@{\hspace{.4ex}}|}{\raisebox{-.3ex}{\textcolor{symbolicC}{$#1$}}}}
                    \raisebox{-.6ex}{$\begin{array}[b]{*{16}c}\cline{1-16}
                            \la{a_{3i-2}}&\la{a_{3i-2}}&\la{a_{3i-2}}&\la{a_{3i-2}}&\la{a_{3i-2}}&\la{a_{3i-2}}&\la{a_{3i-2}}&\la{a_{3i-2}}&\la{a_{3i-2}}&\la{a_{3i-2}}&\la{a_{3i-2}}&\la{a_{3i-2}}&\la{a_{3i-2}}&\la{a_{3i-2}}&\la{a_{3i-2}}&\la{a_{3i-2}}\\\cline{1-16}
                            \la{a_{3i-1}}&\la{a_{3i-1}}&\la{a_{3i-1}}&\la{a_{3i-1}}&\la{a_{3i-1}}&\la{a_{3i-1}}&\la{a_{3i-1}}&\la{a_{3i-1}}&\la{a_{3i-1}}&\la{a_{3i-1}}&\la{a_{3i-1}}&\la{a_{3i-1}}&\la{a_{3i-1}}&\la{a_{3i-1}}&\la{a_{3i-1}}&\la{a_{3i-1}}\\\cline{1-16}
                            \la{a_{3i}}&\la{a_{3i}}&\la{a_{3i}}&\la{a_{3i}}&\la{a_{3i}}&\la{a_{3i}}&\la{a_{3i}}&\la{a_{3i}}&\la{a_{3i}}&\la{a_{3i}}&\la{a_{3i}}&\la{a_{3i}}&\la{a_{3i}}&\la{a_{3i}}&\la{a_{3i}}&\la{a_{3i}}\\\cline{1-16}
                            \lc{c_{8i-9}}&\lc{c_{8i-9}}&\lc{c_{8i-7}}&\lc{c_{8i-7}}&\lc{c_{8i-7}}&\lc{c_{8i-7}}&\lc{c_{8i-5}}&\lc{c_{8i-5}}&\lc{c_{8i-5}}&\lc{c_{8i-5}}&\lc{c_{8i-3}}&\lc{c_{8i-3}}&\lc{c_{8i-3}}&\lc{c_{8i-3}}&\lc{c_{8i-1}}&\lc{c_{8i-1}}\\\cline{1-16}
                            \lc{c_{8i-8}}&\lc{c_{8i-8}}&\lc{c_{8i-6}}&\lc{c_{8i-6}}&\lc{c_{8i-6}}&\lc{c_{8i-6}}&\lc{c_{8i-4}}&\lc{c_{8i-4}}&\lc{c_{8i-4}}&\lc{c_{8i-4}}&\lc{c_{8i-2}}&\lc{c_{8i-2}}&\lc{c_{8i-2}}&\lc{c_{8i-2}}&\lc{c_{8i}}&\lc{c_{8i}}\\\cline{1-16}
                    \end{array}$}
                }
        \end{align*}}
        \caption{The Young tableaux $T_{1,1}$ and $T_{1,i}$ from the proof of Theorem~\ref{thm:semistdhardness}}
        \label{fig:semitabconstr1}
    \end{figure}

    \begin{figure}[tpb]
        \centering
        {\small
            \begin{align*}T_2 &=
                {%
                    \def\lb#1{\multicolumn{1}{|@{\hspace{.6ex}}c@{\hspace{.6ex}}|}{\raisebox{-.3ex}{\textcolor{symbolicB}{$#1$}}}}
                    \def\ld#1{\multicolumn{1}{|@{\hspace{.6ex}}c@{\hspace{.6ex}}|}{\raisebox{-.3ex}{\textcolor{symbolicD}{$#1$}}}}
                    \raisebox{-.6ex}{$\begin{array}[b]{*{14}c}\cline{1-14}
                            \lb{b_1}&\lb{b_1}&\lb{b_1}&\lb{b_1}&\lb{b_1}&\lb{b_1}&\lb{b_1}&\lb{b_1}&\lb{b_1}&\lb{b_1}&\lb{b_1}&\lb{b_1}&\lb{b_1}&\lb{b_1}\\\cline{1-14}
                            \lb{b_2}&\lb{b_2}&\lb{b_2}&\lb{b_2}&\lb{b_2}&\lb{b_2}&\lb{b_2}&\lb{b_2}&\lb{b_2}&\lb{b_2}&\lb{b_2}&\lb{b_2}&\lb{b_2}&\lb{b_2}\\\cline{1-14}
                            \ld{d_1}&\ld{d_1}&\ld{d_1}&\ld{d_1}&\ld{d_1}&\ld{d_1}&\ld{d_1}&\ld{d_1}&\ld{d_1}&\ld{d_1}&\ld{d_1}&\ld{d_1}&\ld{d_1}&\ld{d_1}\\\cline{1-14}
                            \ld{d_2}&\ld{d_2}&\ld{d_2}&\ld{d_2}&\ld{d_2}&\ld{d_2}&\ld{d_2}&\ld{d_2}&\ld{d_2}&\ld{d_2}&\ld{d_2}&\ld{d_2}&\ld{d_2}&\ld{d_2}\\\cline{1-14}
                    \end{array}$}
                } & T_4 &=
                {%
                    \def\ld#1{\multicolumn{1}{|@{\hspace{.6ex}}c@{\hspace{.6ex}}|}{\raisebox{-.3ex}{\textcolor{symbolicD}{$#1$}}}}
                    \raisebox{-.6ex}{$\begin{array}[b]{*{2}c}\cline{1-2}
                            \ld{d_1}&\ld{d_1}\\\cline{1-2}
                            \ld{d_2}&\ld{d_2}\\\cline{1-2}
                    \end{array}$}
                }
        \end{align*}}
        \caption{The Young tableaux $T_2$ and $T_4$ from the proof of Theorem~\ref{thm:semistdhardness}}
        \label{fig:semitabconstr2}
    \end{figure}

    To now show that this evaluation is not possible in time $2^{o\left(\sqrt{n}\right)}$ unless ETH fails, notice that if $G$ has $|V|$ vertices, then $\hat{T}$ has $n = O(|V|)$ many different entries.
    So any evaluation in time $2^{o\left(\sqrt{n}\right)}$ would decide whether $G$ admits a proper $3$-coloring in time $2^{o\left(\sqrt{|V|}\right)}$, which is a contradiction to Lemma~\ref{lem:gridlikeethhardness}\footnote{or Lemma~\ref{lem:gridlikehardnessregular} for a weaker lower bound of $2^{o\left(\sqrt[4]{n}\right)}$} unless ETH fails.
\end{proof}

\begin{remark}
    All these hardness results also hold if the highest weight vectors are given as a Young tableau $T$ with content $(nd) \times 1$ opposed to $\hat{T}$ with content $n \times d$ by replacing the entries containing $1$ in $\hat{T}$ by $1, \ldots, d$ and $2$ by $d+1, \ldots, 2d$ and so on in a left-to-right, top-to-bottom fashion.
    This corresponds to undoing the projection of $\otimes^n\Sym^d V$ onto $\Sym^n\Sym^d V$.
    In the cases when $\hat{T}$ is semistandard $T$ is standard.
\end{remark}

\bibliography{lit}

\begin{thebibliography}{10}

\bibitem{Abd:02}
Abdelmalek Abdesselam.
\newblock Feynman diagrams in algebraic combinatorics.
\newblock {\em S\'eminaire Lotharingien de Combinatoire [electronic only]},
  49:B49c, 45 p., electronic only--B49c, 45 p., electronic only, 2002.
\newblock URL: \url{http://eudml.org/doc/123420}.

\bibitem{AIR:16}
Abdelmalek Abdesselam, Christian Ikenmeyer, and Gordon Royle.
\newblock 16,051 formulas for ottaviani's invariant of cubic threefolds.
\newblock {\em Journal of Algebra}, 447:649 -- 663, 2016.

\bibitem{BO:11}
Daniel~J. Bates and Luke Oeding.
\newblock Toward a salmon conjecture.
\newblock {\em Experimental Mathematics}, 20(3):358--370, 2011.
\newblock \href {https://doi.org/10.1080/10586458.2011.576539}
  {\path{doi:10.1080/10586458.2011.576539}}.

\bibitem{BB:04}
Christine Bessenrodt and Christiane Behns.
\newblock On the durfee size of kronecker products of characters of the
  symmetric group and its double covers.
\newblock {\em Journal of Algebra}, 280(1):132 -- 144, 2004.

\bibitem{BBP:20}
Christine Bessenrodt, Chris Bowman, and Rowena Paget.
\newblock The classification of multiplicity-free plethysms of {S}chur
  functions.
\newblock arXiv:2001.08763, 2020.

\bibitem{Bini:80}
D.~Bini.
\newblock Relations between exact and approximate bilinear algorithms.
  applications.
\newblock {\em CALCOLO}, 17(1):87--97, Jan 1980.
\newblock \href {https://doi.org/10.1007/BF02575865}
  {\path{doi:10.1007/BF02575865}}.

\bibitem{bcrl:79}
Dario Bini, Milvio Capovani, Francesco Romani, and Grazia Lotti.
\newblock {O}$(n^{2.7799})$ complexity for $n \times n$ approximate matrix
  multiplication.
\newblock {\em Inf. Process. Lett.}, 8(5):234--235, 1979.

\bibitem{BI:17}
Markus Bl{\"a}ser and Christian Ikenmeyer.
\newblock Introduction to geometric complexity theory.
\newblock lecture notes, summer 2017 at Saarland University,
  \url{http://people.mpi-inf.mpg.de/~cikenmey/teaching/summer17/introtogct/gct.pdf},
  version from July 25, 2018, 2018.

\bibitem{BIZ:18}
Karl Bringmann, Christian Ikenmeyer, and Jeroen Zuiddam.
\newblock On algebraic branching programs of small width.
\newblock {\em J. ACM}, 65(5), August 2018.
\newblock \href {https://doi.org/10.1145/3209663} {\path{doi:10.1145/3209663}}.

\bibitem{DBLP:journals/siamcomp/BubleyDGJ99}
Russ Bubley, Martin~E. Dyer, Catherine~S. Greenhill, and Mark Jerrum.
\newblock On approximately counting colorings of small degree graphs.
\newblock {\em {SIAM} J. Comput.}, 29(2):387--400, 1999.
\newblock \href {https://doi.org/10.1137/S0097539798338175}
  {\path{doi:10.1137/S0097539798338175}}.

\bibitem{bue:01}
Peter B{\"u}rgisser.
\newblock The complexity of factors of multivariate polynomials.
\newblock In {\em 42nd {IEEE} {S}ymposium on {F}oundations of {C}omputer
  {S}cience ({L}as {V}egas, {NV}, 2001)}, pages 378--385. IEEE Computer Soc.,
  Los Alamitos, CA, 2001.

\bibitem{bci:10}
Peter B\"urgisser, Matthias Christandl, and Christian Ikenmeyer.
\newblock Even partitions in plethysms.
\newblock {\em Journal of Algebra}, 328(1):322 -- 329, 2011.

\bibitem{BI:11}
Peter B{\"u}rgisser and Christian Ikenmeyer.
\newblock Geometric complexity theory and tensor rank.
\newblock {\em Proceedings 43rd Annual ACM Symposium on Theory of Computing
  2011}, pages 509--518, 2011.

\bibitem{BI:13}
Peter B{\"u}rgisser and Christian Ikenmeyer.
\newblock Explicit lower bounds via geometric complexity theory.
\newblock {\em Proceedings 45th Annual ACM Symposium on Theory of Computing
  2013}, pages 141--150, 2013.

\bibitem{BI:17b}
Peter B{\"u}rgisser and Christian Ikenmeyer.
\newblock Fundamental invariants of orbit closures.
\newblock {\em Journal of Algebra}, 477(Supplement C):390 -- 434, 2017.

\bibitem{BIP:19}
Peter B{\"u}rgisser, Christian Ikenmeyer, and Greta Panova.
\newblock No occurrence obstructions in geometric complexity theory.
\newblock {\em Journal of the American Mathematical Society}, 32:163--193,
  2019.
\newblock A conference version appeared in: Proceedings IEEE 57th Annual
  Symposium on Foundations of Computer Science (FOCS 2016), 386--395.

\bibitem{BLMW:11}
Peter B\"urgisser, J.M. Landsberg, Laurent Manivel, and Jerzy Weyman.
\newblock An overview of mathematical issues arising in the {G}eometric
  complexity theory approach to {VP} v.s. {VNP}.
\newblock {\em SIAM J. Comput.}, 40(4):1179--1209, 2011.

\bibitem{cai2001subexponential}
Liming Cai and David Juedes.
\newblock Subexponential parameterized algorithms collapse the w-hierarchy.
\newblock In {\em International Colloquium on Automata, Languages, and
  Programming}, pages 273--284. Springer, 2001.

\bibitem{carlini2012solution}
Enrico Carlini, Maria~Virginia Catalisano, and Anthony~V Geramita.
\newblock The solution to the {W}aring problem for monomials and the sum of
  coprime monomials.
\newblock {\em Journal of algebra}, 370:5--14, 2012.

\bibitem{cheung2017symmetrizing}
Man-Wai Cheung, Christian Ikenmeyer, and Sevak Mkrtchyan.
\newblock Symmetrizing tableaux and the 5th case of the {F}oulkes conjecture.
\newblock {\em Journal of Symbolic Computation}, 80:833--843, 2017.

\bibitem{CHILO:18}
Luca Chiantini, Jonathan~D. Hauenstein, Christian Ikenmeyer, Joseph~M.
  Landsberg, and Giorgio Ottaviani.
\newblock Polynomials and the exponent of matrix multiplication.
\newblock {\em Bulletin of the London Mathematical Society}, 50(3):369--389,
  2018.
\newblock URL:
  \url{https://londmathsoc.onlinelibrary.wiley.com/doi/abs/10.1112/blms.12147},
  \href
  {http://arxiv.org/abs/https://londmathsoc.onlinelibrary.wiley.com/doi/pdf/10.1112/blms.12147}
  {\path{arXiv:https://londmathsoc.onlinelibrary.wiley.com/doi/pdf/10.1112/blms.12147}},
  \href {https://doi.org/10.1112/blms.12147} {\path{doi:10.1112/blms.12147}}.

\bibitem{CDW:12}
Matthias Christandl, Brent Doran, and Michael Walter.
\newblock Computing multiplicities of lie group representations.
\newblock In {\em Proceedings of the 2012 IEEE 53rd Annual Symposium on
  Foundations of Computer Science}, FOCS ’12, page 639–648, USA, 2012. IEEE
  Computer Society.
\newblock \href {https://doi.org/10.1109/FOCS.2012.43}
  {\path{doi:10.1109/FOCS.2012.43}}.

\bibitem{cygan2015parameterized}
Marek Cygan, Fedor~V Fomin, {\L}ukasz Kowalik, Daniel Lokshtanov, D{\'a}niel
  Marx, Marcin Pilipczuk, Micha{\l} Pilipczuk, and Saket Saurabh.
\newblock {\em Parameterized algorithms}, volume~4.
\newblock Springer, 2015.

\bibitem{DHO:14}
Noah Daleo, Jonathan Hauenstein, and Luke Oeding.
\newblock Computations and equations for segre-grassmann hypersurfaces.
\newblock {\em Portugaliae Mathematica}, 73, 08 2014.
\newblock \href {https://doi.org/10.4171/PM/1977} {\path{doi:10.4171/PM/1977}}.

\bibitem{DBLP:conf/icalp/DorflerIP19}
Julian D{\"{o}}rfler, Christian Ikenmeyer, and Greta Panova.
\newblock On geometric complexity theory: Multiplicity obstructions are
  stronger than occurrence obstructions.
\newblock In {\em 46th International Colloquium on Automata, Languages, and
  Programming, {ICALP} 2019, July 9-12, 2019, Patras, Greece.}, pages
  51:1--51:14, 2019.
\newblock journal version accepted for publication in {S}{I}{A}{M} {J} {A}ppl
  {A}lg {G}eom ({S}{I}{A}{G}{A}).
\newblock \href {https://doi.org/10.4230/LIPIcs.ICALP.2019.51}
  {\path{doi:10.4230/LIPIcs.ICALP.2019.51}}.

\bibitem{Far:16}
Cameron Farnsworth.
\newblock Koszul–young flattenings and symmetric border rank of the
  determinant.
\newblock {\em Journal of Algebra}, 447:664--676, 2016.
\newblock URL:
  \url{https://www.sciencedirect.com/science/article/pii/S0021869315005712},
  \href {https://doi.org/https://doi.org/10.1016/j.jalgebra.2015.11.011}
  {\path{doi:https://doi.org/10.1016/j.jalgebra.2015.11.011}}.

\bibitem{FI:20}
Nick Fischer and Christian Ikenmeyer.
\newblock The computational complexity of plethysm coefficients.
\newblock arXiv:2002.00788, 2020.

\bibitem{forbesWACT16}
Michael Forbes.
\newblock Some concrete questions on the border complexity of polynomials.
\newblock Talk presented at the Workshop on Algebraic Complexity Theory, WACT
  2016, Tel Aviv, 2016.
\newblock video available at
  \url{https://www.cs.tau.ac.il/~shpilka/wact2016/videos/index.php} accessed
  10/17/2019.
\newblock URL:
  \url{https://www.cs.tau.ac.il/~shpilka/wact2016/videos/index.php}.

\bibitem{For:14}
Michael~Andrew Forbes.
\newblock {\em Polynomial Identity Testing of Read-Once Oblivious Algebraic
  Branching Programs}.
\newblock PhD thesis, MIT, 2014.
\newblock URL: \url{https://dspace.mit.edu/handle/1721.1/89843}.

\bibitem{DBLP:journals/tcs/GareyJS76}
M.~R. Garey, David~S. Johnson, and Larry~J. Stockmeyer.
\newblock Some simplified np-complete graph problems.
\newblock {\em Theor. Comput. Sci.}, 1(3):237--267, 1976.
\newblock \href {https://doi.org/10.1016/0304-3975(76)90059-1}
  {\path{doi:10.1016/0304-3975(76)90059-1}}.

\bibitem{GMQ:16}
Joshua~A. Grochow, Ketan~D. Mulmuley, and Youming Qiao.
\newblock {Boundaries of VP and VNP}.
\newblock In Ioannis Chatzigiannakis, Michael Mitzenmacher, Yuval Rabani, and
  Davide Sangiorgi, editors, {\em 43rd International Colloquium on Automata,
  Languages, and Programming (ICALP 2016)}, volume~55 of {\em Leibniz
  International Proceedings in Informatics (LIPIcs)}, pages 34:1--34:14,
  Dagstuhl, Germany, 2016. Schloss Dagstuhl--Leibniz-Zentrum fuer Informatik.
\newblock URL: \url{http://drops.dagstuhl.de/opus/volltexte/2016/6314}, \href
  {https://doi.org/10.4230/LIPIcs.ICALP.2016.34}
  {\path{doi:10.4230/LIPIcs.ICALP.2016.34}}.

\bibitem{Her:54}
Charles Hermite.
\newblock Sur la theorie des fonctions homogenes \`{a} deux
  ind\'{e}termin\'{e}es.
\newblock {\em Cambridge and Dublin Mathematical Journal}, 9:172 -- 217, 1854.

\bibitem{iarrobino1999power}
Anthony Iarrobino and Vassil Kanev.
\newblock {\em Power sums, Gorenstein algebras, and determinantal loci}.
\newblock Springer Science \& Business Media, 1999.

\bibitem{ike:12b}
Christian Ikenmeyer.
\newblock {\em Geometric Complexity Theory, Tensor Rank, and
  {L}ittlewood-{R}ichardson Coefficients}.
\newblock PhD thesis, Institute of Mathematics, University of Paderborn, 2012.
\newblock URL: \url{http://nbn-resolving.de/urn:nbn:de:hbz:466:2-10472}.

\bibitem{Ike:15}
Christian Ikenmeyer.
\newblock The {S}axl conjecture and the dominance order.
\newblock {\em Discrete Mathematics}, 338(11):1970 -- 1975, 2015.

\bibitem{IK:19}
Christian Ikenmeyer and Umangathan Kandasamy.
\newblock Implementing geometric complexity theory: On the separation of orbit
  closures via symmetries.
\newblock arXiv: 1911.03990, 2019.

\bibitem{DBLP:journals/jcss/ImpagliazzoPZ01}
Russell Impagliazzo, Ramamohan Paturi, and Francis Zane.
\newblock Which problems have strongly exponential complexity?
\newblock {\em J. Comput. Syst. Sci.}, 63(4):512--530, 2001.
\newblock \href {https://doi.org/10.1006/jcss.2001.1774}
  {\path{doi:10.1006/jcss.2001.1774}}.

\bibitem{Kumar2018}
Mrinal Kumar.
\newblock On top fan-in vs formal degree for depth-3 arithmetic circuits.
\newblock
  \url{https://eccc.weizmann.ac.il/report/2018/068/revision/1/download}, 2018.

\bibitem{Kum:11}
Shrawan Kumar.
\newblock A study of the representations supported by the orbit closure of the
  determinant.
\newblock {\em Compositio Mathematica}, 151, 09 2011.
\newblock \href {https://doi.org/10.1112/S0010437X14007660}
  {\path{doi:10.1112/S0010437X14007660}}.

\bibitem{LK:15}
Shrawan Kumar and J.M. Landsberg.
\newblock Connections between conjectures of alon-tarsi, hadamard-howe, and
  integrals over the special unitary group.
\newblock {\em Discrete Math.}, 338(7):1232--1238, July 2015.
\newblock \href {https://doi.org/10.1016/j.disc.2015.01.027}
  {\path{doi:10.1016/j.disc.2015.01.027}}.

\bibitem{Lan:15}
J.~M. Landsberg.
\newblock Geometric complexity theory: an introduction for geometers.
\newblock {\em ANNALI DELL'UNIVERSITA' DI FERRARA}, 61(1):65--117, 2015.
\newblock \href {https://doi.org/10.1007/s11565-014-0202-7}
  {\path{doi:10.1007/s11565-014-0202-7}}.

\bibitem{Lan:11}
Joseph Landsberg.
\newblock {\em Tensors: Geometry and Applications}, volume 128 of {\em Graduate
  Studies in Mathematics}.
\newblock American Mathematical Society, Providence, Rhode Island, 2011.

\bibitem{LiE}
M.A.A. {Leeuwen, van}, A.M. Cohen, and B.~Lisser.
\newblock {\em Lie : a package for {L}ie group computations}.
\newblock Centrum voor Wiskunde en Informatica, 1992.

\bibitem{MM:14}
Laurent Manivel and Mateusz Micha{\l}ek.
\newblock Effective constructions in plethysms and weintraub's conjecture.
\newblock {\em Algebras and Representation Theory}, 17(2):433--443, Apr 2014.
\newblock \href {https://doi.org/10.1007/s10468-012-9402-y}
  {\path{doi:10.1007/s10468-012-9402-y}}.

\bibitem{gct1}
K.D. Mulmuley and M.~Sohoni.
\newblock Geometric {C}omplexity {T}heory. {I}. {A}n approach to the {P} vs.\
  {NP} and related problems.
\newblock {\em SIAM J. Comput.}, 31(2):496--526 (electronic), 2001.

\bibitem{gct2}
K.D. Mulmuley and M.~Sohoni.
\newblock Geometric {C}omplexity {T}heory. {II}. {T}owards explicit
  obstructions for embeddings among class varieties.
\newblock {\em SIAM J. Comput.}, 38(3):1175--1206, 2008.

\bibitem{nisan1991lower}
Noam Nisan.
\newblock Lower bounds for non-commutative computation.
\newblock In {\em Proceedings of the 23rd ACM Symposium on Theory of Computing,
  ACM Press}. Citeseer, 1991.

\bibitem{OR:11}
Luke Oeding and Claudiu Raicu.
\newblock Tangential varieties of segre-veronese varieties.
\newblock {\em Collectanea Mathematica}, 65, 11 2011.
\newblock \href {https://doi.org/10.1007/s13348-014-0111-1}
  {\path{doi:10.1007/s13348-014-0111-1}}.

\bibitem{Ott:13}
Giorgio Ottaviani.
\newblock Five lectures on projective invariants, lecture notes for trento
  school, september 2012.
\newblock arXiv:1305.2749, to appear in Rendiconti del Seminario Matematico,
  Torino, 2013.

\bibitem{Rai:13}
Claudiu Raicu.
\newblock $3\times 3$ minors of catalecticants.
\newblock {\em Mathematical Research Letters}, 20, 07 2013.
\newblock \href {https://doi.org/10.4310/MRL.2013.v20.n4.a10}
  {\path{doi:10.4310/MRL.2013.v20.n4.a10}}.

\bibitem{DBLP:journals/jct/RobertsonST94}
Neil Robertson, Paul~D. Seymour, and Robin Thomas.
\newblock Quickly excluding a planar graph.
\newblock {\em J. Comb. Theory, Ser. {B}}, 62(2):323--348, 1994.
\newblock \href {https://doi.org/10.1006/jctb.1994.1073}
  {\path{doi:10.1006/jctb.1994.1073}}.

\bibitem{SS:16}
Steven Sam and Andrew Snowden.
\newblock Proof of stembridge's conjecture on stability of {K}ronecker
  coefficients.
\newblock {\em Journal of Algebraic Combinatorics}, 43:1--10, 2016.

\bibitem{Sax:08}
Nitin Saxena.
\newblock Diagonal circuit identity testing and lower bounds.
\newblock In {\em Automata, Languages and Programming}, pages 60--71, Berlin,
  Heidelberg, 2008. Springer Berlin Heidelberg.

\bibitem{Shi:16}
Yaroslav Shitov.
\newblock How hard is the tensor rank?
\newblock arXiv:1611.01559, 2016.

\bibitem{tamassia1989planar}
Roberto Tamassia and Ioannis~G Tollis.
\newblock Planar grid embedding in linear time.
\newblock {\em IEEE Transactions on circuits and systems}, 36(9):1230--1234,
  1989.

\end{thebibliography}
\end{document}